\documentclass[UKenglish, pdfa]{ourflowspaper}
\usepackage{tikz}
\usepackage{float}
\usepackage{doi}
\usetikzlibrary{arrows.meta,positioning}
\usetikzlibrary{shapes.multipart, calc}
\usetikzlibrary{decorations.markings}
\usetikzlibrary{arrows}
\usetikzlibrary {decorations.pathmorphing}
\usepackage[linesnumbered, vlined, ruled]{algorithm2e}
\usepackage{adjustbox}
\usepackage{subcaption}
\usepackage{enumitem}
\usetikzlibrary{cd}
\usepackage{amssymb}
\usepackage{xparse}
\usepackage{hyperref}

\definecolor{forestgreen}{HTML}{00A64F}
\definecolor{lavender}{HTML}{EC008C}%
\definecolor{purple}{rgb}{0.6, 0.2, 0.8}

\theoremstyle{plain}

\usepackage{listings}

\usepackage{mathtools}

\makeatletter
\renewcommand{\p@subfigure}{} %
\makeatother

\lstdefinestyle{inline}{%
    basicstyle=\ttfamily%
}

\usepackage{stmaryrd}

\usepackage[most]{tcolorbox}

\lstdefinelanguage{isabelle}{
    morekeywords={record,type_synonym,definition,partial,fun,primrec,function,where,lemma,theorem,unfolding,by,shows,assumes,and,datatype,using,abbreviation
,moreover,have,hence,thus,qed,proof,ultimately,show,next,locale,fixes,tailrec,domintros,defines,
for,context,inductive,begin,end,locale,fixes,record,type_synonym,definition,fun,function,primrec,where,lemma,theorem,unfolding,by,
shows,assumes,and,datatype,using,abbreviation,moreover,have,hence,thus,qed,proof,ultimately,show,next, thm,corollary,\ in\ ,if,case,\ of,let,else,then,interpretation,
global_interpretation}
, sensitive=true
    , showstringspaces=true
    , framerule=0pt
    , xleftmargin=2em
    , numbers=left
    , numberstyle=\ttfamily\small
    , firstnumber=1
    , stepnumber=2
    , basicstyle=\ttfamily\small
    , backgroundcolor = \color{white}
    , keywordstyle = {\color{blue}}
    , breaklines=true
    , showspaces=false
    , morecomment=[s]{(*}{*)}
    , commentstyle=\color{gray}
    , morestring=[b]"
    , columns=fullflexible
    , escapeinside={[*@}{@*]}
    , literate={\\<times>}{{$\times$}}{1} {\\<equiv>}{{$\equiv$}}{1} {≡}{{$\equiv$}}{1} {\\<forall>}{{$\forall$}}{1}
    {)))}{{\upshape )))}}{1} {))}{{\upshape ))}}{1} {\ )}{{\upshape )}}{1} {)}{{\upshape )}}{1} {\}}{{\upshape $\rbrace$}}{1}
    {\{}{{\upshape $\lbrace$}}{1}
    {\\<^sup>+}{{$^+$}}{1}  {\\<^sup>-}{{$^-$}}{1} {\\<Delta>}{{$\Delta$}}{1}  {\\<^sub>F}{{$_F$}}{1} {\\<^sub>A}{{$_A$}}{1} {\\<^sub>u}{{$_u$}}{1} {\\<^sub>f}{{$_f$}}{1} {\ }{\ }{1}
    {\\<^bsub>f\\<^esub>}{{$_\mathtt{f}$}}{1}  {\\<Sum>}{{$\sum$}}{1} {\\<uu>}{{$\mathfrak{u}$}}{1}
      {∀}{{$\forall$}}{1} {\\<exists>}{{$\exists$}}{1} {\\<and>}{{$\land$}}{1} {∧}{{$\land$}}{1}
        {\\<in>}{{$\in$}}{1} {\\<Rightarrow>}{{$\Rightarrow$}}{1} {\\<lambda>}{{$\lambda$}}{1} {::}{{$::$}}{1}
        {\\<subset>}{{$\subset$}}{1}  {\\<supset>}{{$\supset$}}{1} {\\<subseteq>}{{$\subseteq$}}{1} {\\<^sub>m}{{$_m$}}{1} {\\<^sub>C}{{$_C$}}{1} {\\<^sub>B}{{$_B$}}{1} {\\<longleftrightarrow>}{{$\longleftrightarrow$}}{3} {⟷}{{$\longleftrightarrow$}}{3}
        {\\<pi>}{{ $\pi$ }}{1} {\\<delta>}{{$\delta$}}{1} {\\<lbrakk>}{{$\llbracket$}}{1} {\\<rbrakk>}{{$\rrbracket$}}{1}
        {\\<Longrightarrow>}{{$\Longrightarrow$}}{3} {⟶}{{$\longrightarrow$}}{3}
        {\\<not>}{{$\lnot$}}{1} {\\<le>}{\upshape {$\le$}}{1} {\\<rightharpoonup>}{{$\rightharpoonup$}}{2}
        {\\<^sub>\\<V>}{{$_{\mathcal V}$}}{1} {\\<lparr>}{{$\llparenthesis$}}{1} {\\<rparr>}{{$\rrparenthesis$}}{1}
        {\\<leftarrow>}{{$\leftarrow$}}{1} {\\<^sub>\\<O>}{{$_{\mathcal O}$}}{1} {\\<^sub>I}{{$_{\texttt{I}}$}}{1}
        {\\<^sub>G}{{$_{\texttt{G}}$}}{2} {\\<phi>}{{$\varphi$}}{1} {\\<Phi>}{{$\Phi$}}{1} {\\<psi>}{{$\psi$}}{1} {\\<Psi>}{{$\Psi$}}{1}
        {\\<^sub>S}{{$_{\texttt S}$}}{1} {\\<inverse>}{{$^{-1}$}}{1} {\\<^sub>O}{{$_O$}}{1} {\\<^bold>\\<And>}{{$\bm\bigwedge$}}{1}
        {\\<^bold>\\<or>}{{$\bm\lor$}}{1} {\\<^sub>G}{{$_{\texttt G}$}}{1} {\\<Pi>}{{$\Pi$}}{1} {\\<^sub>I}{{$_{\texttt I}$}}{1} {\\<noteq>}{{$\neq$}}{1} {≠}{{$\neq$}}{1}
        {\\<bottom>}{{$\bot$}}{1} {\\<^sub>+}{{$_\texttt +$}}{1} {\\<^bold>\\<and>}{{$\bm\land$}}{1} {\\<^bold>\\<not>}{{$\bm\lnot$}}{1}
        {\\<^sub>1}{{$_1$}}{1} {\\<^sub>2}{{$_2$}}{1} {\\<bar>}{{$|$}}{1} {\\<A>}{{$\mathcal A$}}{1} {\\<Turnstile>}{{$\models$}}{2} {\\<^sub>\\<forall>}{{$_\forall$}}{1}
        {\\<^sub>0}{{$_0$}}{1} {\\<tau>}{{$\tau$}}{1}  {\\<^sub>\\<Omega>}{{$_\Omega$}}{1} {\\<^sub>V}{{$_V$}}{1} {\\<^bold>\\<Or>}{{$\bm\bigvee$}}{1}
        {\\<^sub>P}{{$_\texttt P$}}{1} {\\<^sub>X}{{$_\texttt X$}}{1} {\\<longrightarrow>}{{$\longrightarrow$}}{2} {\\<or>}{{$\lor$}}{1} {∨}{{$\lor$}}{1} {\\<^sub>\\<pi>}{{$_\pi$}}{1}
        {\\<^sub>s}{{$_s$}}{1} {\\<^sub>t}{{$_t$}}{1} {\\<^sub>a}{{$_a$}}{1} {\\<^sub>r}{{$_r$}}{1} {\\<^sub>t}{{$_t$}}{1} {\\<^sub>e}{{$_e$}}{1} {\\<^sub>n}{{$_n$}}{1} {\\<^sub>d}{{$_d$}}{1} {\\<^sub>i}{{$_i$}}{1} {\\<^sub>v}{{$_v$}}{1} {\\<^sub>j}{{$_j$}}{1} {\\<^sub>b}{{$_b$}}{1} {\\<inter>}{{$\cap$}}{1} {\\<union>}{{$\cup$}}{1} {\\<Union>}{{ $\bigcup$ }}{1} {\\<^sup>c\\<TTurnstile>\\<^sub>=}{{${}^c\models_=$}}{1}
        {\\<open>}{{<}}{1} {\\<close>}{{>}}{1} {\\<langle>}{{$\langle$}}{1} {\\<rangle>}{{$\rangle$}}{1} {\\<ge>}{{\upshape $\ge$}}{1} {\\<^sup>-\\<^sup>1\\<^sub>C}{{$^{\texttt{-1}}_\texttt{C}$}}{2} {\\<^sup>+\\<^sub>C}{{$^\texttt{+}_\texttt{C}$}}{1} {\\<circ>\\<^sub>C}{{$\circ^\texttt C$}}{1} {\\<top>\\<^sub>C}{{$\top_\texttt{C}$}}{2} {\\<bottom>\\<^sub>C}{{$\bot_\texttt{C}$}}{2} {\\<not>\\<^sub>C}{{$\neg_\texttt{C}$}}{2} {\\<circ>}{{$\circ$}}{1}
        {\\<squnion>\\<^sub>C}{{$\sqcup_\texttt{C}$}}{2} {\\<sqinter>\\<^sub>C}{{$\sqcap_\texttt{C}$}}{2} {\\<exists>\\<^sub>C}{{$\exists_\texttt{C}$}}{2}  {\\<forall>\\<^sub>C}{{$\forall_\texttt{C}$}}{2} {=\\<^sub>C}{{$=_\texttt{C}$}}{2} {\\<sigma>}{{$\sigma$}}{2} {\\<notin>}{{$\notin$}}{1} {\\<oplus>}{{$\oplus$}}{1} {\\<nexists>}{{$\nexists$}}{1} {\\<setminus>}{{$\setminus$}}{1} {_}{{$\text{-}$}}{1} {"}{{}}{1}
    {\ (}{{\upshape \ (}}{1}    {(}{{\upshape (}}{1} {=}{{\upshape =}}{1} {:=}{{\upshape :=}}{1}  {\\if}{{if}}{1} {\\else}{{else}}{1} {\\then}{{then}}{1} {\\case}{{case}}{1} {\\of}{{of}}{1}
        {\\<V>}{{$\mathcal{V}$}}{1} {\\<F>}{{$\mathcal{F}$}}{1} {\\<X>}{{$\mathcal{X}$}}{1} {\\<FF>}{{$\mathfrak{F}$}}{1} {\\<E>}{{$\mathcal{E}$}}{1} {\\<C>}{{$\mathcal{C}$}}{1} {\\<epsilon>}{{\upshape $\epsilon$}}{1}
        {\\<gamma>}{{$\gamma$}}{1} {'}{{'}}{1} {\\let}{{let}}{2} {\\in}{{in}}{2}
        {\\<b>}{{\upshape b}}{1} {\\<u>}{{\upshape u}}{1} {\\<c>}{{\upshape c}}{1} {\\<And>}{{$\bigwedge$}}{1}
        {\\<k>}{{k}}{1} {\\<lceil>}{{$\lceil$}}{1} {\\<rceil>}{{$\rceil$}}{1} {\\<l>}{{$\ell$}}{1}
        {+}{{\upshape +}}{1} {|}{{\upshape |}}{1} {\\<infinity>}{{$\infty$}}{1} {\\<cc>}{{$\mathfrak{c}$}}{1} {\\<EE>}{{$\mathfrak{E}$}}{1} {\\<CC>}{{$\mathfrak{C}$}}{1}
        {\\<space>}{{\ }}{1}
        {global\_ interpretation}{{\color{blue} global-interpretation}}{1}
        {*of*}{{\color{blue}of}}{1}
        {\\<textand>}{and}{1} {\\<textfor>}{for}{1}
}

\lstnewenvironment{plainisabelle}{
\lstset{language=isabelle, numbers=left, columns=fullflexible, basicstyle=\rmfamily\footnotesize\ttfamily, keywordstyle=\bfseries\upshape    }
}
{}

\newtcblisting{plainisabelle2}{
\lstset{language=isabelle, numbers=left, columns=fullflexible, basicstyle=\rmfamily\footnotesize\ttfamily, keywordstyle=\bfseries\upshape    }

}

\lstdefinestyle{isainline}{
  language=isabelle,
  basicstyle=%
    \ttfamily\small
}

\usepackage{hyperref}
\usepackage{graphicx}

{\catcode`\#=12 \gdef\swhhash{#}}

{\catcode`\_=12
 \gdef\swhorigin{https://zenodo.org/records/22761342/files/mabdula/Isabelle-Graph-Library-flows_paper_26.zip}%
 \gdef\swhsnapshot{53f2dc7838307a2db51890b3335dd39671ec58a1}%
 \gdef\swhdirectory{41954671fbaa685fa5106aecefd125994e8d9e8e}%
 \gdef\swhprefix{mabdula-Isabelle-Graph-Library-8575a12}%
}

{\catcode`\_=12 \catcode`\&=12
 \gdef\swhurl#1#2#3{%
   https://archive.softwareheritage.org/browse/directory/\swhdirectory/%
   ?origin_url=\swhorigin%
   &snapshot=\swhsnapshot%
   &path=/\swhprefix/#1%
   \swhhash L#2\ifnum#2=#3 \else -L#3\fi}%
}

\newcommand{\isalogo}{%
  \raisebox{-0.5ex}{\includegraphics[height=2.4ex]{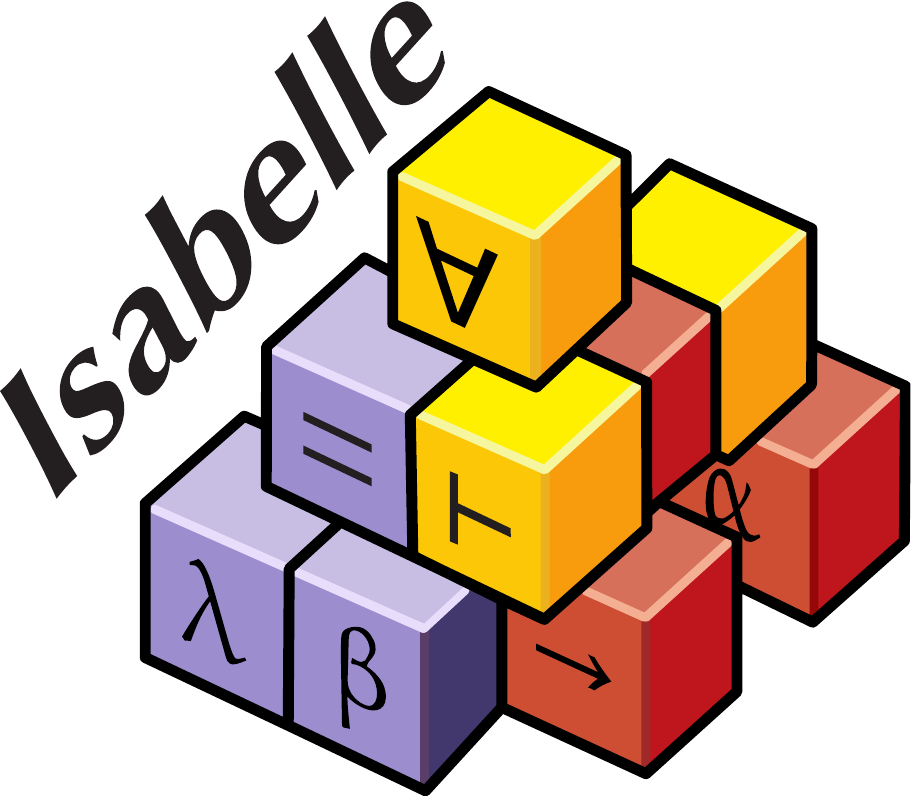}}}

\makeatletter
\let\isa@mark\@empty
\protected\long\def\isa@thmlink#1#2#3{\href{\swhurl{#1}{#2}{#3}}{\isalogo}\enspace}
\newcommand{\isathm}[3]{%
  \xdef\isa@mark{\isa@mark\isa@thmlink{#1}{#2}{#3}}}
\newcommand{\isaitem}[3]{\isa@thmlink{#1}{#2}{#3}}
\newcommand{\isa@usemark}{\isa@mark\global\let\isa@mark\@empty}
\let\isa@head@plain\thmhead@plain
\def\thmhead@plain#1#2#3{\isa@usemark\isa@head@plain{#1}{#2}{#3}}
\let\thmhead\thmhead@plain
\@ifundefined{thmhead@defC}{}{%
  \let\isa@head@defC\thmhead@defC
  \def\thmhead@defC#1#2#3{\isa@usemark\isa@head@defC{#1}{#2}{#3}}}
\@ifundefined{thmhead@thmC}{}{%
  \let\isa@head@thmC\thmhead@thmC
  \def\thmhead@thmC#1#2#3{\isa@usemark\isa@head@thmC{#1}{#2}{#3}}}
\makeatother

\makeatletter
\let\isa@eqmark\@empty
\protected\long\def\isa@eqlink#1#2#3{%
  \href{\swhurl{#1}{#2}{#3}}{\isalogo}\thinspace}
\newcommand{\isaeq}[3]{%
  \xdef\isa@eqmark{\isa@eqmark\isa@eqlink{#1}{#2}{#3}}}
\let\isa@maketag\maketag@@@
\def\maketag@@@#1{%
  \isa@maketag{%
    \ifmeasuring@\else\isa@eqmark\global\let\isa@eqmark\@empty\fi
    #1}}
\makeatother

\newcommand{\isainl}[4]{\href{\swhurl{#1}{#2}{#3}}{\textcolor{black}{\texttt{#4}}}}

\begin{document}
\title{A Formal Analysis of Capacity Scaling Algorithms for Minimum-Cost Flows}

\author[M. Abdulaziz]{Mohammad Abdulaziz\ourflowspaperorcid{0000-0002-8244-518X}}[a]
\author[T. Ammer]{Thomas Ammer\ourflowspaperorcid{0009-0001-5301-4620}}[a]

\address{Department of Informatics, King's College London}{30 Aldwych, Bush House, London, WC2B 4BG, United Kingdom}
\email{mohammad.abdulaziz@kcl.ac.uk, thomas.ammer@kcl.ac.uk}

\thanks{We thank Jens Vygen for providing us with an elegant proof of Lemma~\ref{FBPelim}, leading to a simplification of our own formalisation.}

\begin{abstract}
We present formalisations of the correctness of executable algorithms to solve minimum-cost flow problems in Isabelle/HOL.
Two of the algorithms are based on the technique of scaling, most notably Orlin's algorithm, which has the fastest known running time for solving the problem of minimum-cost flow.
We also include a formalisation of the worst-case running time argument for Orlin's algorithm.
Our verified implementation of this algorithm, which is derived by the technique of stepwise refinement, is fully executable and was integrated into a reusable formal library on graph algorithms.
Because the problems for which Orlin's algorithm works are restricted, we also verified an executable reduction from the general minimum-cost flow problem.
We believe we are the first to formally consider the problem of minimum-cost flows and, more generally, any scaling algorithms.
Our work has also led to a number of mathematical insights and improvements to proofs as well as theorem statements, compared to all existing expositions.
\end{abstract}

\maketitle

\section{Introduction}
\newcommand{\aaugmentedges}{\textit{augment-edges}\xspace}
\newcommand{\sendflowpar}{\textit{send-flow}\xspace}
\newcommand{\sendflow}{\textit{send-flow}\xspace}

Flow networks are some of the most important structures in combinatorial optimisation and computer science.
In addition to many immediate practical applications, flow networks and problems defined on them have numerous connections to other important problems in computer science, most notably, the connection between maximum cardinality (maximum weight) bipartite matching and the problem of maximum (minimum-cost) flow.
Because of this practical and theoretical relevance, network flows have been intensely studied, leading to many important milestone results in computer science, like the Edmonds-Karp algorithm~\cite{EdmondsKarpScalingFlows} for computing the maximum flow between two vertices in a network.
Furthermore, flow algorithms were some of the earliest algorithms to be considered for formal analysis.
The first such effort was in 2005 by Lee in the prover Mizar~\cite{FordFulkersonMizar}, where the Ford-Fulkerson algorithm for maximum flow was verified.
Later on, Lammich and Sefidgar~\cite{LammichFlows} formally analysed the same algorithm and also the Edmonds-Karp algorithm~\cite{EdmondsKarpScalingFlows}, which is one of its polynomial worst-case running time refinements, in Isabelle/HOL.

In this work, we formalise in Isabelle/HOL the correctness of a number of algorithms for the \emph{minimum-cost flow problem}, which is another important computational problem defined on flow networks.
Given a flow network, costs per unit flow associated with every edge, and a desired flow value between a number of sources and a number of targets, a solution to this problem is a flow achieving that value, but for the minimum cost.
This problem can be seen as a generalisation of maximum flow, and thus many problems can be reduced to it, e.g.\ shortest path, maximum flow, plus maximum weight and minimum weight perfect bipartite matching.

More specifically, we formalise (1)\ the problem of minimum-cost flows,
(2)\ the main optimality criterion used to justify most algorithms for minimum-cost flow,
and (3)\ the correctness of three algorithms to compute minimum-cost flows:
(3a)\ successive shortest paths, which has an exponential worst-case running time,
(3b)\ capacity scaling, which has a polynomial worst-case running time,
(3c)\ Orlin's algorithm, which has a strongly polynomial worst-case running time,
and (3d) the core of the running time argument for Orlin's algorithm.
A noteworthy outcome of our work is that we present the first complete combinatorial proof for the correctness of Orlin's algorithm.
For instance, an important property, namely, optimality preservation (Theorem~\ref{thm911}), has a gap in all combinatorial proofs of which we are aware.
We show that gaps existed in most combinatorial proofs (Lemma~\ref{FBPelim}) and that some graphical arguments are complex when treated formally (Theorem~\ref{optcrit}, which needs so-called circulation decomposition (Lemma~\ref{lemma:circdecomp})), proving yet another example of complications uncovered in graphical and geometric arguments when treated in formal settings.
This discrepancy between exact argument and graphical intuition is well known and is widely documented in the literature.
An extreme case in another combinatorial optimisation algorithm formalised in our library is the Ranking algorithm for online bipartite matching, the majority of whose formalisation was spent on an argument originally described by one sketch~\cite{RankingIsabelle}.
Similar discrepancies were documented while formalising geometric arguments too~\cite{hpcpl,poincarebendixson,isabelleGreensThm}.

We use \textit{abstract data types (ADT)} to model mathematical sets and functions in the formalisation of Orlin's algorithm.
This allows us to obtain executable verified code, which still depends on functions selecting paths in the graph with certain properties.
We also give implementations for those functions using the Bellman-Ford algorithm and Depth-First-Search (DFS).
Because Orlin's algorithm only works on specific minimum-cost flow problems, we also formalised an executable reduction of general minimum-cost flow problems to the setting for which the algorithm works.
To our knowledge, Orlin's algorithm is so far the most complex combinatorial optimisation algorithm for which there is verified executable code available.

\subparagraph*{Presentation} We refer to snippets of our formalisation by permalinks.
An Isabelle name occurring in the running text, e.g.\ \isainl{Set_Graphs/Directed_Graphs/Multigraph.thy}{98}{102}{multigraph} , links to the lines containing it.
If a statement in this paper has a formal counterpart, a mark \isalogo\ in front of it links to that counterpart.
Statements without such a mark have no single formal counterpart.
The only code we show are the two abstract data type specifications of Section~\ref{sec:executability_background}.

\subparagraph*{Related Version} This work is an extension of a paper that we submitted to the 15th ITP conference in 2024~\cite{ScalingIsabelle}.
We added sections on the formal running time analysis (Section~\ref{sec:running_time}) and executability of Orlin's algorithm, the verification of an executable implementation of the Bellman-Ford algorithm (Section~\ref{sec:bellman_ford}), as well as the formalisation of a reduction among flow problems (Section~\ref{sec:finite_infinite}).
The whole executability pipeline is new.
We added background on ADTs (Section~\ref{sec:executability_background}) and significantly extended the Sections~\ref{sec:orlins}, \ref{section:loopb} and \ref{section:loopa} (all rooted in the conference paper) with details on how to use ADTs for executability.
Contrary to the conference paper, we now also work with multigraphs instead of simple graphs (see background section for terminology).
We renamed the subprocedure \aaugmentedges in the conference paper to \sendflowpar in the current paper.
The conference paper already closed the gap of Lemma~\ref{FBPelim} for which we now present a simplified version.

There is also an older version of this paper available on arXiv~\cite{abdulaziz2026formalanalysiscapacityscaling} that contains many code listings instead of using permalinks.

\subparagraph*{Availability}
The formalisation as presented here is archived at
\begin{itemize}
\item Zenodo under \texttt{\doi{10.5281/zenodo.22761342}} which is mirrored at
\item Software Heritage under \href{https://archive.softwareheritage.org/browse/directory/41954671fbaa685fa5106aecefd125994e8d9e8e/?origin_url=https://doi.org/10.5281/zenodo.15758197&path=mabdula-Isabelle-Graph-Library-8575a12&release=7&snapshot=fbc2f8cb790b658a32a6fcabfb689348c2dfbb51}{\texttt{swh:1:dir:41954671fbaa685fa5106aecefd125994e8d9e8e}} to which the permalinks point.
\end{itemize}

\tableofcontents

\section{Background and Definitions}
\subsection{Informal Definitions: Multigraphs and Flows}
\label{sec:background_basic}
\newcommand{\mathfst}{\ensuremath{\mathsf{first}}}
\newcommand{\mathsnd}{\ensuremath{\mathsf{second}}}
\newcommand{\E}{\ensuremath{\mathcal{E}}}
\newcommand{\V}{\ensuremath{\mathcal{V}}}
\newcommand{\deltaplus}{\ensuremath{\delta^+}}
\newcommand{\deltaminus}{\ensuremath{\delta^-}}
\newcommand{\Deltaplus}{\ensuremath{\Delta^+}}
\newcommand{\Deltaminus}{\ensuremath{\Delta^-}}
\newcommand{\residualu}{\ensuremath{\mathfrak{u}}}
\newcommand{\residualc}{\ensuremath{\mathfrak{c}}}
We define a \textit{simple directed graph} (henceforth `simple graph') as a set $\E$ of ordered pairs, the so-called \textit{edges}.
Components of a pair $e$ are the \textit{endpoints} of the edge $e$.
For the endpoints $v_1$ and $v_2$, we say that the edge $(v_1, v_2)$ \textit{leads} or \textit{points} from $v_1$ to $v_2$.
The endpoints of all edges in a graph are the set of the graph's vertices $\V$.
In a simple directed graph, edges are unique, i.e.\ there is at most one edge leading from $v_1$ to $v_2$ for all vertices $v_1$ and $v_2$.
For example, $\E = \lbrace (v_1, v_2), (v_2, v_3), (v_1, v_3) \rbrace$ is a simple directed graph over the vertices $\V=\lbrace v_1, v_2, v_3\rbrace$.

In contrast, a \textit{directed multigraph} (henceforth `multigraph') allows for multiple distinct edges pointing from one vertex to another.
It is a set of objects, again called \textit{edges}, together with two functions $\mathfst$ and $\mathsnd$ mapping the edges to their first and second endpoint.
If, for example, $(\cdot, \cdot)$ and $\langle \cdot, \cdot\rangle$ denote two distinct types of ordered pairs, then $\lbrace (v_1,v_2), \langle v_1, v_2\rangle, (v_2,v_3), (v_1,v_3), \langle v_3, v_1\rangle \rbrace$ is a multigraph involving the vertices $v_1$, $v_2$ and $v_3$.
It is a multigraph because $(v_1,v_2)$ and $\langle v_1,v_2 \rangle$ are two distinct edges leading from $v_1$ to $v_2$.
For example, $\mathsf{first}(v_1,v_2) = v_1$, $\mathsf{first}\langle v_1,v_2 \rangle = v_1$ and $\mathsnd(v_2,v_3) = v_3$.
We can interpret simple graphs as multigraphs if we define $\mathsf{first}$ and $\mathsnd$ to give the first and second element of an ordered pair, respectively.

We always assume directed multigraphs (henceforth `directed graph'), i.e.\ edges between vertices are not necessarily unique.
In examples and figures, we often use the setting of a simple graph for the sake of simplicity.

The first vertex of an edge is called the \textit{source} and the second one is the \textit{target} of the edge, respectively.
For a specific vertex $v$, we denote the set of all edges that have $v$ as their first or second endpoint, i.e.\ those that leave or enter the vertex $v$, by $\deltaplus(v)$ or $\deltaminus(v)$, respectively. Analogously to single vertices, we denote the set of leaving and entering edges of a set of vertices $X$ by $\Deltaplus(X)$ and $\Deltaminus(X)$, respectively.

A \textit{path over vertices} from $v_1$ to $v_n$ is a sequence $v_1,...,v_n \subseteq \V$ where there is an edge pointing from $v_i$ to $v_{i+1}$ for all $i < n$. A \textit{path over edges} is a sequence $e_1, ...,e_n \subseteq \E$ where $\mathsnd(e_i) = \mathfst(e_{i+1})$. It leads from $\mathfst(e_1)$ (source) to $\mathsnd(e_n)$ (target). Any vertex-based path implies an edge-based path and the other way round, which is why we often identify the two notions in the informal analysis. A \textit{cycle} is a path where source and target are identical.

A maximal set of vertices $C$, where there is, under omission of edge directions, a path between $x$ and $y$ for all $x, y \in C$, is a \textit{connected component}.
A \textit{representative function} $r : \V \rightarrow \V$ maps all vertices within a component $C$ to the same vertex $ r_C\in C$.
Consequently, we call $r(x)$ the representative of component $C$ for a vertex $x \in C$.

A \textit{flow network} consists of a directed graph over edges $\E$ and vertices $\V$, and a capacity function $u : \E \rightarrow \mathbb{R}^+_0 \cup \lbrace \infty \rbrace$. If $u(e)= \infty$ for all $e \in \E$, the network is \textit{uncapacitated}.
The goal is to find a function,
i.e.\ a \textit{flow} $f : \E \rightarrow \mathbb{R}^+_0$ satisfying $f(e) \leq u(e)$ for any edge $e$.
An edge is \textit{saturated} if its flow $f(e)$ equals its capacity $u(e)$, otherwise the edge is \textit{unsaturated}.
The \textit{excess} of a flow $f$ at the vertex $v$, $\text{ex}_f(v)$,
is the difference between ingoing and outgoing flow $\text{ex}_f(v) \overset{\text{def}}{=} \sum \limits_{e \in \deltaminus(v)} f(e) - \sum \limits_{e \in \deltaplus(v)} f(e)$.

We call an ordered bipartition $(X, \V \setminus X)$ of the graph's vertices a \textit{cut}.
The (reverse) capacity of a cut $X$ is the accumulated edge capacity of all (ingoing) outgoing edges
$\text{cap}(X) \overset{\text{def}}{=} \sum \limits_{e \in \Deltaplus(X)} u(e)$ $(\text{acap}(X) \overset{\text{def}}{=} \sum \limits_{e \in \Deltaminus(X)} u(e))$. Edges going between $X$ and $\V\setminus X$ are edges \textit{crossing} the cut. If all edges going from $X$ to $\V\setminus X$ are saturated for some flow, we say that the cut is saturated.

A \textit{minimum-cost flow problem} consists of two further ingredients, apart from the edges and their capacities.
We introduce \textit{balances} $b : \V \rightarrow \mathbb{R}$ denoting the amount of flow that should be caught or emitted at every vertex.
We call a flow satisfying the balance and capacity constraints \textit{valid} or \textit{feasible}.
In addition, there is a function $c : \E \rightarrow \mathbb{R}$ telling us about the \textit{costs} of sending one unit of flow through an edge. A flow's $f$ \textit{total costs} $c(f)$ are $c(f) = \sum \limits_{e \in \E} f(e) \cdot c(e)$.
The set of \textit{feasible flows} is $\lbrace f \, | \, \forall v \in \V.\, -ex\, _f(v)   = b ( v)  \wedge \forall e \in \E . \,f(e) \leq u(e) \rbrace$.
We aim to find a \textit{minimum-cost flow (min-cost flow)} which is a feasible flow of least total costs.
We call a network without cycles of negative total costs \textit{weight-conservative}.

\begin{exa}\label{example:transport}
An intuitive application of minimum-cost flows would be the transportation of products:
at certain locations (vertices in the graph), e.g.\ production and retail sites, there is a supply of and demand for a certain product $P$.
We need to match demands and supplies to one another.
Therefore, we must ship $P$ along some ways of transport (edges in the graph).
Each edge $e$ comes with a transportation capacity $u(e)$ and costs $c(e)$ for each unit of $P$ that we send along $e$.
We can encode demands and supplies into the balances $b$, where $b(v) > 0$ is a supply and $b(v) < 0$ is a demand at the location $v$.
A valid transportation plan means that (1) we remove all supply from all production facilities,
(2) the demand at each retail site is exactly satisfied, and
(3) the amount of $P$ sent along $e$ remains below $u(e)$ for each transportation link $e$.
When modelling this as a flow problem, a minimum-cost flow $f$ would be a valid transportation plan of minimum costs.
 \end{exa}

Given a flow $f$, we define the \textit{residual network}:
for any edge $e \in \E$ pointing from $x$ to $y$, we have two \textit{residual edges}, namely,
 the \textit{forward} $F \, e$ and \textit{backward} $B \, e$ edge pointing from $x$ to $y$ and from $y$ to $x$, respectively.
These form a pair of \textit{reverse} edges.
We write the reverse of a residual edge $e$ as $\overset{\leftarrow}{e}$.
We define the \textit{residual cost} $\residualc$ of a residual edge as $\residualc(F\,e) = c(e)$ and $\residualc(B\,e) = -c(e)$, respectively.
For a flow $f$, we define the \textit{residual capacities} $\residualu_f$.
On forward edges, this is the difference between the actual capacity and the flow that we currently send through this edge: $\residualu_f(F \,e) = u(e) - f(e)$.
The capacity of a backward edge equals the flow assigned to the original edge:
$\residualu_f(B\, e) = f(e)$.

We define the residual capacity $\residualu_f(p)$ of a path $p$ as $\residualu_f(p) = \min \lbrace \residualu_f(e)\, . \; e \in p \rbrace$.
We obtain the residual costs $\residualc(p)$ by accumulating residual costs for the edges contained in $p$.
We define the reverse $\overset{_{\longleftarrow}}{p}$ of a path $p$ over residual edges as the reversal of each residual edge in $p$ in reversed order.
For example, the reverse of $[F\,e_1, B\, e_2, F\, e_3]$ is $[B \,e_3, F\, e_2, B \,e_1]$. From the definitions it follows that $\residualc(\overset{_{\longleftarrow}}{p}) = - \residualc(p)$.

The residual network is a multigraph itself regardless of whether the actual flow network is a simple graph.
The residual network for a simple graph has two copies of the same edge, e.g.\ $F\, (u, v)$ and $B \, (u, v)$ (see the residual network in Figure~\ref{graphBo}).
Intuitively, a forward edge of the residual network indicates how much more could be added to the flow and a backward edge indicates how much could be removed from the flow.
\renewcommand{\frame}{}
\begin{figure}[]
\frame{
\begin{subfigure}{0.485\textwidth}
{\centering
\frame{
   \trimbox{0mm {0.5\baselineskip} 0mm 0mm}{
\begin{tikzpicture}[
      mycircle/.style={
         circle,
         draw=black,
         fill=gray,
         fill opacity = 0.3,
         text opacity=1,
         inner sep=0pt,
         minimum size=12pt,
         font=\normalsize},
      myarrow/.style={-Stealth},
      node distance=1cm and 4cm
      ]
      \node[ mycircle] (c1) {$u$} ;
      \node[ mycircle,right=of c1] (c2) {$v$};
     \draw [myarrow] (c1) to [out=10,in=170] node[sloped,font=\footnotesize, above]
      {$5$, \color{forestgreen} $8$\color{black}\text{,} \color{red} $2$} (c2);
     \draw [myarrow] (c2) to [out=250,in=290] node[sloped,font=\footnotesize, above]
      {$4$, \color{forestgreen} $6$\color{black}\text{,} \color{red} $3$} (c1);
      \end{tikzpicture}
      }
      }
      \par}
      \caption{\label{graphAo} The original flow}
\end{subfigure}
}
\nolinebreak
\frame{
\begin{subfigure}{0.485\textwidth}
{\centering
   \frame{
      \trimbox{0mm {0.25\baselineskip} 0mm 0mm}{
\begin{tikzpicture}[
      mycircle/.style={
         circle,
         draw=black,
         fill=gray,
         fill opacity = 0.3,
         text opacity=1,
         inner sep=0pt,
         minimum size=12pt,
         font=\normalsize},
      myarrow/.style={-Stealth},
      node distance=1cm and 4cm
      ]
      \node[ mycircle] (c3) {$u$} ;
      \node[ mycircle,right=of c3] (c4) {$v$};
     \draw [myarrow, blue] (c3) to [out=10,in=170] node[sloped,font=\footnotesize, above=-0.75mm]
      {\color{forestgreen} $3 (=8-5)$\color{black}\text{,} \color{red} $2$} (c4);
    \draw [myarrow, purple, dashed] (c4) to [out=190,in=350] node[sloped,font=\footnotesize, below=-0.75mm]
      {\color{forestgreen} $5$\color{black}\text{,} \color{red} $-2$} (c3);
    \draw [myarrow, purple, dashed] (c3) to [out=35,in=145] node[sloped,font=\footnotesize, above=-1mm]
      {\color{forestgreen} $4$\color{black}\text{,} \color{red} $-3$} (c4);
    \draw [myarrow, blue] (c4) to [out=215,in=325] node[sloped,font=\footnotesize, below]
      {\color{forestgreen} $2(=6-4)$\color{black}\text{,} \color{red} $3$} (c3);
      \end{tikzpicture}
      }
      }
      \par}
     \caption{\label{graphBo} Original Residual Graph}
      \end{subfigure}
}

\caption{\label{fig2o} Flows, (residual) capacities and (residual) costs are black, green and red, respectively. Forward edges are blue, backward edges purple.}
\end{figure}
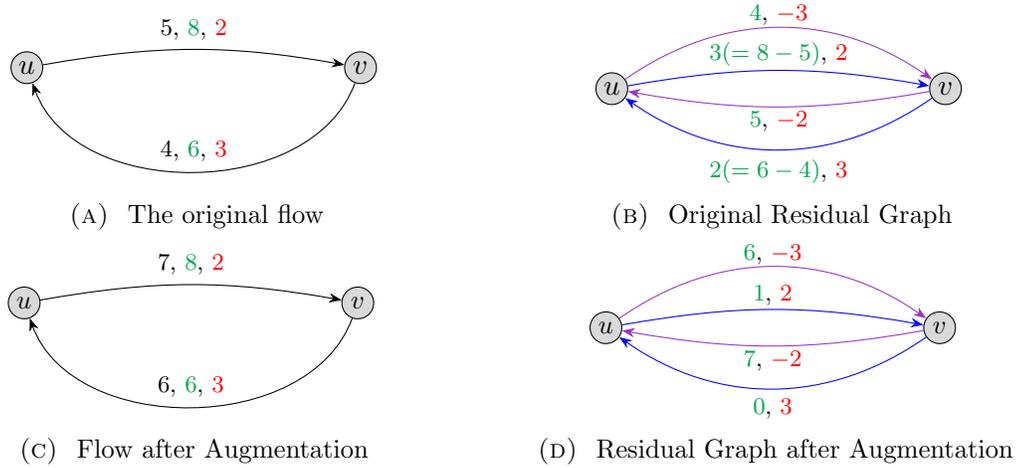

\begin{exa}
Figure~\ref{graphAo} shows a flow network, where every edge is labelled by a flow, capacity, and a cost per unit of flow.
The colouring convention from the caption applies.
Figure~\ref{graphBo} shows the residual network for the network in Figure~\ref{graphAo}.
The residual network has, for every flow network edge $(x,y)$, two edges: one is the forward edge $F (x,y)$,
a copy of the original edge, and the second is the backward edge $B (x, y)$, going in the opposite direction.
Each forward edge of the residual network is labelled by the residual capacity, indicating how much more flow can still go through the edge, without violating the edge's capacity constraint.
Each backward edge is labelled by a capacity: intuitively this capacity represents how much flow we can remove from the edge and, accordingly, its value is the same as the flow going through the original edge.
Each backward edge is labelled by a negative cost indicating that removing from the flow on that edge saves cost.
\end{exa}
\begin{exa}
In Figure~\ref{graphBo}, the sequences of residual edges $[F\;(u, v),\; B\;(u, v)]$, $[F\;(u,v), \\ F\, (v, u), B\,(v,u), B\,(u, v)]$ and $[B\, (u, v), B\, (v,u)]$ are examples of edge-distinct paths in the residual graph.
Since the first and last vertices are identical ($u$, $u$ and $v$, respectively), these even form cycles.
The residual path capacities are $3$, $2$, $4$, and the costs are $0$, $0$ and $-5$, respectively.
\end{exa}

\subsection{Locales, Multigraphs and Minimum-Cost Flows in Isabelle/HOL}
\label{sec:isabelle_intro}
\newcommand{\multigraphspec}{\isainl{Set_Graphs/Directed_Graphs/Multigraph.thy}{5}{9}{multigraph-spec}\xspace}
\newcommand{\multigraph}{\isainl{Set_Graphs/Directed_Graphs/Multigraph.thy}{98}{102}{multigraph}\xspace}
\newcommand{\flownetwork}{\isainl{Flow_Theory/Residual.thy}{35}{38}{flow-network}\xspace}
\newcommand{\flownetworkspec}{\isainl{Flow_Theory/Residual.thy}{31}{33}{flow-network-spec}\xspace}
\newcommand{\costflowspec}{\isainl{Flow_Theory/Residual.thy}{40}{41}{cost-flow-spec}\xspace}
\newcommand{\edgetype}{\texttt{'edge}\xspace}
\newcommand{\fst}{\texttt{fst}\xspace}
\newcommand{\snd}{\texttt{snd}\xspace}
\newcommand{\fstv}{\texttt{fstv}\xspace}
\newcommand{\sndv}{\texttt{sndv}\xspace}
\newcommand{\createedge}{\texttt{create-edge}\xspace}
\newcommand{\isOpt}{\isainl{Flow_Theory/Cost_Optimality.thy}{155}{155}{is-Opt}\xspace}
\newcommand{\tttu}{\texttt{u}\xspace}
\newcommand{\tttf}{\texttt{f}\xspace}
\newcommand{\tttb}{\texttt{b}\xspace}
\newcommand{\sharevertex}{\texttt{share-vertex}\xspace}
\newcommand{\eone}{\texttt{e1}\xspace}
\newcommand{\etwo}{\texttt{e2}\xspace}
\newcommand{\tttdeltaplus}{\isainl{Set_Graphs/Directed_Graphs/Multigraph.thy}{31}{32}{delta-plus}\xspace}
\newcommand{\tttdeltaminus}{\isainl{Set_Graphs/Directed_Graphs/Multigraph.thy}{34}{35}{delta-minus}\xspace}
\newcommand{\tttDeltaplus}{\isainl{Set_Graphs/Directed_Graphs/Multigraph.thy}{39}{40}{Delta-plus}\xspace}
\newcommand{\tttDeltaminus}{\isainl{Set_Graphs/Directed_Graphs/Multigraph.thy}{42}{43}{Delta-minus}\xspace}
\newcommand{\Redge}{\isainl{Flow_Theory/Residual.thy}{341}{341}{Redge}\xspace}
\newcommand{\tttex}{\isainl{Flow_Theory/Residual.thy}{144}{145}{ex}\xspace}
\newcommand{\isuflow}{\isainl{Flow_Theory/Residual.thy}{158}{159}{isuflow}\xspace}
\newcommand{\isbflow}{\isainl{Flow_Theory/Residual.thy}{180}{181}{isbflow}\xspace}
\newcommand{\rcap}{\isainl{Flow_Theory/Residual.thy}{371}{373}{rcap}\xspace}
\newcommand{\residualflow}{\isainl{Flow_Theory/Cost_Optimality.thy}{408}{410}{residual-flow}\xspace}
In Isabelle/HOL, \textit{locales}~\cite{LocalesBallarin} are named contexts that allow us to fix mathematical entities (`locale constants') and to specify/assume their properties.
Inside the locale, we can use the constants and assumptions on them for definitions and proofs, respectively.
We use locales to formalise multigraphs:
the \textit{specification locale} \multigraphspec fixes a set $\E$ of edges of type \edgetype, together with functions \fst and \snd from the edge to the vertex type.
These are intended to return the first and second vertex of an edge, respectively.
\fst and \snd correspond to the informal functions $\mathfst$ and $\mathsnd$ that are associated with every multigraph.
In the case of a simple graph, one would use the functions \fst and \snd that are predefined for ordered pairs in Isabelle/HOL.
Furthermore, for every pair over the vertex type there should also be a possible edge of type \edgetype with the respective vertices as endpoints.
We express this by the function \createedge.
This edge does not have to be in $\E$, it only ensures the possibility of a connection between vertices, similar to the fact that we can construct an ordered pair for arbitrary components $x$ and $y$.

The specification locale only fixes the constants while being ignorant of their behaviour.
We have a second locale, the \textit{proof locale} \multigraph, that extends the specification locale by making the  assumptions that were informally stated in the previous paragraph.
This distinction between a specification locale (fixing constants and types) and a proof locale (assuming the constants' properties) is standard in our formalisation.
The rationale behind this split follows from the mechanisms of \textit{locale instantiation}.

We can later instantiate the constants fixed by a locale, e.g.\ in the case of \multigraph by giving a set $\E$ and implementations for \fst, \snd and \createedge.
We must then prove all assumptions, e.g.\ an instantiation of \multigraph would require us to show finiteness of the $\E$ that is used for the instantiation.
An instantiation turns every definition and theorem inside a locale into analogous definitions and theorems outside the locale.
The latter involve the entities used to instantiate the locale constants.
\begin{exa}
By instantiating, a definition $\sharevertex \; \eone \; \etwo = \lbrace \fst \;\eone, \snd \;\eone \rbrace \cap \lbrace \fst \;\etwo, \snd \;\etwo \rbrace \neq \lbrace \rbrace$ inside \multigraphspec to check whether two edges share a vertex,
becomes a constant $\sharevertex$ outside the locale that again takes two edges, but whose type and implementation of \fst and \snd is fixed according to the concrete instantiation, for instance, to the product type and the predefined \fst and \snd.
\end{exa}

Isabelle/HOL also allows for the generation of executable code from definitions given in a locale, provided that the definitions are executable themselves.
This requires a so-called \textit{global interpretation}, which is a certain type of locale instantiation.
Global interpretations must be outside of any context and can be generic w.r.t.\ all locale constants on which there are no non-trivial assumptions.
For example, a global interpretation of \multigraph only works for specific sets of edges, since some sets are finite and others are not.
In contrast, one can interpret \multigraphspec globally for all possible values of the locale constants because there are no assumptions.

As noted earlier, we distinguish between a specification locale, which fixes constants and types, and a proof locale, which makes assumptions about the constants' properties.
We put every function definition that is meant to be executed into the specification locale.
On the other hand, we place any reasoning regarding the correctness of an algorithm in the proof locale.
By this, a single global interpretation will be sufficient to obtain code that would work for several input graphs.
This code has the locale constants, e.g.\ \fst or \snd as function arguments.
Correctness of the code for a specific input depends on a few simple properties, namely, that the input indeed satisfies the assumptions made by the proof locale.

Within \multigraphspec, we define basic concepts like \tttdeltaplus, \tttdeltaminus, \tttDeltaplus and \tttDeltaminus, i.e.\ the formal counterparts of $\deltaplus$, $\deltaminus$, $\Deltaplus$ and $\Deltaminus$.

We use a functional data type \Redge to model residual edges.
As an extension of \multigraphspec, we have a specification locale \flownetworkspec for flow networks with capacities (edges mapped to reals with infinity) but without costs.
Within that locale, we formalise the concepts of flow excess (\tttex), compliance with balance and capacity constraints (\isbflow and \isuflow), as well as, residual capacities (\rcap).
The proof locale \flownetwork assumes that capacities are non-negative.
The locale \costflowspec extends \flownetworkspec by assuming a cost function from edges to reals.
In that locale, we define residual costs, and the notion of minimum-cost flows.
\isOpt \tttb \tttf would mean that \tttf is a minimum-cost flow for balances \tttb and the edge capacities \tttu fixed by the locale.

There are two reasons why it is beneficial to consider flows in multigraphs as formalised by \multigraphspec instead of simple graphs formalised as a set of ordered pairs.
Firstly, when modelling a real-world problem, e.g.\ transportation of goods (Example~\ref{example:transport}), there can be multiple direct links between two sites $v_1$ and $v_2$, each of which has different per-unit costs.
It is desirable to have algorithms that directly deal with that.
The other reason is that we can easily interpret the residual graph as another flow network.
The latter uses, however, the residual edges, residual capacities and residual costs, e.g. \residualflow.
From time to time, flows in the residual graph are considered in informal proofs.
Reflecting this in the formalisation is straightforward when using \multigraphspec.

\subsection{Obtaining Executable Algorithms}
\label{sec:executability_background}
\newcommand{\tttif}{\texttt{if}}
Functions and sets as they are available in Isabelle/HOL are already sufficient to formalise supporting theory and to describe the algorithms mathematically.
Formalised in that way, algorithms are often not executable and, even if they are, the running time might be inflated.
For example, in a mathematical description, one might perform updates to a function, that maps indices to values to model an array, by an \tttif -expression.
Due to the evaluation strategy, these build up to nested \tttif s resulting  in a running time that linearly depends on the number of updates.
For Orlin's algorithm, which is the most efficient algorithm presented in this paper, we aim for executable code of reasonable efficiency.

\begin{rem}
The running time analysis in the literature is done w.r.t.\ an imperative model of computation.
Thus, we aim to bring the formalisation of Orlin's algorithm closer to the semantics of an imperative language.
There is also some ongoing effort to implement an automated translation of Isabelle/HOL functions to IMP-~\cite{cookLevinFramework}, a deliberately simple but Turing-complete imperative language used to study computability and complexity theory.
In the future, we might attempt this for Orlin's algorithm.
Full executability is a prerequisite, however.
\end{rem}

\paragraph{Definition of Abstract Data Types}
One approach to obtain executable Isabelle/HOL functions with reasonable running time is so-called \textit{abstract data types (ADT)}.
These have been studied first by Wirth~\cite{refinementWirth} and Hoare~\cite{hoareRefinement}, and are an example of Wirth's concept of stepwise refinement where a program instruction is decomposed into more detailed instructions.
For ADTs, the scope of application is much wider than our own and immediately related work, as e.g.\ also Java, C or Haskell libraries use them to specify interfaces for sets or queues.

In the context of ADTs, for example, we can describe the behaviour of $\cup_{impl}$ (implementation of set union {$\cup$}) by $\alpha (A \cup_{impl} B) = \alpha (A) \cup \alpha (B)$ where $\alpha$ is an abstraction function denoting the semantics of a set implementation, i.e.\ which implementations represent which mathematical sets.
Mathematically speaking, $\alpha$ has to be a homomorphism from the implementations to mathematical sets w.r.t.\ the set operations.
A data structure invariant, which is a well-formedness property that the operations have to preserve, is usually a precondition.

\paragraph{Abstract Data Types in Isabelle/HOL}
\newcommand{\Some}{\texttt{Some}\xspace}
\newcommand{\None}{\texttt{None}\xspace}
\newcommand{\tttx}{\texttt{x}\xspace}
\newcommand{\tttm}{\texttt{m}\xspace}
\newcommand{\lookup}{\texttt{lookup}\xspace}
\newcommand{\invar}{\texttt{invar}\xspace}
The built-in standard library of Isabelle/HOL contains specifications of several ADTs by Nipkow~\cite{ADTIsabelle}, e.g.\ maps and sets (see Listings~\ref{isabelle:AD_sets} and~\ref{isabelle:AD_maps}).
For example, a map modelling a (partial) function $f$ would return \Some \tttx if $f(x)$ is defined and \None otherwise.
For a map \tttm, we call the set of those \texttt{x} where \lookup \tttm \tttx $\neq$ \None the \textit{domain} of \tttm.
Like flow networks or multigraphs, we formalise ADTs by means of locales.
Within these and derived locales, we can use the properties of the assumed set or map implementations, provided that the well-formedness invariant \invar holds.
Later, we can instantiate the locales to obtain executable code.

For example, Abdulaziz~\cite{graphAlgosBook} used these locales to model graphs as adjacency maps, which are maps mapping vertices to set implementations of their neighbourhoods (sets of the adjacent vertices), and to verify and obtain executable implementations of simple graph algorithms like Breadth-First-Search (BFS) or Depth-First-Search (DFS).
Notably, he can mostly hide the fact that the algorithms work on assumed implementations instead of mathematical functions and sets by automation that applies the data structure semantics (Listings~\ref{isabelle:AD_sets} and~\ref{isabelle:AD_maps}) in the background.
Another example is the formalisation of algorithms for matroids and greedoids~\cite{MatroidIsabelle}, where ADTs allow for deriving executable implementations of Kruskal's and Prim's algorithms, and maximum cardinality bipartite matching.
We use the ADT approach in our work to obtain executable code for Orlin's algorithm from the formalisation.

\begin{figure}
\begin{lstlisting}[
 language=isabelle,
 caption={Abstract Data Type Specification for Sets~\cite{ADTIsabelle}},
 label={isabelle:AD_sets},
 captionpos=b,
 numbers=none,
 xleftmargin=0cm,
 columns= fullflexible
 ]
locale Set =
  fixes empty :: "'s"            and insert :: "'a \<Rightarrow> 's \<Rightarrow> 's"
  and isin :: "'s \<Rightarrow> 'a \<Rightarrow> bool" and set :: "'s \<Rightarrow> 'a set"       and invar :: "'s \<Rightarrow> bool"
  assumes set_empty: "set empty = {}"
  assumes set_isin: "invar s \<Longrightarrow> isin s x = (x \<in> set s)"
  assumes set_insert: "invar s \<Longrightarrow> set (insert x s) = set s \<union> {x}"
  assumes invar_empty: "invar empty"
  assumes invar_insert: "invar s \<Longrightarrow> invar (insert x s)"
\end{lstlisting}
\end{figure}

\begin{figure}
\begin{lstlisting}[
 language=isabelle,
 caption={Abstract Data Type Specification for Maps~\cite{ADTIsabelle}},
 label={isabelle:AD_maps},
 captionpos=b,
 numbers=none,
 xleftmargin=0cm,
 columns= fullflexible
 ]
locale Map =
  fixes empty :: "'m"                   and update :: "'a \<Rightarrow> 'b \<Rightarrow> 'm \<Rightarrow> 'm"
  and lookup :: "'m \<Rightarrow> 'a \<Rightarrow> 'b option" and invar :: "'m \<Rightarrow> bool"
  assumes map_empty: "lookup empty = (\<lambda>_. None)"
  and map_update: "invar m \<Longrightarrow> lookup (update a b m) = (lookup m)(a := Some b)"
  and invar_empty: "invar empty"
  and invar_update: "invar m \<Longrightarrow> invar (update a b m)"
\end{lstlisting}
\end{figure}

\subsection{Fundamental Results in Flow Theory}
\newcommand{\thmflowvalue}{\isainl{Flow_Theory/Residual.thy}{854}{857}{flow-value}\xspace}
\newcommand{\thmbalancecut}{\isainl{Flow_Theory/Residual.thy}{1211}{1213}{flow-less-cut-capacities}\xspace}
\newcommand{\thmrescutsat}{\isainl{Flow_Theory/Residual.thy}{1216}{1219}{rescut-ingoing-zero-ingoing-saturated}\xspace}
\newcommand{\thmcircdecomp}{\isainl{Flow_Theory/Decomposition.thy}{471}{480}{flowcycle-decomposition}\xspace}
Proofs about flows often require some kind of decomposition of the network or the flow itself.
Cuts allow us to view two parts $X$ and $\V\setminus X$ of the network as single vertices connected by a pair of edges in opposite directions.
For a flow $f$ valid for balances $b$, the sum of $b$ over $X$ is the difference between the flow going from $X$ to $\V\setminus X$ and the flow going from $\V\setminus X$ to $X$.
We make this precise in a version of the \textit{Flow-Value Lemma}~\cite{FordFulkersonFlowsNetworks}, a fundamental result in flow theory:

\isathm{Flow_Theory/Residual.thy}{854}{857}%
\begin{lem}[Flow Value Lemma\label{lemma:flowvalue}]
Let $f$ be a flow that is valid for the balances $b$ and let $(X, \V\setminus X)$ be a cut. Then
\[
\sum \limits_{v \in X} b(v) = \sum \limits_{e \in \Deltaplus(X)} f(e) - \sum \limits_{e \in \Deltaminus(X)} f(e)
\]
\end{lem}
We conduct the proof by replacing $b(v)$ with $-ex_f(v)$ and the observation that the flow $f(e)$ along edges $e$ not crossing the cut is cancelled since it appears both as $f(e)$ and $-f(e)$ in the sum.

Cuts are an important tool for finding a bound for the amount of flow that we send between different parts of the network.
For example, it follows from the definition of cuts and flow validity that we have $0 \leq \sum \limits_{e \in \Deltaplus(X)} f(e) \leq \text{cap}(X)$ and  $0 \leq \sum \limits_{e \in \Deltaminus(X)} f(e) \leq \text{acap}(X)$ for any valid flow $f$ and cut $X$.
This implies together with the Flow-Value Lemma:
\isathm{Flow_Theory/Residual.thy}{1211}{1213}%
\begin{cor}
\label{corobalance}
\[
- \text{\textnormal{acap}}(X) \leq  \sum \limits_{v \in X} b(v) \leq \text{\textnormal{cap}}(X)
\]
for any flow $f$ valid for balances $b$ and cut $X$.
\end{cor}
The intuition behind this statement is that the sum of the balances in a cut is the amount of flow to be sent or removed from the cut. The cut capacities are a natural bound for the movement of flow.

For a vertex $v$ and a valid flow $f$, we call those vertices that are reachable from $v$ by a path over residual edges with positive residual capacity a \textit{residual cut} $Rescut_f(v)$. A simple and well-known fact follows from the definition of the residual cut:
\isathm{Flow_Theory/Residual.thy}{1216}{1219}%
\begin{lem}\label{lemma:residualsat}
For a residual cut $X$ around a vertex $v$, every edge in $\Deltaplus(X)$ is saturated, i.e.\ the total flow leaving $X$ is $\textnormal{cap}(X)$ and every edge in $\Deltaminus(X)$ has zero flow.
\end{lem}

If $f(e) < u(e)$ for $e \in \Deltaplus(X)$ or if $f(e) > 0$ for some $e\in\Deltaminus(X)$, there would be a residual edge with positive residual capacity leaving $X$, which contradicts the definition of the residual cut.

Apart from cuts, \textit{flow decomposition}~\cite{GallaiDecomposition,FordFulkersonFlowsNetworks} is another fundamental technique when reasoning about flows.
There are several decomposition results, one of which is the decomposition of so-called circulations, which is also relevant to minimum-cost flows.
We call a flow $g$ that is valid for $b(v) =0$, i.e.\ where $ex_g(v) = 0$ for all $v \in \V$, a \textit{circulation}.

\begin{figure}
\begin{center}
\begin{tikzpicture}[
      mycircle/.style={
         circle,
         draw=black,
         fill=gray,
         fill opacity = 0.3,
         text opacity=1,
         inner sep=0pt,
         minimum size=10pt,
         font=\small},
      myarrow/.style={-Stealth},
      node distance=0.7cm and 1.4cm
      ]
      \node[ mycircle] (c1)  {};
      \node[ mycircle,below right=of c1] (c2) {};
      \node[ mycircle,right=of c2] (c3) {};
      \node[ mycircle,above right=of c1] (c4) {};
      \node[ mycircle,right=of c4] (c5) {};
      \node[ mycircle,below right=of c5] (c6) {};

    \foreach \i/\j/\txt/\p in {%
      c2/c1/$2$/below,
      c1/c4/$2$/above,
      c3/c2/$4$/below,
      c3/c6/$2$/below,
      c5/c4/$3$/above,
      c6/c5/$2$/above,
      c2/c5/$7$/below,
      c5/c3/$6$/below,
      c4/c2/$5$/above}
      \draw [myarrow] (\i) -- node[sloped,font=\scriptsize,\p] {\txt} (\j);

\end{tikzpicture}
\nolinebreak
\hspace*{1cm}
\begin{tikzpicture}[
      mycircle/.style={
         circle,
         draw=black,
         fill=gray,
         fill opacity = 0.3,
         text opacity=1,
         inner sep=0pt,
         minimum size=10pt,
         font=\small},
      myarrow/.style={-Stealth},
      node distance=0.7cm and 1.4cm
      ]
      \node[ mycircle] (c1)  {};
      \node[ mycircle,below right=of c1] (c2) {};
      \node[ mycircle,right=of c2] (c3) {};
      \node[ mycircle,above right=of c1] (c4) {};
      \node[ mycircle,right=of c4] (c5) {};
      \node[ mycircle,below right=of c5] (c6) {};

\draw[myarrow] (c2) -- node[sloped, color = green, font = \scriptsize, below] {2} (c1);
\draw[myarrow] (c1) -- node[sloped, color = green, font = \scriptsize, above] {2} (c4);
\draw[myarrow] (c4.235) -- node[sloped, color = green, font = \scriptsize, below] {2} (c2.125);
\draw[myarrow] (c4.285) -- node[sloped, color = red, font = \scriptsize, above] {3} (c2.75);
\draw[myarrow] (c5) -- node[sloped, color = red, font = \scriptsize, above] {3} (c4);

\draw[myarrow] (c5.255) -- node[sloped, color = blue, font = \scriptsize, below] {4} (c3.105);
\draw[myarrow] (c5.305) -- node[sloped, color = green, font = \scriptsize, above] {2} (c3.55);

\draw[myarrow] (c2.55) -- node[sloped, color = red, font = \scriptsize, above] {3} (c5.200);
\draw[myarrow] (c2.15) -- node[sloped, color = blue, font = \scriptsize, below] {4} (c5.240);

\draw[myarrow] (c3) -- node[sloped, color = blue, font = \scriptsize, below] {4} (c2);

\draw[myarrow] (c3) -- node[sloped, color = green, font = \scriptsize, below] {2} (c6);

\draw[myarrow] (c6) -- node[sloped, color = green, font = \scriptsize, above] {2} (c5);

\end{tikzpicture}

\end{center}

\caption{A Circulation and its Decomposition\label{fig:decomp}}
\end{figure}
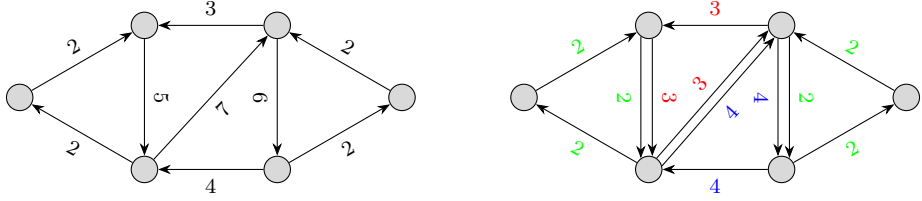

\isathm{Flow_Theory/Decomposition.thy}{471}{480}%
\begin{lem}[Decomposition of Circulations~\cite{FordFulkersonFlowsNetworks}\label{lemma:circdecomp}]
For any circulation $g$, we can find a \textit{decomposition}, i.e.\ there is a set of cycles $\mathcal{C}$ and weights $w : \mathcal{C} \rightarrow \mathbb{R}^+$ where any edge $e$ has $g(e) = \sum \limits_{C \in \mathcal{C} \wedge e \in C} w(C)$. Moreover, if $g(e) \in \mathbb{N}$ for all $e$, then $w(C) \in \mathbb{N}$ for all $C \in \mathcal{C}$.
\end{lem}
This is illustrated by Figure~\ref{fig:decomp}.
\begin{proof}[Proof Sketch]
We use a greedy search to find cycles: as long as there is $e$ with $g (e)\neq 0$, there must be a vertex that has a leaving edge $e$ with $g(e) > 0$.
Since $g$ is a circulation with zero excess everywhere, for every $v$ and $e_1$ with $e_1 \in \deltaminus(v)$ and $g(e_1) > 0$ we can find $e_2 \in \deltaplus(v)$ with $g(e_2) > 0$.
If we repeat this search $n= |\V|$ times, a vertex in the constructed path is repeated implying a cycle $C$.
We assign $w(C):=\min\limits_{e \in C} g(e)$, decrease $g(e)$ by $w(C)$ for every $e \in C$ and continue this process until $g(e)= 0$ for all $e$.
\end{proof}

\newcommand{\findcycle}{\isainl{Flow_Theory/Decomposition.thy}{25}{30}{find-cycle}\xspace}
\newcommand{\tttg}{\texttt{g}\xspace}
\newcommand{\ttte}{\texttt{e}\xspace}
\newcommand{\Abs}{\texttt{Abs}\xspace}
\newcommand{\flowcycle}{\texttt{flowcycle}\xspace}

\paragraph*{Formalisation of Circulation Decomposition}
Here, we encounter some difficulties with non-algebraic/graphical/geometric reasoning:
The existence of a structure (a cycle) is intuitively obvious, yet we need to make its construction formally precise which is done by
the function \findcycle.
It searches for a path by modelling the ``steps'' from the path's informal construction.
The ``continue this process'' from the informal proof corresponds to formally proving \thmcircdecomp by induction on the cardinality of the support (set of edges with strictly positive flow).
The hardest part is showing that the flow remains a circulation, i.e.\ a flow with zero excess at every vertex, after decreasing the flow along one of its cycles.
The central property, which one would barely pay attention to informally, is that for every cycle $C$ found by \findcycle and for every vertex $v$ in the graph,
$v$ either has no leaving or entering edge in $C$ or exactly one leaving and one entering edge in $C$.

\section{Towards a simple Algorithm}
\newcommand{\augcycle}{\ensuremath{C}\xspace}
\newcommand{\augpath}{\ensuremath{P}\xspace}
\newcommand{\flow}{\ensuremath{f}\xspace}
\newcommand{\ssp}{\textit{successive-shortest-path}\xspace}
\newcommand{\ssppar}{\textit{successive-shortest-path}\xspace}
\newcommand{\augmentedge}{\isainl{Flow_Theory/Augmentation.thy}{429}{434}{augment-edge}\xspace}
\newcommand{\augmentedges}{\isainl{Flow_Theory/Augmentation.thy}{454}{456}{augment-edges}\xspace}
\newcommand{\augpathvalidness}{\isainl{Flow_Theory/Augmentation.thy}{1049}{1056}{augment-path-validness-b-pres-source-target-distinct}\xspace}
\newcommand{\costchangeaug}{\isainl{Flow_Theory/Augmentation.thy}{1219}{1221}{cost-change-aug}\xspace}
\newcommand{\oedge}{\texttt{oedge}\xspace}
\newcommand{\set}{\texttt{set}\xspace}
\newcommand{\Rcap}{\texttt{Rcap}\xspace}
\newcommand{\hd}{\texttt{hd}\xspace}
\newcommand{\last}{\texttt{last}\xspace}
\newcommand{\state}{\texttt{state}\xspace}
\newcommand{\recupd}{\texttt{:=}\xspace}
\newcommand{\infeasible}{\texttt{infeasible}\xspace}
\newcommand{\tlpar}{\texttt{(}\xspace}
\newcommand{\trpar}{\texttt{)}\xspace}
\newcommand{\return}{\isainl{Mincost_Flow_Algorithms/SSP_Preps.thy}{13}{13}{return}\xspace}
\newcommand{\getsource}{\isainl{Mincost_Flow_Algorithms/Successive_Shortest_Path.thy}{12}{12}{get-source}\xspace}
\newcommand{\getreachabletarget}{\isainl{Mincost_Flow_Algorithms/Successive_Shortest_Path.thy}{13}{13}{get-reachable-target}\xspace}
\newcommand{\getminaugpath}{\isainl{Mincost_Flow_Algorithms/Successive_Shortest_Path.thy}{14}{14}{get-min-augpath}\xspace}
\newcommand{\fdom}{\texttt{f-dom}\xspace}
\newcommand{\tttinput}{\texttt{input}\xspace}
\newcommand{\sspcallfourcond}{\isainl{Mincost_Flow_Algorithms/Successive_Shortest_Path.thy}{146}{161}{SSP-call-4-cond}\xspace}
\newcommand{\True}{\texttt{True}\xspace}
\newcommand{\False}{\texttt{False}\xspace}
\newcommand{\SSP}{\texttt{SSP}\xspace}
\newcommand{\SSPfun}{\isainl{Mincost_Flow_Algorithms/Successive_Shortest_Path.thy}{36}{54}{SSP}\xspace}
\newcommand{\SSPloc}{\isainl{Mincost_Flow_Algorithms/Successive_Shortest_Path.thy}{9}{33}{SSP}\xspace}
\newcommand{\Algostate}{\isainl{Mincost_Flow_Algorithms/SSP_Preps.thy}{11}{13}{Algo-state}\xspace}
\newcommand{\diffisrescirc}{\isainl{Flow_Theory/Cost_Optimality.thy}{320}{322}{diff-is-res-circ}\xspace}
\newcommand{\rcostdifference}{\isainl{Flow_Theory/Cost_Optimality.thy}{414}{414}{rcost-difference}\xspace}
\newcommand{\noaugcyclemincostflow}{\isainl{Flow_Theory/Cost_Optimality.thy}{487}{489}{no-augcycle-min-cost-flow}\xspace}
\newcommand{\isoptiffnoaugcycle}{\isainl{Flow_Theory/Cost_Optimality.thy}{628}{630}{is-opt-iff-no-augcycle}\xspace}
\newcommand{\SSPupdfour}{\isainl{Mincost_Flow_Algorithms/Successive_Shortest_Path.thy}{202}{210}{SSP-upd4}\xspace}
\newcommand{\sspsimps}{\isainl{Mincost_Flow_Algorithms/Successive_Shortest_Path.thy}{226}{231}{SSP-simps}\xspace}
\newcommand{\sspinduct}{\isainl{Mincost_Flow_Algorithms/Successive_Shortest_Path.thy}{252}{256}{SSP-induct}\xspace}
\newcommand{\invaropt}{\isainl{Mincost_Flow_Algorithms/SSP_Preps.thy}{119}{119}{invar-opt}\xspace}
\newcommand{\invaroptholdsfour}{\isainl{Mincost_Flow_Algorithms/Successive_Shortest_Path.thy}{644}{646}{invar-opt-holds-4}\xspace}
\newcommand{\invaroptholds}{\isainl{Mincost_Flow_Algorithms/Successive_Shortest_Path.thy}{725}{727}{invar-opt-holds}\xspace}
\newcommand{\invaroptinitialstate}{\isainl{Mincost_Flow_Algorithms/Successive_Shortest_Path.thy}{770}{771}{invar-opt-initial-state}\xspace}
\newcommand{\totalcorrectness}{\isainl{Mincost_Flow_Algorithms/Successive_Shortest_Path.thy}{820}{824}{total-correctness}\xspace}

\emph{Augmentation} is a principal technique in combinatorial optimisation by which we incrementally improve a candidate solution until an optimal solution is found.
In the context of flows, augmenting (along) a forward edge by a positive real $\gamma$ means to increase the flow assigned to the original edge by $\gamma$.
We augment along a backward edge by decreasing the flow value of the original edge by $\gamma$.
One can augment along a path by augmenting along each edge contained.
We call a path of residual edges along which the residual capacities are strictly positive an \textit{augmenting path}.
Closed augmenting paths \augpath with $\residualc(p) < 0$ are \textit{augmenting cycles}.
\begin{exa}
In Figure~\ref{graphC}, one can see the result of augmenting the flow from Figure~\ref{graphA} along the edges $(u,v)$ and $(v,u)$ by $2$.
Figure~\ref{graphD} shows the resulting new residual network.
Please note the change in capacities in forward and backward edges.
\end{exa}

The formalisation of augmentation consists of two functions:
\augmentedge augments the flow \tttf along a single edge \ttte, which \augmentedges uses to augment along a sequence of residual edges, which is formalised as an ordinary list.
We can then show \augpathvalidness, which is that augmenting a valid $b$-flow \flow by $\gamma \leq \residualu_f(p)$ along an augmenting path \augpath from a vertex $s$ to a vertex $t$ results in a flow $f'$ that is valid for the balances $b'$ where $b'(s) = b(s) + \gamma$, $b'(t) = b(t) - \gamma$ and $b'(v) = b(v)$ for $v \neq s,t$. We also show in \costchangeaug that the costs change by $\gamma \cdot \residualc(p)$ when augmenting along \augpath by $\gamma$.

\renewcommand{\frame}{}
\begin{figure}[]
\frame{
\begin{subfigure}{0.485\textwidth}
{\centering
\frame{
   \trimbox{0mm {0.5\baselineskip} 0mm 0mm}{
\begin{tikzpicture}[
      mycircle/.style={
         circle,
         draw=black,
         fill=gray,
         fill opacity = 0.3,
         text opacity=1,
         inner sep=0pt,
         minimum size=12pt,
         font=\normalsize},
      myarrow/.style={-Stealth},
      node distance=1cm and 4cm
      ]
      \node[ mycircle] (c1) {$u$} ;
      \node[ mycircle,right=of c1] (c2) {$v$};
     \draw [myarrow] (c1) to [out=10,in=170] node[sloped,font=\footnotesize, above]
      {$5$, \color{forestgreen} $8$\color{black}\text{,} \color{red} $2$} (c2);
     \draw [myarrow] (c2) to [out=250,in=290] node[sloped,font=\footnotesize, above]
      {$4$, \color{forestgreen} $6$\color{black}\text{,} \color{red} $3$} (c1);
      \end{tikzpicture}
      }
      }
      \par}
      \caption{\label{graphA} The original flow}
\end{subfigure}
}
\nolinebreak
\frame{
\begin{subfigure}{0.485\textwidth}
{\centering
   \frame{
      \trimbox{0mm {0.25\baselineskip} 0mm 0mm}{
\begin{tikzpicture}[
      mycircle/.style={
         circle,
         draw=black,
         fill=gray,
         fill opacity = 0.3,
         text opacity=1,
         inner sep=0pt,
         minimum size=12pt,
         font=\normalsize},
      myarrow/.style={-Stealth},
      node distance=1cm and 4cm
      ]
      \node[ mycircle] (c3) {$u$} ;
      \node[ mycircle,right=of c3] (c4) {$v$};
     \draw [myarrow, blue] (c3) to [out=10,in=170] node[sloped,font=\footnotesize, above=-0.75mm]
      {\color{forestgreen} $3 (=8-5)$\color{black}\text{,} \color{red} $2$} (c4);
    \draw [myarrow, purple, dashed] (c4) to [out=190,in=350] node[sloped,font=\footnotesize, below=-0.75mm]
      {\color{forestgreen} $5$\color{black}\text{,} \color{red} $-2$} (c3);
    \draw [myarrow, purple, dashed] (c3) to [out=35,in=145] node[sloped,font=\footnotesize, above=-1mm]
      {\color{forestgreen} $4$\color{black}\text{,} \color{red} $-3$} (c4);
    \draw [myarrow, blue] (c4) to [out=215,in=325] node[sloped,font=\footnotesize, below]
      {\color{forestgreen} $2(=6-4)$\color{black}\text{,} \color{red} $3$} (c3);
      \end{tikzpicture}
      }
      }
      \par}
     \caption{\label{graphB} Original Residual Graph}
      \end{subfigure}
} \\

\frame{
   \begin{subfigure}{0.485\textwidth}
   {\centering
      \frame{
      \trimbox{0mm {0.5\baselineskip} 0mm 0mm}{
\begin{tikzpicture}[
      mycircle/.style={
         circle,
         draw=black,
         fill=gray,
         fill opacity = 0.3,
         text opacity=1,
         inner sep=0pt,
         minimum size=12pt,
         font=\small},
      myarrow/.style={-Stealth},
      node distance=1cm and 4cm
      ]
    \node[ mycircle] (d1) {$u$} ;
    \node[ mycircle,right=of d1] (d2) {$v$};
    \draw [myarrow] (d1) to [out=10,in=170] node[sloped,font=\footnotesize, above]
      {$7$, \color{forestgreen} $8$\color{black}\text{,} \color{red} $2$} (d2);
     \draw [myarrow] (d2) to [out=250,in=290] node[sloped,font=\footnotesize, above=]
      {$6$, \color{forestgreen} $6$\color{black}\text{,} \color{red} $3$} (d1);
    \end{tikzpicture}
    }
    }
\par}
\caption{\label{graphC} Flow after Augmentation}
\end{subfigure}
}
\nolinebreak
\frame{
   \begin{subfigure}{0.485\textwidth}
   {\centering
      \frame{
      \trimbox{0mm {0.25\baselineskip} 0mm 0mm}{
\begin{tikzpicture}[
      mycircle/.style={
         circle,
         draw=black,
         fill=gray,
         fill opacity = 0.3,
         text opacity=1,
         inner sep=0pt,
         minimum size=12pt,
         font=\small},
      myarrow/.style={-Stealth},
      node distance=1cm and 4cm
      ]

      \node[ mycircle] (d3) {$u$} ;
      \node[ mycircle,right=of d3] (d4) {$v$};
     \draw [myarrow, blue] (d3) to [out=10,in=170] node[sloped,font=\footnotesize, above=-0.75mm]
      {\color{forestgreen} $1$\color{black}\text{,} \color{red} $2$} (d4);
    \draw [myarrow, purple, dashed] (d4) to [out=190,in=350] node[sloped,font=\footnotesize, below=-0.75mm]
      {\color{forestgreen} $7$\color{black}\text{,} \color{red} $-2$} (d3);
    \draw [myarrow, purple, dashed] (d3) to [out=35,in=145] node[sloped,font=\footnotesize, above=-1mm]
      {\color{forestgreen} $6$\color{black}\text{,} \color{red} $-3$} (d4);
    \draw [myarrow, blue] (d4) to [out=215,in=325] node[sloped,font=\footnotesize, below]
      {\color{forestgreen} $0$\color{black}\text{,} \color{red} $3$} (d3);
    \end{tikzpicture}
    }
    }
\par}
\caption{\label{graphD} Residual Graph after Augmentation}
\end{subfigure}
}

\caption{\label{fig2} Flows, (residual) capacities and (residual) costs are black, green and red, respectively. Forward edges are blue, backward edges purple.
Subfigures~\ref{graphA} and~\ref{graphB} are a repetition of Figure~\ref{fig2o} for the sake of readability.}
\end{figure}

{
\SetAlCapHSkip{0pt}
\SetAlgoHangIndent{0pt}
\SetInd{0.5em}{0.5em}
\SetVlineSkip{0.3mm}
\setlength{\algomargin}{1.2cm}
\begin{algorithm}[t!]
\renewcommand{\theAlgoLine}{SSP\arabic{AlgoLine}}
\caption{$\ssp(\mathcal{E}, \mathcal{V}, u, c, b)$\label{algo:SSP}}
initialise $b' \leftarrow b$ and $f \leftarrow 0$;\\
\While{True}{
\eIf{$\forall \, v \in \mathcal{V}. \; b'(v) = 0$}{
\textbf{return} $f$ as optimum flow;
}
{
take some $s$ with $b'(s) > 0$;\\
\uIf{$\exists \, t$ reachable from $s$ $\wedge$ $b'(t) < 0$}{
take such a $t$ and a min-cost augpath $P$ from $s$ to $t$;\\
$\gamma = \min \lbrace b(s), - b(t), \mathfrak{u}_f (P)\rbrace$ ;\\
augment $\gamma$ along $P$;\\
$b'(s) \leftarrow b'(s) - \gamma$; $b'(t) \leftarrow b'(t) + \gamma$;
}
\textbf{else} \textbf{return} \textit{infeasible};
}
}
\end{algorithm}
}

For this and all subsequent sections, we fix a weight-conservative flow network $\mathcal{G} = (\E, \V, u, c)$.
Unless said otherwise, costs and capacities refer to this network.
We keep the balances $b$ generic, since these are subject to change during the computation.
Algorithm~\ref{algo:SSP} is one of the most basic minimum-cost flow algorithms.
$\ssp(\mathcal{G}, b)$ repeatedly selects a source $s$ with positive balance, a target $t$ with negative balance and a \emph{minimum-cost augmenting path} \augpath connecting $s$ to $t$,
i.e.\ a minimum-cost path in the residual network connecting $s$ to $t$.
Following that, it sends as much flow as possible, i.e.\ as much as the minimum capacity of any residual edge in \augpath or as the balances of $s$ and/or $t$ allow, from $s$ to $t$ along \augpath.
We lower and increase the balances at $s$ and $t$ by the same amount, respectively.
We continue this until all balances reach zero or until we can infer infeasibility from the absence of an augmenting path.

Conceptually, the algorithm fixes the variables $\E$, $\V$, $u$, $c$ and $b$, and works on \textit{program states} consisting of the changing variables \flow (current flow) and $b'$ (remaining balances).
Invariants are predicates defined on states and might use the fixed variables, as well.
If not all variables are relevant to an invariant, we say that only the involved variables satisfy the invariant.

The loop of Algorithm~\ref{algo:SSP} is formalised as the function \SSPfun whose non-trivial termination argument requires the Isabelle function package~\cite{functionPackageIsabelle}.
We use a record \Algostate to model program states.
Most notation is standard functional programming notation.
The main exception is record updates, e.g.\ $\state \tlpar \return \; \recupd \; \infeasible \trpar$ denotes \state, but with the variable \return updated to \infeasible .
Formally, selecting reachable targets and minimum-cost paths corresponds to using functions (\getsource, \getreachabletarget and \getminaugpath, respectively) that compute those items non-deterministically.
These are fixed by a locale called \SSPloc that also makes assumptions about the functions' behaviour.

\paragraph{Correctness of Algorithm~\ref{algo:SSP}.} To prove that the algorithm is correct, we show that the following invariants hold for the states encountered during the main loop of the procedure $\ssp(\mathcal{G}, b)$ (Algorithm~\ref{algo:SSP}):

\begin{enumerate}[label=(ISSP-\theenumi)]
\renewcommand{\theenumi}{ISSP-\arabic{enumi}}
\item The flow \flow is a minimum-cost flow for the balance $b - b'$\footnote{For any $v$, we define this as $(b-b')(v) = b(v) - b'(v)$.}.
\label{invar:optimality}
\item If capacities $u$ and balances $b$ are integral, then $b'(x)$ and $f(e)$ are integral for any vertex $x \in \V$ and $e \in \E$, respectively.
\label{invar:integrality}
\item The sum of $b'$ over all vertices $v$ is zero: $\smash{\sum \limits _{v \in \V} b'(v) = 0}$.
\label{invar:zerobalance}
\end{enumerate}

Invariant~\ref{invar:optimality} is well known in the literature~\cite{KorteVygenOptimisation} and we made the two simple ones explicit by the formalisation.

Proving Invariant~\ref{invar:optimality} was the most demanding and we dedicate most of this section to it.
The other two invariants easily follow from the algorithm's structure.
Correctness of all non-trivial algorithms for minimum-cost flows depends on the following optimality criterion:
\isathm{Flow_Theory/Cost_Optimality.thy}{628}{630}
\begin{thm}[Optimality Criterion~\cite{optCritMinCostFlows}]
\label{optcrit}
A flow \flow valid for balance $b$ is optimum iff there is no augmenting cycle w.r.t.\ \flow.
\end{thm}

\begin{proof}[Proof Sketch]
$(\Rightarrow)$ An augmenting cycle is a possibility to decrease costs while still meeting the balance constraints.

\noindent $(\Leftarrow)$ Assume a valid flow $f'$ with $c(f') < c(f)$.
We define the flow $g$ in the residual graph as $g(F\, e) = \max \lbrace  0, f'(e) - f(e)\rbrace$ and $g(B\,e) = \max \lbrace  0, f(e) - f'(e)\rbrace$.
$g$ is a flow in the residual graph with zero excess for every vertex, i.e.\ a circulation.
Moreover, $\residualc(g) = c(f') - c(f)$ and any residual edge $e$ with $\residualu_f(e) = 0$ has $g(e) = 0$.
We can decompose the circulation $g$ according to Lemma~\ref{lemma:circdecomp}, i.e.\ there is a set of cycles $\mathcal{C}$ and weights $w : \mathcal{C} \rightarrow \mathbb{R}^+$ where any residual edge $e$ has $g(e) = \sum \limits_{C \in \mathcal{C} \wedge e \in C} w(C)$.
Since $0 > c(f') - c(f) = \residualc(g) = \sum \limits _{C \in \mathcal{C}} w(C) \cdot \residualc(C)$, there has to be a cycle $C \in \mathcal{C}$ where $\residualc(C) < 0$. $\residualu_f(C)$ must be positive making this an augmenting cycle w.r.t.\ \flow.
\end{proof}

To show suboptimality in case of an augmenting cycle, we use results like \costchangeaug.
To formalise the second direction of Theorem~\ref{optcrit}, we interpret the residual graph as a flow network \residualflow.
We show in \diffisrescirc that the difference flow is a circulation in the residual graph.
This allows us to apply the formal counterpart of Lemma~\ref{lemma:circdecomp} to the difference flow.
After reasoning about the costs of the difference flow in \rcostdifference, we can show \noaugcyclemincostflow as the second direction of Theorem~\ref{optcrit}.

For proving the mathematical main result of this section, namely, Theorem~\ref{thm911}~\cite{jewellMinCostSSP}, one has to prove a topological property if certain edges are removed from a graph (Lemma~\ref{FBPelim}).
This is a gap in all published proofs of which we are aware and we document our reasoning here.

We consider colourings of the residual edges in red and blue.
We count the same edge coloured once in red and once in blue as two distinct edges.
We saw that for a graph $G$, the residual edges emerging from the edges in $G$ form a multigraph.
For any set $R$ of residual edges (which is a multigraph with $\fstv$ and $\sndv$ in the sense of Section~\ref{sec:background_basic}), the graph consisting of all edges in $R$ both in red and blue is another multigraph, provided that $\fstv$ and $\sndv$ are defined to ignore the colours.
The notions of paths and cycles as defined in Section~\ref{sec:background_basic} apply again.

\begin{exa}\label{example:colours}
Assume a simple graph consisting of three vertices $v_1$, $v_2$ and $v_3$ which are all connected to one another.
When taking colours into account, $p = [{\color{red}F\, (v_1,v_2)}, {\color{red} F\,(v_2,v_3)}, \\{\color{blue} B \, (v_2, v_3)}, {\color{blue}B\, (v_3, v_2)}, {\color{red}F \, (v_3, v_1)},{\color{blue}F\, (v_1,v_2)}]$ is an edge-distinct path visiting $v_1$, $v_2$, $v_3$, $v_2$, $v_3$, $v_1$ and $v_2$.
${\color{red}F\, (v_1,v_2)}$ and ${\color{blue}F\, (v_1,v_2)}$ are two distinct edges which is why the path is edge-distinct despite the repetition.
If we drop the last edge, we obtain a cycle leading from $v_1$ to $v_1$.
We see that for one edge $e$ in the original graph, there would be four coloured residual edges: ${\color{red} F \, e}$, ${\color{red} B \, e}$, ${\color{blue} F \, e}$ and ${\color{blue} B \, e}$.
(\augpath does not use all of them, however).
\end{exa}

We call pairs of residual edges of different colours, where one edge is the reversed one of the other, \textit{forward-backward-pairs (FBP)}, e.g.\ $F\;e$ in blue and $B\;e$ in red.
\textit{FBP}s are involved in lemmas one can use to prove preservation of the optimality invariant of Algorithm~\ref{algo:SSP}.
Subsequently, \textit{disjointness} of paths and cycles means their edge-disjointness.
\begin{figure}
{\centering
\rotatebox{0}{
\trimbox{0.5cm 1.6cm 0.6cm 1.8cm}{
\begin{tikzpicture}[
      mycircle/.style={
         circle,
         draw=black,
         fill=black,
         text opacity=1,
         inner sep=0pt,
         minimum size=3pt,
         font=\small},
      myarrow/.style={-Stealth},
      node distance=0cm and 2cm
      ]
      \node (k) at (0,0) {};
      \node[mycircle,right] (c2) {};
      \node[mycircle,right= 0.5cm of c2] (c3) {};
     \draw [myarrow, blue, thick] (c3) to [in = 40, out=140] (c2);
     \draw [myarrow, red, thick ] (c2) to [out = -40, in=-140] (c3);
     \draw [myarrow, black, dotted, thick, -stealth] (c3) to [out = 60, in=-60, distance=2.6cm] (c3);
      \draw [myarrow, black, dotted, thick, stealth-] (c2) to [out = 120, in=240, distance=2.6cm] (c2);
 \end{tikzpicture}
 }
 }
 \hspace*{2cm}
 \raisebox{-1mm}{
 \rotatebox{90}{
\trimbox{0.9cm 0.4cm 0.9cm 0.4cm}{
\begin{tikzpicture}[
      mycircle/.style={
         circle,
         draw=black,
         fill=black,
         text opacity=1,
         inner sep=0pt,
         minimum size=3pt,
         font=\small},
      myarrow/.style={-Stealth},
      node distance=0cm and 2cm
      ]
      \node[mycircle,right] (c2) {};
      \node[mycircle,right= 0.75cm of c2] (c3) {};
      \draw [myarrow, blue, thick] (c3) to [in = 40, out=140] (c2);
      \draw [myarrow, red, thick ] (c2) to [out = -40, in=-140] (c3);
      \draw [myarrow, black, dotted, thick, stealth- ] (c3) to [out = -60, in=-120, distance=1.8cm] (c2);
      \draw [myarrow, black, dotted, thick, -stealth] (c3) to [out = 60, in=120, distance=2cm] (c2);
\end{tikzpicture}
}
}
}
}
\caption{Eliminating FBPs: \label{FBPcycles} Members of an $FBP$ belong either to the same (left) or to two different cycles (right). When we drop the $FBP$ on the left, we obtain two new cycles. On the right, $FBP$ deletion results in a single new cycle. Disjointness is preserved.}
\end{figure}
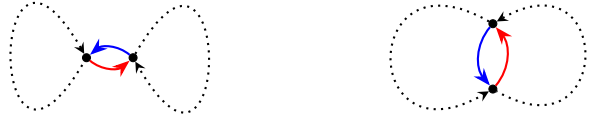
\begin{exa}
In Example~\ref{example:colours}, $\lbrace {\color{red} F\,(v_2,v_3)}, {\color{blue} B \, (v_2, v_3)} \rbrace$ is the only FBP among $p$'s edges.
\end{exa}

\isathm{Mincost_Flow_Algorithms/Path_Aug_Opt.thy}{818}{827}
\begin{lem}\label{FBPcycle}
Deleting all $FBP$s from a set of disjoint cycles yields another set of disjoint cycles.
\end{lem}

\begin{proof}[Proof Sketch]
Proof by induction on the number of $FBP$s.
For every $FBP$, the partners either belong to the same or to a different cycle, as one can see in Figure~\ref{FBPcycles}.
Deleting an $FBP$ results in two or a single new cycle, respectively.
\end{proof}

\isathm{Mincost_Flow_Algorithms/Path_Aug_Opt.thy}{1736}{1755}
\begin{lem}\label{FBPelim}
Assume an $s$-$t$-path \augpath and some cycles $\mathcal{C}$ where every $FBP$ is between the path and a cycle, all items disjoint with one another.
Deleting all $FBP$s results in an $s$-$t$-path and some cycles, again all disjoint.
\end{lem}
\begin{proof}[Proof Sketch]
We add an edge $e^*$ from $t$ to $s$.
By definition, let it be distinct from any other edge and not be a member of any $FBP$.
The result of adding this additional edge is an edge-disjoint set of edge-distinct cycles $\mathcal{C}$, to which we can apply Lemma~\ref{FBPcycle}.
This yields another set of disjoint cycles $\mathcal{D}$ where exactly the $FBP$s have been removed.
As $e^*$ cannot be member of an $FBP$ in $\mathcal{C}$, one of the cycles from $\mathcal{D}$ must still contain this edge.
Removing $e^*$ from $\mathcal{D}$ results in an $s$-$t$-path and some cycles.
\end{proof}

\noindent The following theorem implies the preservation of Invariant~\ref{invar:optimality}.

\isathm{Mincost_Flow_Algorithms/Path_Aug_Opt.thy}{2429}{2434}
\begin{thm}[Optimality Preservation\label{thm911}~\cite{jewellMinCostSSP}]
Let \flow be a minimum-cost flow for balances $b$.
Take an $s$-$t$-path \augpath of minimum residual costs and $\gamma \leq \residualu_f(P)$.
If we augment \flow by $\gamma$ along \augpath then the result is still optimum for modified balances $b'$ where $b'(s) = b(s) + \gamma$, $b'(t) = b(t) - \gamma$ and $b'(v) = b(v)$ for any other $v$.
\end{thm}
\begin{proof}[Proof Sketch]
\augpath is vertex-disjoint since any cycle in \augpath would have positive costs (Theorem~\ref{optcrit}) contradicting the optimality of \augpath.
Assume the flow $f'$ after the augmentation is not optimum.
By Theorem~\ref{optcrit}, there exists an augmenting cycle \augcycle.
Wlog.\ \augcycle is vertex-disjoint: otherwise split \augcycle into vertex-disjoint cycles of which one has negative residual costs.

We now construct an augmenting cycle for \flow.
First colour edges in \augpath blue and edges in \augcycle red.
Neither \augpath nor \augcycle can contain any $FBP$s.
We can therefore apply Lemma~\ref{FBPelim} to the union of \augpath and \augcycle yielding another $s$-$t$-path $P'$ and a set of cycles $\mathcal{C}'$.
Their edges have positive residual capacity w.r.t.\ \flow:
for any $e \in P' \cup \bigcup \mathcal{C}'$ we have $\residualu_f(e) > 0$ or $\residualu_{f'}(e) > 0$.
If only the latter holds, then $e \in C$ and $\overset{\leftarrow}{e} \in P$ which is an $FBP$.
However, we would have deleted this.
Since deleting $FBPs$ preserves costs, we have $\residualc(P') + \residualc(\mathcal{C}') = \residualc(P) + \residualc(C)$.
Because $\residualc(P') \geq \residualc(P)$ (optimality of \augpath) and $\residualc(C) < 0$, there must be $D \in \mathcal{C}'$ with $\residualc(D) < 0$.
This is an augmenting cycle w.r.t.\ \flow, which, together with Theorem~\ref{optcrit}, contradicts the assumption that \flow is optimal.
\end{proof}

\isathm{Mincost_Flow_Algorithms/Successive_Shortest_Path.thy}{820}{824}
\begin{thm}[Correctness of Algorithm~\ref{algo:SSP}~\cite{EdmondsKarpScalingFlows}]\label{SSPtotalcorrect}
Assume the sum of balances $b$ over $\V$ is zero.
An execution of $\ssp(\mathcal{G}, b)$ terminates, decides about the existence of a valid flow and returns one in case of existence.
\end{thm}
\begin{proof}[Proof Sketch]\let\qed\relax
Due to weight conservativity, the zero flow is optimum for the zero balance making the optimality invariant (Invariant~\ref{invar:optimality}) initially true.
The preservation of Invariant~\ref{invar:optimality} follows from Theorem~\ref{thm911}.
Invariants~\ref{invar:zerobalance} and~\ref{invar:integrality} also hold initially.
We can see their preservation from the algorithm.

As the $\gamma$ used for the augmentations will be a natural number,
the sum of the absolute values of balances will decrease yielding a termination measure.

It remains to show that there is no valid flow if the procedure returns \textit{infeasible}, a case for which we need a further auxiliary result.
We note that \flow is a valid flow w.r.t.\ $b-b'$.
Lemma~\ref{lemma:flowvalue} (Flow-Value Lemma) and Lemma~\ref{lemma:residualsat} yield for any $b$ and any flow \flow valid w.r.t.\ $b$:
\isaeq{Flow_Theory/Residual.thy}{1047}{1049}
\begin{equation}
\sum \limits_{x \in Rescut_f(v)} b(x) = \sum \limits_{e \in \Delta^+(Rescut_f(v))} f(e) =  \text{\upshape cap}(Rescut_f(v)) \hfill \label{lemma:corosat}
\end{equation}
\noindent From the algorithm's control flow we can infer that, if the algorithm returns \textit{infeasible} there must be an $s$ with $b'(s) > 0$ without a reachable $t$ where $b'(t) < 0$,
i.e.\ any $x$ in the rescut has a $b' (x)\geq 0$.
We show that there is a contradiction if there exists a flow $f'$ valid w.r.t.\ $b$.
{\allowdisplaybreaks
\begin{align*}
\sum \limits_{x \in Rescut_f(s)} b(x) & \leq
\text{\upshape cap}(Rescut_f(s))  & \text{($f'$ is valid flow and Corollary~\ref{corobalance})} \\
& = \sum \limits_{x \in Rescut_f(s)} (b-b')(x)  & \text{(Equality~\ref{lemma:corosat} for \flow and $b - b'$)} \\
& <
\sum \limits_{x \in Rescut_f(s)} b(x) & \hspace*{7.9cm} \qedsymbol
\end{align*}
}
\end{proof}

\paragraph{Formalising Correctness of Algorithm~\ref{algo:SSP}}

We introduced definitions to specify which execution branch is taken when doing an iteration of the loop body,
e.g.\ \sspcallfourcond indicating the recursive case of \SSPfun.
It has the same structure as the loop body and returns \True for exactly one branch and \False for all others.
Similarly, we model the effect of a single execution branch, e.g.\ \SSPupdfour \state.
We can show a simplification lemma \sspsimps and an induction principle \sspinduct for \SSP.
We prove the preservation of the invariants for single updates, e.g.\ \invaroptholdsfour for the optimality invariant \invaropt,
and by the simplification and induction lemmas, we can lift this to a complete execution of \SSP giving the preservation lemma \invaroptholds for the whole loop.
Proof automation makes this process very smooth and convenient.
We show the invariants for the initial state, e.g. \invaroptinitialstate, and obtain \totalcorrectness.
We follow the same formalisation methodology for all the algorithms we consider below.

\begin{rem}
Lemma~\ref{FBPelim} is a gap in all combinatorial proofs of the algorithm's correctness to which we have access, including the proof by Korte and Vygen~\cite{KorteVygenOptimisation}.
The only other complete proof of which we are aware is a non-combinatorial proof by Orlin~\cite{OrlinScalingFlow, OrlinFlowBook}, in which he uses LP-theory.
The main challenge is being able to show that, if after the augmentation, $f'$ is not optimal, then \flow must have also not been optimal.
To construct an augmenting cycle for \flow, we need to extract one cycle from \augcycle that is disjoint from all other cycles in \augcycle.
Otherwise the proof that one cycle has negative cost would be impossible due to the interference in summations of costs of members of \augcycle.
To do that we had to devise the machinery of colouring the edges, FBPs, and the removal of FBPs.
Other than that machinery, the proof boils down to one trick: adding an extra edge to turn the path \augpath into a cycle makes the set of edges from \augcycle and \augpath a union of disjoint cycles, which is more uniform, and which allows us to easily prove Lemma~\ref{FBPelim}.
\end{rem}

\section{The Capacity Scaling Algorithm}
\newcommand{\getsourcetargetpath}{\isainl{Mincost_Flow_Algorithms/Scaling.thy}{1094}{1119}{get-source-target-path}\xspace}
\newcommand{\notyettermcs}{\isainl{Mincost_Flow_Algorithms/SSP_Preps.thy}{10}{10}{notyetterm}\xspace}
\newcommand{\sspcs}{\isainl{Mincost_Flow_Algorithms/Scaling.thy}{112}{127}{SSP}\xspace}
\newcommand{\sspthr}{\isainl{Mincost_Flow_Algorithms/Scaling.thy}{1139}{1139}{ssp}\xspace}
\newcommand{\scalingloc}{\isainl{Mincost_Flow_Algorithms/Scaling.thy}{1092}{1122}{Scaling}\xspace}
\newcommand{\scalingfun}{\isainl{Mincost_Flow_Algorithms/Scaling.thy}{1150}{1157}{Scaling}\xspace}
\newcommand{\invarintegralholdscs}{\isainl{Mincost_Flow_Algorithms/Scaling.thy}{493}{495}{invar-integral-holds}\xspace}
\newcommand{\integralbalanceterminationcs}{\isainl{Mincost_Flow_Algorithms/Scaling.thy}{525}{527}{integral-balance-termination}\xspace}
\newcommand{\scalingtotal}{\isainl{Mincost_Flow_Algorithms/Scaling.thy}{1405}{1405}{Scaling-total}\xspace}
\newcommand{\invarcomb}{\isainl{Mincost_Flow_Algorithms/Scaling.thy}{1318}{1318}{invar-comb}\xspace}
\newcommand{\totalcorrectnesscs}{\isainl{Mincost_Flow_Algorithms/Scaling.thy}{1534}{1539}{total-correctness}\xspace}
\newcommand{\notyetterm}{\texttt{notyetterm}\xspace}
\newcommand{\tttl}{\texttt{l}\xspace}
\newcommand{\scaling}{\texttt{scaling}\xspace}
\newcommand{\ascaling}{\textit{scaling}\xspace}
\newcommand{\ascalingpar}{\textit{scaling}\xspace}

\begin{figure}[t!]
{
\lineskiplimit=-\maxdimen
\captionsetup{labelformat=empty}
\SetAlCapHSkip{0pt}
\SetAlgoHangIndent{0pt}
\SetInd{0.5em}{0.5em}
\SetVlineSkip{0.3mm}
\setlength{\algomargin}{1cm}
\begin{algorithm}[H]
\renewcommand{\theAlgoLine}{CS\arabic{AlgoLine}}
\caption{\textit{capacity-scaling}$(\mathcal{E}, \mathcal{V}, u, c, b)$}
\label{algo:scling}
initialise $b' \leftarrow b$, $f \leftarrow 0$ and $\gamma = 2^{\lfloor \log_2 B\rfloor}$ where \smash{$B = \max \lbrace 1, \frac{1}{2} \sum \limits_{v \in \mathcal{V}} |b(v)| \rbrace$}; \\
\While{True}{
\While{True}{
\textbf{if} $\forall \, v \in \mathcal{V}. \; b'(v) = 0$ \textbf{then} \textbf{return} $f$;\\
\uElseIf{$\exists \, s \, t \, P$. $P$ is $s$-$t$-path, $\mathfrak{u}_f(P)\geq\gamma$, $b'(s) \geq \gamma$ and $b'(t) \leq - \gamma$}{
take such $s$, $t$ and $P$; augment $\gamma$ along $P$; \\
$b'(s) \leftarrow b'(s) - \gamma$; $b'(t) \leftarrow b'(t) + \gamma$;
}
\textbf{else} \textbf{break};
}

\textbf{if} $\gamma = 1$ \textbf{then} \textbf{return} $\text{\textit{infeasible}}$;\\
\textbf{else} $\gamma \leftarrow \frac{1}{2} \cdot \gamma$;
}
\end{algorithm}
}

\end{figure}

The naive successive shortest path algorithm (Algorithm~\ref{algo:SSP}) arbitrarily selects sources, targets and min-cost paths for the augmentations.
We refine this to \textit{Capacity Scaling (CS)} by selecting those triples where the residual capacities and balances are above a certain threshold that we halve from one scaling phase to another (Algorithm~\ref{algo:scling}).
It was proposed by Edmonds and Karp~\cite{EdmondsKarpScalingFlows}.
As SSP, CS fixes $\E$, $\V$, $u$, $c$ and $b$, and works on a state consisting of the changing variables $b'$ and $f$.
The algorithm uses two nested loops:
the outer one is responsible for monitoring the scaling and determining problem infeasibility.
The inner one's purpose is to process every suitable path until none are remaining.
It is also responsible for terminating the execution when a solution has been found.
As this refines SSP, the proofs for correctness and termination do not differ significantly.
We can reapply the same three invariants.
Note that capacities and balances must still be integral to ensure termination.

Intuitively CS behaves like SSP, but only greedily chooses large steps towards the optimal solution,
thus hastening progress leading to a polynomial rather than exponential worst-case running time.
Each of these steps is a minimum-cost augmenting path $p$ from a source $s$ to a target $t$ with a `step size' of $\min \lbrace \mathfrak{u}_f(p), b(s), -b(t) \rbrace$.
When no paths with the right step size remain, it halves the thresholds for treatment and continues with a more fine-grained analysis.

\paragraph{Partial Correctness.}
A modified \sspcs realises the inner loop, which we define in a locale fixing the threshold $\gamma$ and the selection function.
The latter returns paths that we assume to have capacity above $\gamma$.
The top level loop's locale \scalingloc assumes a generic selection function such that \getsourcetargetpath $\gamma$ returns a source, a target and a minimum-cost path with capacity above $\gamma$.
For a fixed $\gamma$, \getsourcetargetpath $\gamma$ serves as the inner loop's selection function.
The outer loop \scalingfun is a function on two arguments, namely, the logarithm \tttl of the threshold and the program state.
To halve $\gamma$, it decreases \tttl in every iteration.
In contrast to the \SSPfun formalising Algorithm~\ref{algo:SSP},
\sspcs uses the flag \notyettermcs to indicate that no more suitable paths exist, and the decision on infeasibility is left to the outer loop.

We inherit most of the claims and proofs from the formalisation of SSP and therefore they are conditioned on termination for the respective input state like the preservation lemma \invarintegralholdscs.
Termination of the inner loop still follows from Invariant~\ref{invar:integrality}.
The lemma \scalingtotal proves the outer loop's termination by the decrease in $\gamma$.
Together with invariants of the outer loop (\invarcomb), we prove a total correctness statement \totalcorrectnesscs for capacity scaling.

\paragraph{Termination.}
For SSP, we had the sum of the balances' absolute values as a termination measure decreasing in any iteration by at least 1.
This is linear in the balances and therefore exponential w.r.t.\ input length.
In contrast, CS halves the measure after a polynomial number of iterations resulting in fast progress.
The number of scaling phases is logarithmic w.r.t.\ the greatest balance and thus, linear in terms of input length.
The time for finding a minimum-cost path is polynomial and an augmentation is $\mathcal{O}(n)$.
Provided infinite capacities, the number of augmentations per phase is at most $4n$~\cite{KorteVygenOptimisation}.
\begin{rem}
For an efficient path computation, CS even runs in $\mathcal{O}(n(n^2 + m) \log B)$~\cite{EdmondsKarpScalingFlows, KorteVygenOptimisation},
where $B$ is the greatest absolute value of a balance.
This is polynomial w.r.t.\ input length including the representation of balances.
However, since its running time is not polynomial only in the number of vertices and edges, it is not a strongly polynomial algorithm.
\end{rem}

\section{Orlin's Algorithm}
\label{sec:orlins}
\newcommand{\abalance}{\texttt{a-balance}\xspace}
\newcommand{\acurrentflow}{\texttt{a-current-flow}\xspace}
\newcommand{\balance}{\texttt{balance}\xspace}
\newcommand{\currentflow}{\isainl{Mincost_Flow_Algorithms/Orlins_Preparation.thy}{23}{23}{current-flow}\xspace}
\newcommand{\abstractbalmap}{\texttt{abstract-bal-map}\xspace}
\newcommand{\abstractflowmap}{\isainl{Mincost_Flow_Algorithms/Orlins_Preparation.thy}{413}{413}{abstract-flow-map}\xspace}
\newcommand{\torlins}{\texttt{orlins}\xspace}
\newcommand{\torlinsfun}{\isainl{Mincost_Flow_Algorithms/Orlins.thy}{36}{44}{orlins}\xspace}
\newcommand{\torlinsloc}{\isainl{Mincost_Flow_Algorithms/Orlins.thy}{277}{290}{orlins}\xspace}
\newcommand{\newgammafun}{\isainl{Mincost_Flow_Algorithms/Orlins.thy}{28}{34}{new-\ensuremath{\gamma}}\xspace}
\newcommand{\initialstateorlins}{\isainl{Mincost_Flow_Algorithms/Orlins.thy}{2807}{2814}{initial-state-orlins-dom-and-results}\xspace}
\newcommand{\invarintegralorl}{\isainl{Mincost_Flow_Algorithms/Orlins_Preparation.thy}{1107}{1109}{invar-integral}\xspace}
\newcommand{\invarisoptflow}{\isainl{Mincost_Flow_Algorithms/Orlins_Preparation.thy}{1129}{1130}{invar-isOptflow}\xspace}
\newcommand{\tsendflow}{\isainl{Mincost_Flow_Algorithms/Send_Flow.thy}{26}{46}{send-flow}\xspace}
\newcommand{\tmaintainforest}{\isainl{Mincost_Flow_Algorithms/Maintain_Forest.thy}{46}{89}{maintain-forest}\xspace}
\newcommand{\initialstate}{\texttt{initial-state}\xspace}
\newcommand{\initial}{\isainl{Mincost_Flow_Algorithms/Orlins.thy}{57}{65}{initial}\xspace}
\newcommand{\areall}{\isainl{Mincost_Flow_Algorithms/Orlins_Preparation.thy}{36}{36}{are-all}\xspace}
\newcommand{\F}{\ensuremath{\mathcal{F}}}
\newcommand{\implementationinvar}{\isainl{Mincost_Flow_Algorithms/Orlins_Preparation.thy}{488}{499}{implementation-invar}\xspace}
\newcommand{\underlyinginvars}{\isainl{Mincost_Flow_Algorithms/Orlins_Preparation.thy}{924}{932}{underlying-invars}\xspace}
\newcommand{\orlins}{\textit{orlins}\xspace}
\newcommand{\maintainforest}{\textit{maintain-forest}\xspace}
\newcommand{\orlinspar}{\textit{orlins}\xspace}
\newcommand{\maintainforestpar}{\textit{maintain-forest}\xspace}
\newcommand{\flag}{\textit{flag}\xspace}
\newcommand{\ainfeasible}{\textit{infeasible}\xspace}
\newcommand{\asuccess}{\textit{success}\xspace}
\newcommand{\aactives}{\textit{actives}\xspace}

Orlin's algorithm (Algorithm~\ref{algo:orlins}) allows for a strongly polynomial worst-case running time of $\mathcal{O}(n\log n (m + n \log n))$~\cite{OrlinScalingFlow,OrlinFlowBook}.
It requires infinite capacities, and, like the other two algorithms, conservative weights (absence of negative cycles).
Like CS, it reuses the idea of scaling to make progress in larger chunks compared to SSP, but it additionally performs a controlled reduction of the number of non-zero vertices, resulting in a simpler problem.
Only non-zero vertices need further treatment (by augmentations) to receive or emit more flow.

Similar to Algorithm~\ref{algo:scling}, we have an outer loop monitoring the threshold $\gamma$ and an inner loop treating all paths with a capacity above that (\sendflowpar).
After each threshold decrease, the algorithm updates a continuously growing spanning forest $\F$ behind which the intuition is to reduce the number of non-zero vertices.
In the return value of \sendflowpar, the $flag$ indicates whether the algorithm found an optimum flow or detected infeasibility.
Otherwise, the computation continues.

We fix an uncapacitated flow problem ($u=\infty$) consisting of edges $\E$, vertices $\V$, costs $c$ and balances $b$.
The top loop works on a program state consisting of the changing variables $b'$, $f$, $\F$, $\aactives$, $\gamma$ and $r$.
Each subprocedure has access to the fixed variables and expects the changing variables as arguments.
The only exception is that $r$ is not an argument for \sendflowpar because we do not need $r$ for that subprocedure.
For the sake of brevity, we still sometimes write $\sendflow(state)$ although only \maintainforestpar takes the full state.
The subprocedures return the changed variables.

Here, we can augment the flow on only \textit{active edges} and forest edges, which is the purpose of \sendflowpar.
All edges are active initially.
Edges deleted from this set are \textit{deactivated}.
The subprocedures use a small positive constant $\epsilon$.
Its value influences the timespans between increases to the spanning forest.
Positivity ensures termination of \sendflowpar.

For Algorithms~\ref{algo:SSP}~and~\ref{algo:scling}, the running times depended on $B$.
In contrast, Orlin's algorithm avoids that using the continuously growing \textit{spanning forest} $\F$ of edges.
The crucial observation is that this forest is used s.t.\ we only consider one vertex (henceforth, the representative) per forest connected component (henceforth, $\F$-\textit{component}) as a source or as a target in searching for augmenting paths.
The algorithm achieves this reduction in augmentation effort by maintaining the forest, which is possible in time polynomial in $n$ and $m$.

{
\setlength{\algomargin}{0.6cm}
\SetAlCapHSkip{0pt}
\begin{algorithm}[t]
\renewcommand{\theAlgoLine}{O\arabic{AlgoLine}}
\caption{$\orlins(\E, \V, c, b)$}
\label{algo:orlins}
initialise $b' \leftarrow b$; $f \leftarrow 0$; $r(v) \leftarrow v$ for any $v$;
 $\F \leftarrow \emptyset$;
 $\aactives = \E$;
 \smash{$\gamma \leftarrow \max \limits_{v \in \V} |b'(v)|$};\\
\While{True \label{line_cond}}{
$(b', f, \flag) \leftarrow$ $\sendflow(b', f, \F, \aactives, \gamma)$; \label{call_augment_edges}\\
\textbf{if} $\flag = success$ \textbf{then} \textbf{return} $f$;\\
\textbf{if} $\flag = \ainfeasible$ \textbf{then} \textbf{return} $\ainfeasible$;\\
 \textbf{if} $\forall \, e \, \in \, \aactives. \; f(e) = 0$\label{line24}\\
 \textbf{then}  {$\gamma$ $\leftarrow$ $\min \lbrace \frac{\gamma}{2}, \, \smash{\max \limits_{v \in \V}|b'(v)|}\rbrace$};\label{line25}\\
 \textbf{else}   {$\gamma$ $\leftarrow$ $\smash{\frac{\gamma}{2}}$};\label{line27}\\
 $(b', f, \F, r, \aactives) \leftarrow$ $\maintainforest(b', f, \F, \aactives, \gamma, r)$;\label{call_maintain_forest}
}
\end{algorithm}
}

We show partial correctness of the algorithm from invariants on program states and properties of the subprocedures.
Line specifications refer to the state \textit{after} executing the respective line.
We also write e.g.\ $state^i$ or $var^i$ to refer to the state or variable $var$ after executing Line $i$.
Consider the following invariants about the execution of $\orlins(\E, \V, c, b)$:
\newcommand{\refline}[1]{{\color{blue}^{\ref{#1}}}}

\begin{enumerate}[label=(IO-\theenumi)]
\renewcommand{\theenumi}{IO-\arabic{enumi}}
\item $\gamma$ is strictly positive, except when $b = 0$.\label{invargamma}
\item Any active edge $e$ outside $\F$ has a flow $f(e)$ that is a non-negative integer multiple of $\gamma$.\label{invarintegral}
\item Endpoints of a deactivated edge belong to the same $\F$-component.\label{invarsamecomp}
\item Only representatives can have a non-zero balance.\label{invar7}
\item For states in Line~\ref{call_augment_edges},
any edge $e \in \F\refline{call_augment_edges}$ has a flow above $4 n \gamma$,
i.e.\ $f\refline{call_augment_edges}(e) > 4  n \gamma\refline{call_augment_edges}$.\label{invarforest}
\item $f$ is optimum for the balance $b - b'$.\label{invaropt}
\end{enumerate}

With varying degree of explicitness, the invariants are contained in the discussion by Korte and Vygen~\cite{KorteVygenOptimisation}.

The \textit{balance potential} $\Phi$~\cite{KorteVygenOptimisation} is important for the number of augmentations during the subprocedures and their termination.
We define it for $b'$ and $\gamma$ as:
\begin{align*}
\Phi (b', \gamma)= \sum \limits_{v \in \V} \Bigl \lceil \frac{|b'(v)|}{\gamma} -(1 -\epsilon) \Bigr \rceil
\end{align*}
When writing $\Phi(state)$, we refer to the respective $b'$ and $\gamma$.
We will later see in detail how the subprocedures work.
For now, assume the subprocedures have the following properties:

\begin{enumerate}[label=(PO-\theenumi)]
\renewcommand{\theenumi}{PO-\arabic{enumi}}
\item Let $(b'', f', \flag) = \sendflow(b', f, \F, \aactives, \gamma)$.
For any $x \in \V$, we have $|b''(x)| \leq (1 - \epsilon) \cdot \gamma$.
If $flag \neq success$ there is an $x$ with $b''(x) > 0$.
\sendflowpar sets $\flag$ to $\asuccess$ when it reaches $b'' = 0$.\label{P2}
\item Let $(b'', f', \flag) = \sendflow(b', f, \F, \aactives, \gamma)$. For any $e \in \E$, $f'(e) - f(e)$ is an integer multiple of $\gamma$.\label{P2a}
\item If $state$ satisfies Invariants~\ref{invarsamecomp} and~\ref{invar7}, then $\maintainforest(state)$ satisfies them, as well.
Similarly, $\sendflow(state)$ satisfies Invariant~\ref{invar7}, if $state$ satisfies Invariant~\ref{invar7}.\label{P5a}
\item Let $(b'', f', \F', r', \aactives') = \maintainforest(b', f, \F, \aactives, \gamma, r)$.
For any edge $e \in \F'$, the flow reduction $f(e) - f'(e)$ is at most $n\beta$ if $\beta$ is some bound for the balances, i.e.\ $\forall \,v.\; |b'(v)| \leq \beta$.
For edges $e \not \in\F'$, $f'(e) = f(e)$.\label{P4}
\item $\Phi(\maintainforest(state)) \leq \Phi(state) + n$.\label{P51}
\item Provided the other invariants hold, Invariant~\ref{invaropt} (optimality) is preserved by either subprocedure. (For \maintainforestpar and \sendflowpar, we have to replace $4n\gamma$ in Invariant~\ref{invarforest} with $8n\gamma$ and $6n\gamma$, respectively.) \label{P6}
\item Let $(b'', f', \flag) = \sendflow(state)$ where $state = (b', f, \F, \aactives, \gamma)$.
The number of path augmentations during $\sendflow(state)$ is at most $\Phi(state)$.
For all $e \in \E$, $f(e) - f'(e) \leq \Phi (state)\cdot \gamma$, if $f(e) > 6n\gamma$ for all $e \in \F$ in $state$.\label{P52}
\item Assume all invariants hold on $state = (b', f, \F, \aactives, \gamma)$ and let $(b'', f', \flag) = \sendflow(state)$.
If $\flag = \asuccess$, $f'$ is optimum.
If $\flag = \ainfeasible$, the flow problem is indeed infeasible.\label{P1}
\end{enumerate}

In contrast to the invariants, which were taken from the literature, stating the subprocedures' properties explicitly (and more generally, the separation between main loop and subprocedures) is a product of the formalisation.
We examine how the invariants and properties yield partial correctness.

\isathm{Mincost_Flow_Algorithms/Orlins.thy}{2807}{2814}
\begin{thm}[Partial Correctness\label{lemma:orlinspartialcorrect}]
Assume the algorithm terminates on an uncapacitated instance with conservative weights.
If it finds a flow, it is a min-cost flow, otherwise the problem is infeasible.
\end{thm}
\begin{proof}[Proof Sketch]
$b = 0$ yields immediate termination with $f\refline{call_augment_edges} = 0$ as correct result (from~\ref{P2}).

If $b \neq 0$, we show the invariants for the last state that we give as argument to \sendflowpar.
All invariants hold for the initialisation in the algorithm.
The pseudocode and~\ref{P2} imply preservation of Invariant~\ref{invargamma}.
Only \sendflowpar changes the flow along edges outside $\F$ in Line~\ref{call_augment_edges} (from~\ref{P4}) and the change is an integral multiple of $\gamma$ (from~\ref{P2a}) implying preservation of Invariant~\ref{invarintegral}.
The arguments for Invariants~\ref{invarsamecomp} and~\ref{invar7} are simple (from~\ref{P5a}).

Preservation of Invariant~\ref{invarforest} is more difficult.
Assume it holds in Line~\ref{call_augment_edges}.
We know $|b'\refline{call_augment_edges}(x)| \leq (1 - \epsilon) \cdot \gamma\refline{call_augment_edges}$ for any $x$ (from~\ref{P2}).
After modifying $\gamma$ (Lines~\ref{line24} - \ref{line27}), the flow along forest edges is above $8n\gamma\refline{line27}$, $\Phi(b'\refline{line27}, \gamma\refline{line27}) \leq n$ and $\forall x. \, |b'\refline{line27}(x)| < 2 \cdot \gamma\refline{line27}$.
Executing \maintainforestpar can cause a decrease of $2n\gamma\refline{call_maintain_forest}$ for forest edges (from~\ref{P4}) and an increase in $\Phi$ by at most $n$ (\ref{P51}).
At this point, we have $f\refline{call_maintain_forest}(e) > 6n\gamma\refline{call_maintain_forest}$ for all $e \in \F\refline{call_maintain_forest}$ and $\Phi(state\refline{call_maintain_forest})\leq 2n$.
\sendflowpar in Line~\ref{call_augment_edges} is responsible for a further flow decrease of at most $2n\gamma\refline{call_maintain_forest}$ (from~\ref{P52}).
The overall decrease along forest edges was at most $4n\gamma\refline{call_maintain_forest}$, hence again $f\refline{call_augment_edges}(e)>4n\gamma\refline{call_augment_edges}$ for $e\in \F\refline{call_augment_edges}$ in the next iteration.
This implies preservation of Invariant~\ref{invarforest}.

By~\ref{P6} we obtain preservation of Invariant~\ref{invaropt} as well.
The invariants hold for the last $state\refline{call_maintain_forest}$ before termination (state before the last usage of \sendflowpar).
\ref{P1} gives the claim.
\end{proof}

This proof refines proofs in the literature~\cite{KorteVygenOptimisation} and it was reconstructed from the formalisation.

\begin{rem}
We restrict ourselves to infinite edge capacities, which is insignificant because we can reduce problems to the uncapacitated setting with a linear increase in input length, as we will see later.
Our theorems here require weight-conservativity which we inherit from Algorithm~\ref{algo:SSP} and Algorithm~\ref{algo:scling} as a constraint.
However, we drop the restriction to integral capacities and balances.
\end{rem}

\paragraph{Formalisation} We formalise Algorithm~\ref{algo:orlins} as the function \torlinsfun in a locale \torlinsloc of the same name.
Since the properties of the subprocedures \sendflowpar and \maintainforestpar are quite complex (unlike the selection of paths), and since implementations of the subprocedures are not useful on their own, we do not specify them at the beginning of the locale \torlinsloc to be instantiated later.
Instead, we set up two locales for the subprocedures (details later) and define \torlinsloc, the locale for the top level loop, as their conjunction.

\begin{rem}
We present Orlin's algorithm in a top-down manner to make it easier to understand, but the formalisation was developed bottom-up, i.e.\ we first formalised \sendflowpar and \maintainforestpar.
The top-down approach with specifying subprocedures in a locale and instantiating them later is only practicable when the subprocedures are meaningful on their own and their properties are comparably simple.
\end{rem}

As in the case of \scaling and \SSP, we use records to model the program state.
This time, however, the variables of the program state are abstract data types (ADTs), e.g.\ maps from vertices or edges to reals to model the flow or balances.
To specify these ADTs, we reuse the locales from Listings~\ref{isabelle:AD_sets} and~\ref{isabelle:AD_maps} multiple times as a part of the locales for Orlin's algorithm, subject to some extensions.
An example of an ADT operation used by \torlinsfun is \areall, which we use to check all elements of the data structure that models the active edges for a certain predicate.
Operations like \abstractflowmap turn a mapping data structure \tttm into a mathematical function $f_{\tttm}$.
When \tttm is undefined for a certain input, $f_{\tttm}$ returns a reasonable default value, e.g.\ $0$ for reals.

In addition to those variables mentioned above, the state contains some auxiliary ones that ease the implementation.
We will give details in Section~\ref{section:loopa}, which deals with \maintainforestpar where these variables are relevant.
We also encapsulate the data structure invariants and some conditions on the domains, e.g.\ that the domain of the flow map \currentflow \state is $\E$, into an invariant \implementationinvar that the loops preserve.

Invariants are again boolean predicates on states.
We formalise Invariants~\ref{invarintegral} and~\ref{invaropt} as \invarintegralorl and \invarisoptflow, respectively.
There is a number of conceptually less interesting or even trivial properties, e.g.\ $\F \subseteq \E$, or that $r$ is a valid representative function, i.e.\ $r(u) = r(v)$ iff $u$ and $v$ are reachable from one another in $\F$.
In the formalisation, we encapsulate these into the auxiliary invariant \underlyinginvars.

The Isabelle/HOL function \torlinsfun first decreases $\gamma$, followed by \tmaintainforest and \tsendflow, contrary to the pseudocode from Algorithm~\ref{algo:orlins}.
We enter the loop body with the decrease of $\gamma$ while preserving the relative order of the operations
so that the invariants refer exclusively to the beginning/end of iterations.\footnote{
We formalise a semantically equivalent program: suppose, we have a loop whose body performs the operations $c_1$, $c_2$ and $c_3$, then we can transform it into $c_1$ followed by a loop with the body $c_2$, $c_3$ and $c_1$.
For the semantic equivalence, some conditions, e.g.\ loop condition true after $c_1$, need to be satisfied.
Here, this is trivially the case because the loop condition is $True$.
}
It is then much easier to formally state preservation lemmas about Invariant~\ref{invarforest}, which talks about the state after finishing Line~\ref{call_augment_edges}, in the formalisation.
Because of this, we write valid executions of Orlin's algorithm as $\torlinsfun \;\tlpar\tsendflow \;\initial\trpar$ in the formalisation.
Proving the formal counterpart of Theorem~\ref{lemma:orlinspartialcorrect} requires work that is discussed in the next subsections.

We now state the subprocedures and show that they satisfy the aforementioned properties, including both formal and informal proofs.

\subsection{Augmenting the Flow}
\label{section:loopb}
\newcommand{\pleasecontinue}{\textit{please continue}}
\newcommand{\sendflowspec}{\isainl{Mincost_Flow_Algorithms/Send_Flow.thy}{9}{23}{send-flow-spec}\xspace}
\newcommand{\sendflowproof}{\isainl{Mincost_Flow_Algorithms/Send_Flow.thy}{661}{691}{send-flow}\xspace}
\newcommand{\getsourcetargetpatha}{\isainl{Mincost_Flow_Algorithms/Send_Flow.thy}{17}{18}{get-source-target-path-a}\xspace}
\newcommand{\getsourcetargetpathacond}{\isainl{Mincost_Flow_Algorithms/Send_Flow.thy}{509}{514}{get-source-target-path-a-cond}\xspace}
\newcommand{\getsourcetargetpathaxioms}{\isainl{Mincost_Flow_Algorithms/Send_Flow.thy}{664}{668}{get-source-target-path-a-axioms}\xspace}
\newcommand{\sendflowcallonecondPhidecr}{\isainl{Mincost_Flow_Algorithms/Send_Flow.thy}{1333}{1337}{send-flow-call1-cond-Phi-decr}\xspace}
\newcommand{\sendflowcalltwocondPhidecr}{\isainl{Mincost_Flow_Algorithms/Send_Flow.thy}{1441}{1445}{send-flow-call2-cond-Phi-decr}\xspace}
\newcommand{\sendflowsimps}{\isainl{Mincost_Flow_Algorithms/Send_Flow.thy}{370}{377}{send-flow-simps}\xspace}
\newcommand{\sendflowinduct}{\isainl{Mincost_Flow_Algorithms/Send_Flow.thy}{415}{420}{send-flow-induct}\xspace}
\newcommand{\sendflowPhi}{\isainl{Mincost_Flow_Algorithms/Send_Flow.thy}{2567}{2574}{send-flow-Phi}\xspace}
\newcommand{\invarabovesixNgamma}{\isainl{Mincost_Flow_Algorithms/Orlins_Preparation.thy}{1254}{1256}{invar-above-6Ngamma state}\xspace}
\newcommand{\orlinsentry}{\isainl{Mincost_Flow_Algorithms/Orlins_Preparation.thy}{1180}{1181}{orlins-entry}\xspace}
\newcommand{\orlinsentryaftersendflow}{\isainl{Mincost_Flow_Algorithms/Send_Flow.thy}{2496}{2504}{orlins-entry-after-send-flow}\xspace}
\newcommand{\sendflowinvarisOptpres}{\isainl{Mincost_Flow_Algorithms/Send_Flow.thy}{2621}{2627}{send-flow-invar-isOpt-pres}\xspace}
\newcommand{\sendflowcompleteness}{\isainl{Mincost_Flow_Algorithms/Send_Flow.thy}{2824}{2832}{send-flow-completeness}\xspace}

{\renewcommand{\theAlgoLine}{SF\arabic{AlgoLine}}
\setlength{\algomargin}{0.85cm}
\SetAlCapHSkip{0pt}
\SetAlgoHangIndent{0pt}
\SetInd{0.5em}{0.5em}
\SetVlineSkip{0.3mm}
\begin{algorithm}[t]
\caption{$\sendflow(b', f, \F, actives, \gamma)$}
\label{algo:loopB}
 \While{True}{
 \textbf{if} $\forall v \in \V. \; b'(v) = 0$ \textbf{then} \textbf{return} $(b', f, \asuccess)$; \label{line6}\\
  \uElseIf{$\exists \, s \,. \; b'(s) > (1 - \epsilon) \cdot \gamma$\label{line8}}{
    \uIf{$\exists \, t \, . \; b'(t) < - \epsilon \cdot \gamma \wedge t \text{ is reachable from } s$\label{line9}}
        {
take such $s$, $t$, and a connecting path $P$ with original edges from $actives \cup \F$;\label{line10}\\
augment $f$ along $P$ from $s$ to $t$ by $\gamma$;\label{line11}\\
$b'(s) \leftarrow b'(s) - \gamma$; $b'(t) \leftarrow b'(t) + \gamma$\label{line12};
        }
        \textbf{else} \textbf{return} $(b', f, \ainfeasible)$;\label{line14}
   }
  \uElseIf{$\exists \, t\,. \; b'(t) < - (1 - \epsilon) \cdot \gamma$\label{line15}}{
    \uIf{$\exists \, s \,. \;b'(s) > \epsilon \cdot \gamma \wedge t \text{ is reachable from } s$\label{line16}}
        {
take such $s$, $t$, and a connecting path $P$ with original edges from $actives \cup \F$;\label{line17}\\
augment $f$ along $P$ from $s$ to $t$ by $\gamma$;\label{line18}\\
$b'(s) \leftarrow b'(s) - \gamma$; $b'(t) \leftarrow b'(t) + \gamma$;\label{line19}
        }
        \textbf{else} \textbf{return} $(b', f, \ainfeasible)$;\label{line21}
   }\label{line22}
  \textbf{else} \textbf{return} $(b', f, \pleasecontinue)$;\label{line23}
  }
\end{algorithm}
}
We now argue why \sendflowpar (Algorithm~\ref{algo:loopB}) satisfies the properties stated in the previous section.
\ref{P2}, \ref{P2a}, \ref{P5a}, \ref{P6}, \ref{P52} and~\ref{P1} assert something about \sendflowpar.
We can infer~\ref{P2} from the subprocedure's structure.
The only possible amount to augment is $\gamma$ which yields~\ref{P2a}.
\ref{P5a} (Preservation of Invariant~\ref{invar7}) also follows from the structure.

\begin{proof}[Proof Sketch for~\ref{P6}]
We assume the invariants to hold for $state = (b', f, \F, actives, \gamma)$.
Then $f\refline{line12}$ and $f\refline{line19}$ are minimum-cost flows for the balances $b - b'\refline{line12}$ and $b - b'\refline{line19}$, respectively, because of Theorem~\ref{thm911}
and the fact that  $f\refline{line6}$ is a minimum-cost flow for the balance $b - b'\refline{line6}$ (Invariant~\ref{invaropt}).
The restriction to active and forest edges (Lines~\ref{line10} and~\ref{line17}) is unproblematic:
simulate deactivated edges with forest paths which are of minimum-costs (see Lemma~\ref{lemma:optforestpath} in the next section where we discuss the forest).
\end{proof}

\begin{lem}\label{lem:phiB}
Assume $f(e) > 6n\gamma$ for all $e \in \F$.
Any iteration of \sendflowpar decreases $\Phi$ at least by $1$:
\begin{itemize}
\item \isaitem{Mincost_Flow_Algorithms/Send_Flow.thy}{1333}{1337}$\Phi(b'\refline{line11}, \gamma) - \Phi(b'\refline{line12}, \gamma)\geq 1$
\item \isaitem{Mincost_Flow_Algorithms/Send_Flow.thy}{1441}{1445}$\Phi(b'\refline{line18}, \gamma) - \Phi(b'\refline{line19}, \gamma)\geq 1$
\end{itemize}
\end{lem}

\begin{proof}[Proof Sketch for~\ref{P52}]
Assume $f(e) > 6n\gamma$ for all $e \in \F$.
We have the decrease in $\Phi$ according to Lemma~\ref{lem:phiB}.
Because of Lines~\ref{line8} and~\ref{line15}, $\Phi$ cannot become negative.
Thus, \sendflow$(state)$ performs at most $\Phi(state)$ iterations
and the flow decrease for an edge is at most $\Phi (state) \cdot \gamma$.
Note that during the whole computation, $f(e) > 0$ for all $e \in \F$ meaning that we can always select a path using active and forest edges, as required in Lines~\ref{line10} and~\ref{line17}.
The proof of a strict decrease in $\Phi$ only works for $\epsilon > 0$.
\end{proof}

\begin{proof}[Proof Sketch for~\ref{P1}]
If the algorithm found a flow, then $\forall v \in \V \, . \; b'\refline{line6}(v) = 0$ (Line~\ref{line6}).
Together with preservation of the optimality invariant, it gives the first subclaim.

If the algorithm asserts infeasibility, one can exploit information about $b'$ from Lines~\ref{line8}, \ref{line9}, \ref{line15} and~\ref{line16}.
By employing residual cuts and the analogous definition where every direction is reversed, one can show infeasibility for both cases.
The proof is similar to that of Theorem~\ref{SSPtotalcorrect}.
The argument only works for $\epsilon \leq \frac{1}{n}$.
\end{proof}

By this, we have shown that \sendflowpar satisfies all asserted properties.

\begin{rem}
We also saw the restriction $0 < \epsilon \leq \frac{1}{n}$.
One might think that $\epsilon$ could be assigned to $0$.
As we shall see later on, that would make it impossible for us to derive the worst-case running time bound.
In essence, we need to allow vertices to be processed if they are 'slightly' below the threshold,
otherwise the algorithm might take exponentially many iterations, and its running time will depend on $B$.
Interested readers should consult Section 4.2 of Orlin's original paper~\cite{OrlinScalingFlow}.
\end{rem}

\paragraph*{Formalisation}
We formalise the loop in a locale \sendflowspec that assumes selection functions, e.g.\ \getsourcetargetpatha to find the path required in Line~\ref{line10} of \sendflowpar.
We set up the customised \sendflowsimps and \sendflowinduct for simplification and induction.
We also introduce definitions to encapsulate the preconditions ensuring that the selection functions work properly, e.g.\ \getsourcetargetpathacond for \getsourcetargetpatha.
Among others, this condition requires edges in $\F$ to have positive flow.
This is necessary to ensure the existence of an augmenting path that uses active and forest edges only, i.e.\ no deactivated ones, as required in Line~\ref{line10} of the algorithm.
We use those conditions to specify the behaviour of the selection functions in the assumptions of the locale \sendflowproof, e.g.\ \getsourcetargetpathaxioms.
Since the preconditions like \getsourcetargetpathacond are of considerable complexity, it is necessary to encapsulate them in definitions, without which the locale \sendflowproof would barely be readable.
Whenever we use a function, its application precondition is provable allowing us to retrieve information about the result.

In the locale \sendflowproof, we formalise the proofs that we informally discussed in this section.
\sendflowcallonecondPhidecr formalises that taking the first recursive step in \tsendflow (Lines~\ref{line10}-\ref{line12} in Algorithm~\ref{algo:loopB}) results in $\Phi$ being decreased by at least $1$; \sendflowcalltwocondPhidecr is its analogue for the second recursive step.
\sendflowPhi is the formal counterpart to~\ref{P52} where \invarabovesixNgamma asserts $f(e) > 6n\gamma$ for all $e \in \F$ in \state.  Preservation of this invariant implies positive flow along forest edges (necessary for Lines~\ref{line10} and~\ref{line17} and precondition for \sendflowcallonecondPhidecr).
\orlinsentryaftersendflow proves a part of~\ref{P2} where \orlinsentry \state asserts that $|b'(v)| \leq (1-\epsilon)\gamma$ in \state.
\ref{P6} is proven by \sendflowinvarisOptpres, and \sendflowcompleteness belongs to the formal counterpart of~\ref{P1}, where \tttb refers to the balances of the actual flow problem.

\subsection{Maintaining the Forest}
\label{section:loopa}
\newcommand{\maintainforestspec}{\isainl{Mincost_Flow_Algorithms/Maintain_Forest.thy}{8}{16}{maintain-forest-spec}\xspace}
\newcommand{\getpath}{\isainl{Mincost_Flow_Algorithms/Maintain_Forest.thy}{16}{16}{get-path}\xspace}
\newcommand{\getpathaxioms}{\isainl{Mincost_Flow_Algorithms/Maintain_Forest.thy}{162}{164}{get-path-axioms}\xspace}
\newcommand{\maintainforestgetpathcond}{\isainl{Mincost_Flow_Algorithms/Maintain_Forest.thy}{19}{23}{maintain-forest-get-path-cond}\xspace}
\newcommand{\maintainforesttt}{\isainl{Mincost_Flow_Algorithms/Maintain_Forest.thy}{156}{164}{maintain-forest}\xspace}
\newcommand{\convtoredge}{\isainl{Mincost_Flow_Algorithms/Orlins_Preparation.thy}{26}{26}{conv-to-rdg}\xspace}
\newcommand{\tttv}{\texttt{v}\xspace}
\newcommand{\notblocked}{\isainl{Mincost_Flow_Algorithms/Orlins_Preparation.thy}{31}{31}{not-blocked}\xspace}
\newcommand{\notblockedupdall}{\isainl{Mincost_Flow_Algorithms/Orlins_Preparation.thy}{342}{343}{not-blocked-upd-all}\xspace}
\newcommand{\actives}{\isainl{Mincost_Flow_Algorithms/Orlins_Preparation.thy}{27}{27}{actives}\xspace}
\newcommand{\filter}{\isainl{Mincost_Flow_Algorithms/Orlins_Preparation.thy}{35}{35}{filter}\xspace}
\newcommand{\repcompcard}{\isainl{Mincost_Flow_Algorithms/Orlins_Preparation.thy}{30}{30}{rep-comp-card}\xspace}
\newcommand{\repcompupdall}{\isainl{Mincost_Flow_Algorithms/Orlins_Preparation.thy}{336}{337}{rep-comp-upd-all}\xspace}
\newcommand{\invarFone}{\isainl{Mincost_Flow_Algorithms/Orlins_Preparation.thy}{1193}{1195}{invar-F1}\xspace}
\newcommand{\invarFtwo}{\isainl{Mincost_Flow_Algorithms/Orlins_Preparation.thy}{1215}{1218}{invar-F2}\xspace}
\newcommand{\invarFtwopres}{\isainl{Mincost_Flow_Algorithms/Maintain_Forest.thy}{2779}{2787}{invar-F2-pres}\xspace}
\newcommand{\outsideFnoflowchange}{\isainl{Mincost_Flow_Algorithms/Maintain_Forest.thy}{2944}{2947}{outside-F-no-flow-change}\xspace}
\newcommand{\monoonestepphi}{\isainl{Mincost_Flow_Algorithms/Maintain_Forest.thy}{2685}{2688}{mono-one-step-phi}\xspace}
\newcommand{\Phiincrease}{\isainl{Mincost_Flow_Algorithms/Maintain_Forest.thy}{2850}{2853}{Phi-increase}\xspace}
\newcommand{\PhiincreasebelowN}{\isainl{Mincost_Flow_Algorithms/Maintain_Forest.thy}{2879}{2881}{Phi-increase-below-N}\xspace}
\newcommand{\component}{\texttt{component}\xspace}
\newcommand{\comps}{\texttt{comps}\xspace}
\newcommand{\tttE}{\texttt{E}\xspace}
\newcommand{\tttV}{\texttt{V}\xspace}
\newcommand{\FF}{\ensuremath{\mathfrak{F}}}

\begin{figure}
\centering
\trimbox{0mm 0mm -2mm 0mm}{
\begin{tikzpicture}[
      mycircle/.style={
         circle,
         draw=black,
         fill=gray,
         fill opacity = 0.3,
         text opacity=1,
         inner sep=0pt,
         minimum size=12pt,
         font=\normalsize},
      myarrow/.style={-Stealth},
      node distance=0.7cm and 1.3cm
      ]
      \node[label=below:{\scriptsize \color{blue} $b_x$\color{black}, \color{purple} $x$}, mycircle] (x) {$x$} ;
      \node[label=below:{\scriptsize \color{blue} $b_y$\color{black}, \color{purple} $y$}, mycircle,above right=of x] (y) {$y$};
      \node[label=below:{\scriptsize \color{blue} $b_z$\color{black}, \color{purple} $z$}, mycircle,below right=of y] (z) {$z$};

\draw[myarrow] (x) -- node[sloped, font = \scriptsize, below] {$9n\gamma$} (y);
\draw[myarrow] (y) -- node[sloped, font = \scriptsize, below] {$9n\gamma$} (z);
    \end{tikzpicture}
    }
    \trimbox{0mm 0mm -2mm 0mm}{
    \begin{tikzpicture}[
      mycircle/.style={
         circle,
         draw=black,
         fill=gray,
         fill opacity = 0.3,
         text opacity=1,
         inner sep=0pt,
         minimum size=12pt,
         font=\normalsize},
      myarrow/.style={-Stealth},
      node distance=0.7cm and 1.3cm
      ]
      \node[label=below:{\scriptsize \color{blue} $b_x$\color{black}, \color{purple} $x$}, mycircle] (x) {$x$} ;
      \node[label=below:{\scriptsize \color{blue} $0$\color{black}, \color{purple} $z$}, mycircle,above right=of x] (y) {$y$};
      \node[label=below:{\scriptsize \color{blue} $b_y+b_z$\color{black}, \color{purple} $z$}, mycircle,below right=of y] (z) {$z$};

\draw[-{Stealth}, red, line width=0.5mm] (y) -- node[sloped,  font = \scriptsize, above, color=black] {$9n \gamma + b_y$} (z);
\draw[myarrow] (x) -- node[sloped, font = \scriptsize, below] {$9n\gamma$} (y);
    \end{tikzpicture}
    }
    \trimbox{0mm 0mm 0mm 0mm}{
     \begin{tikzpicture}[
      mycircle/.style={
         circle,
         draw=black,
         fill=gray,
         fill opacity = 0.3,
         text opacity=1,
         inner sep=0pt,
         minimum size=12pt,
         font=\normalsize},
      myarrow/.style={-Stealth},
      node distance=0.7cm and 1.3cm
      ]
      \node[label=below:{\scriptsize \color{blue} $0$\color{black}, \color{purple} $z$}, mycircle] (x) {$x$} ;
      \node[label=below:{\scriptsize \color{blue} $0$\color{black}, \color{purple} $z$}, mycircle,above right=of x] (y) {$y$};
      \node[label=below:{\scriptsize \color{blue} $b_x+b_y+b_z$\color{black}, \color{purple} $z$}, mycircle,below right=of y] (z) {$z$};

\draw[-{Stealth}, red, line width=0.5mm] (y) -- node[sloped,  font = \scriptsize, above, color=black] {$9n\gamma + b_x+b_y$} (z);
\draw[-{Stealth}, red, line width=0.5mm] (x) -- node[sloped,  font = \scriptsize, above, color=black] {$9n\gamma + b_x$} (y);
    \end{tikzpicture}
    }
   \vspace*{-2mm}
\caption{
Merging Forest Components:
remaining balances $b'$, forest edges, flow values and representatives are blue, red, black and purple, respectively.
Assume the balances are non-negative.
Initially there are three components: $\{x\}$, $\{y\}$, and $\{z\}$.
Then we merge $\{y\}$ and $\{z\}$ into one component, where $z$ is the representative, inheriting the balance of both $y$ and $z$.
Then we add $x$ to the component, and $z$ is the representative of the final component.
To make sure that the new balances are consistent, we move flow along the red edges to/from the representative vertex.
The reduction in flow in each edge is bounded.}
\label{fig:merges}
\end{figure}
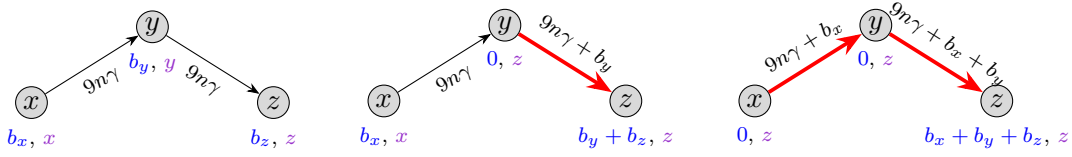
{
\setlength{\algomargin}{1.1cm}
\renewcommand{\theAlgoLine}{MF\arabic{AlgoLine}}
\begin{algorithm}[t]
\SetAlCapHSkip{0pt}
\SetAlgoHangIndent{0pt}
\SetInd{0.5em}{0.5em}
\SetVlineSkip{0.3mm}
\setlength{\algomargin}{10pt}
\caption{$\maintainforest(b', f, \F, actives, \gamma, r)$\label{algo:loopA}}
 \While{$\exists \, e = (x, y).\; e \in \aactives \wedge e \not \in \F \wedge f(e) > 8 n \gamma$ \label{line28}}
 {
  $\F \leftarrow \F \cup \lbrace e\rbrace$;\label{line28a}\\
  \textbf{if} $|\F\text{-component of } y| < |\F\text{-component of } x|$ \textbf{then} exchange $x$ and $y$;\label{line30} \\
  let $x' = r(x)$ and $y' = r(y)$ be the respective representatives;\label{line29}\\
  take residual path $Q \subseteq \FF$ connecting $x'$ and $y'$;\label{line31}\\
  \textbf{if} $b'(x') > 0$
  \textbf{then} augment $f$ along $Q$ by $b'(x')$ from $x'$ to $y'$;\label{line33}\\
  \textbf{else} augment $f$ along $\overset{_{\longleftarrow}}{Q}$ by $- b'(x)$ from $y'$ to $x'$;\label{line35}\\
   $b'(y') \leftarrow b'(y') + b'(x')$; $b'(x') \leftarrow 0$;\label{line36}\\
   \textbf{foreach} \textit{$v$ reachable from $y'$ in $\FF$}
   \textbf{do} $r(v) \leftarrow y'$ \label{line40};\\
   \textbf{foreach} $d = (u, v).  \; d \in \aactives \wedge \lbrace r(u), \, r(v) \rbrace = \lbrace y' \rbrace$
   \textbf{do} \aactives $\leftarrow$ \aactives$\setminus \lbrace d \rbrace$;\label{line37}\\
 }
 \textbf{return} $(b', f, \F, actives, r)$;
\end{algorithm}
}
We now discuss the last major part of the algorithm, namely, maintaining the forest.
Algorithm~\ref{algo:loopA} shows the definition of \maintainforestpar.
We add active edges with flow above $8n\gamma$ to $\F$,
which inevitably changes the connected components of the forest -- we create a new component for every added edge by merging the components
to which the two endpoints of the edge belong.
Since non-zero balance is only allowed for representatives,
we must re-concentrate balances at one of the two representatives after merging two components.\footnote{Cardinalities of forest components in the pseudocode refer to the number of vertices in the component.}
Moreover, we deactivate all edges between them and update the representatives.
For a forest $\F$, $\FF$ is the corresponding residual network consisting only of forest edges.
We perform the re-concentration by augmenting along paths in $\FF$.
Figure~\ref{fig:merges} displays an example of how balances, flow and forest change.

\begin{rem}
The high-level reason for merging components and concentrating the balances at unique representatives is the controlled reduction of the number of non-zero vertices, resulting in a simpler problem.
Only non-zero vertices need further treatment (by augmentations) to receive or emit more flow.
The contraction of vertices is the foundation of several strongly polynomial time algorithms for the minimum-cost flow problem.
The version by Korte and Vygen~\cite{KorteVygenOptimisation} that we follow simulates the contractions by maintaining the forest $\F$ and component representatives.
\end{rem}

Only~\ref{P5a} - \ref{P6} make assertions about \maintainforestpar and we argue why they are indeed satisfied.
\ref{P5a} (preservation of Invariants~\ref{invarsamecomp} and~\ref{invar7}) is easy to see since it is precondition for deactivation to have the same representatives (Line~\ref{line37}).
\begin{proof}[Proof Sketch for~\ref{P4}]
Invariants~\ref{invara1} and~\ref{invara2} bound the forest edges' flow decrease:
\begin{enumerate}[label=(IMF-\theenumi)]
\renewcommand{\theenumi}{IMF-\arabic{enumi}}
\item For any $x$, the product of its $\F\refline{line28}$-component's cardinality and $\beta$ is a bound for $|b'\refline{line28}(x)|$. \label{invara1}
\item For any $e \in \F\refline{line28}$, we have $f\refline{line28}(e) > \alpha - \beta \cdot |X|$ where $X$ is the $\F\refline{line28}$-component of $e$.\label{invara2}
\end{enumerate}
\renewcommand{\theenumi}{\arabic{enumi}}
where $\alpha$ and $\beta$ are initial bounds for the flow and the balances, respectively.
The conjunction of~\ref{invara1} and~\ref{invara2} is preserved by \maintainforestpar.
\end{proof}
\begin{rem}
It is important to always concentrate the balances at the larger component's representative as it is done in the algorithm. Otherwise, the preservation proof for invariant~\ref{invara2}, which is omitted for the sake of brevity, would not work.
\end{rem}

Any iteration merges two $\F$-components making $n - 1$ an upper bound for the number of iterations.
\ref{P51} asserts $\Phi(\maintainforest(state)) \leq \Phi(state) + n$.
This holds because the potential cannot increase by more than $1$ per iteration, as one can see in the following lemma.
\isathm{Mincost_Flow_Algorithms/Maintain_Forest.thy}{2685}{2688}
\begin{lem}\label{phiA}
$\Phi$'s increase during a single iteration of \maintainforestpar is at most $1$.
\end{lem}
\begin{proof}\let\qed\relax
Only Line~\ref{line36} is relevant for the potential since it is the only place where $b'$ is changed.
We look at $state\refline{line35}$ (state before changing $b'$) and $state\refline{line36}$ (state after changing $b'$).
It follows:
{\allowdisplaybreaks
\begin{align*}
\Phi(state\refline{line36}) & = \sum \limits_{v \in \V}
   \Bigl \lceil \frac{|b'\refline{line36}(v)|}{\gamma} -(1 -\epsilon) \Bigr \rceil
\\
& = \sum \limits_{v \in \V \setminus \lbrace x', y'\rbrace}
   \Bigl \lceil \frac{|b'\refline{line35}(v)|}{\gamma} -(1 -\epsilon) \Bigr \rceil +
   \Bigl \lceil \frac{0}{\gamma} -(1 -\epsilon) \Bigr \rceil \\ & +
   \Bigl \lceil \frac{|b'\refline{line35}(y') + b'\refline{line35}(x')|}{\gamma} -(1 -\epsilon) \Bigr \rceil
\\
 & =
 \sum \limits_{v \in \V \setminus \lbrace x', y'\rbrace}
 \Bigl \lceil \frac{|b'\refline{line35}(v)|}{\gamma} -(1 -\epsilon) \Bigr \rceil +
   \Bigl \lceil \frac{|b'\refline{line35}(y')| + |b'\refline{line35}(x')|}{\gamma} -(1 -\epsilon) \Bigr \rceil
\\
   & \leq
 \sum \limits_{v \in \V \setminus \lbrace x', y'\rbrace}
 \Bigl \lceil \frac{|b'\refline{line35}(v)|}{\gamma} -(1 -\epsilon) \Bigr \rceil +
   \Bigl \lceil \frac{|b'\refline{line35}(x')|}{\gamma}\Bigr \rceil+
    \Bigl \lceil \frac{|b'\refline{line35}(y')|}{\gamma} -(1 -\epsilon) \Bigr \rceil
\\
    & \leq
 \sum \limits_{v \in \V \setminus \lbrace x', y'\rbrace}
 \Bigl \lceil \frac{|b'\refline{line35}(v)|}{\gamma} -(1 -\epsilon) \Bigr \rceil +
   \Bigl \lceil \frac{|b'\refline{line35}(x')|}{\gamma} - (1 - \epsilon)\Bigr \rceil \\ &+ 1 +
    \Bigl \lceil \frac{|b'\refline{line35}(y')|}{\gamma} -(1 -\epsilon) \Bigr \rceil \\
    & = \Phi(state\refline{line35}) + 1 \hspace*{9.5cm} \qedsymbol
\end{align*}
}
\end{proof}

\begin{lem}\label{lemma:optforestpath}
If all invariants hold in Line~\ref{line28a}, both
\begin{itemize}
\item \isaitem{Mincost_Flow_Algorithms/Maintain_Forest.thy}{2116}{2116}  the path $Q$
\item \isaitem{Mincost_Flow_Algorithms/Maintain_Forest.thy}{2028}{2028}  and its reversal $\overset{_{\longleftarrow}}{Q}$
\end{itemize}
are minimum-cost paths.
\end{lem}
\begin{proof}
Suppose this were not true.
There is $P$ with $\mathfrak{c}(P) < \mathfrak{c}(Q)$ connecting the same vertices.
Since $\mathfrak{c}(\overset{_{\longleftarrow}}{Q}) = - \mathfrak{c}(Q)$, $P\overset{_{\longleftarrow}}{Q}$ is a cycle with $\mathfrak{c}(P\overset{_{\longleftarrow}}{Q}) =
\mathfrak{c}(P) +\mathfrak{c}(\overset{_{\longleftarrow}}{Q}) < 0$
and $\mathfrak{u}_{f}(P\overset{_{\longleftarrow}}{Q})> 0$ contradicting Theorem~\ref{optcrit} and Invariant~\ref{invaropt}.
The case for $\overset{_{\longleftarrow}}{Q}$ is analogous.
\end{proof}

\begin{proof}[Proof of~\ref{P6}]
To apply Theorem~\ref{thm911} for Lines~\ref{line33} and~\ref{line35}, we need optimality of the forest paths $Q$ and $\overset{_{\longleftarrow}}{Q}$ that are used for augmentations.
This follows from Lemma~\ref{lemma:optforestpath}.
\end{proof}

This concludes our proofs that \maintainforestpar satisfies properties~\ref{P5a} - \ref{P6}.

\paragraph*{Formalisation}
In the loop function \tmaintainforest, there is only one recursive case.
Also, the specification locale \maintainforestspec fixes only one path selection function \getpath.
We encapsulate the preconditions for specifying its behaviour into a definition \maintainforestgetpathcond, similar to \getsourcetargetpathacond for \sendflowpar, and use them to state the assumptions \getpathaxioms of the locale \maintainforesttt.
An important part of those preconditions is the auxiliary and implementation invariants.

The forest is specified as an ADT for an adjacency map~\cite{graphAlgosBook}.
We add an undirected edge to the forest by adding two edges in opposite directions to the representing adjacency map, making the latter symmetric.
For connected components of $\F$, we reuse Abdulaziz and Mehlhorn's~\cite{BlossomIsabelleJAR} work.
We assume the function \getpath takes $\F$ as an adjacency map and returns a path as a sequence of vertices.
We must transform these into paths over residual edges to perform an augmentation.
We have a program variable \convtoredge, again specified as an ADT for maps, that takes a pair of vertices $(u, v)$ and returns a positive-capacity residual edge in $\FF$ pointing from $u$ to $v$.
\convtoredge \state is always updated in a way such that residual edges returned by $\convtoredge \;\state \;\tlpar\tttu, \tttv\trpar$ and $\convtoredge \;\state \;\tlpar\tttv, \tttu\trpar$ originate from the same graph edge.
Due to this, $Q$ and $\overset{_{\longleftarrow}}{Q}$ have opposite costs as needed to show Lemma~\ref{lemma:optforestpath}.

We deactivate edges in Line~\ref{line37} of \maintainforestpar if they are active and if their endpoints' representatives are within the same newly created forest component.
The other loop \sendflowpar requires paths that use active and forest edges only.
To facilitate the computation of those paths later, we have the program variable \notblocked which is a map (specified as ADT, of course) from graph edges to booleans.
\notblocked \state \ttte being true means that the edge \ttte is not deactivated, i.e.\ it is either a forest or an active edge.
In the formalisation, we assume the function \notblockedupdall in the respective ADT specification.
We use it to iterate over the whole domain of \notblocked \state (which is $\E$):
if for an edge $e \in \E$ with endpoints $x_1$ and $x_2$, $\lbrace r(x_1), r(x_2) \rbrace = \lbrace x', y'\rbrace$ holds, we set \notblocked \state \ttte to \False.
We also remove these edges from \actives \state with a function \filter, which comes from the ADT for \actives.
Like the list function of the same name, it iterates over all elements and keeps only those that satisfy a predicate.

To make the comparison in Line~\ref{line30} executable in constant time, we maintain cardinalities of each vertex's forest component.
Since we have to update these in the same way as the representatives, we combine both into a program variable \repcompcard which maps vertices to a pair of their representatives and $\F$-component cardinalities.
Updates to this variable, which are required by Line~\ref{line40} of \maintainforestpar, are realised as follows:
We only have to change representatives and component cardinalities for vertices whose $\F$-component grew by adding $e$.
For a vertex $v$, this is the case iff its old representative was $x'$ or $y'$.
We use an iteration function \repcompupdall (specified by the respective ADT, similar to \notblockedupdall) to iterate over the domain (set of vertices $\V$), to perform the required updates. The new representative is $y'$ and the cardinality value is $|C_x| + |C_y|$ where $C_x$ and $C_y$ are the $\F$-components of $x$ and $y$ respectively.

Invariants~\ref{invara1} and~\ref{invara2} are represented  by \invarFone and \invarFtwo, respectively, whose preservation (\invarFtwopres) is needed for~\ref{P4}.
Its second part, i.e.\ that the flow outside the forest does not change, is \outsideFnoflowchange.
Since the program modifies all variables, the preservation proofs for the invariants, even for the auxiliary invariant \underlyinginvars, which talks about simple relationships among the program variables, are more verbose than in case of \sendflowpar.

\subsection{Termination}

\label{sec:term}
\newcommand{\orlinsonestep}{\isainl{Mincost_Flow_Algorithms/Orlins.thy}{79}{85}{orlins-one-step}\xspace}
\newcommand{\orlinsonestepcheck}{\isainl{Mincost_Flow_Algorithms/Orlins.thy}{87}{91}{orlins-one-step-check}\xspace}
\newcommand{\tflag}{\texttt{flag}\xspace}
\newcommand{\ttteq}{\texttt{=}\xspace}
\newcommand{\success}{\texttt{success}\xspace}
\newcommand{\importantormergeortermination}{\isainl{Mincost_Flow_Algorithms/Orlins.thy}{1365}{1380}{important-or-merge-or-termination}\xspace}
\newcommand{\tttk}{\texttt{k}\xspace}
\newcommand{\ifimportantthencompincreaseortermination}{\isainl{Mincost_Flow_Algorithms/Orlins.thy}{1632}{1645}{if-important-then-comp-increase-or-termination}\xspace}
\newcommand{\tttz}{\texttt{z}\xspace}
\newcommand{\succfailtermsamedom}{\isainl{Mincost_Flow_Algorithms/Orlins.thy}{134}{137}{succ-fail-term-same-dom}\xspace}
\newcommand{\succfailtermsame}{\isainl{Mincost_Flow_Algorithms/Orlins.thy}{207}{210}{succ-fail-term-same}\xspace}
\newcommand{\finallytermination}{\isainl{Mincost_Flow_Algorithms/Orlins.thy}{2105}{2117}{finally-termination}\xspace}
\newcommand{\compow}{\texttt{compow}\xspace}
\newcommand{\ttti}{\texttt{i}\xspace}

For the inner loops, we can build a termination measure out of the number of forest components and the value of $\Phi$.
Now, call a vertex $v$ \textit{important} iff $|b'(v)| > (1 - \epsilon) \cdot \gamma$~\cite{KorteVygenOptimisation}, i.e.\ its contribution to $\Phi$ is positive.
We repeatedly wait for the occurrence of an important vertex and ensure a merge of two forest components some iterations later.
We define $\ell = \lceil \log (4 \cdot m \cdot n + (1 - \epsilon)) - \log \epsilon\rceil +1$ and $k = \lceil \log n\rceil + 3$.
\isathm{Mincost_Flow_Algorithms/Orlins.thy}{1632}{1645}
\begin{lem}[\cite{PlotkinTardosDualNetworkSimplex,KorteVygenOptimisation}\label{lemma:if_imporant_then}]
If there is an important vertex in an iteration, $\ell + 1$ additional iterations of the outer loop enforce its component being increased (or the algorithm terminates).
\end{lem}
\isathm{Mincost_Flow_Algorithms/Orlins.thy}{1365}{1380}
\begin{lem}[\cite{KorteVygenOptimisation}\label{lemma:wait_for_important}]
If we wait for $k + 1$ iterations, a vertex has become important (or the algorithm terminates).
\end{lem}
\isathm{Mincost_Flow_Algorithms/Orlins.thy}{2105}{2117}
\begin{thm}[Termination of Orlin's Algorithm]
\label{term3}
After at most $(n - 1) \cdot (k + {\ell} + 2)$ iterations of the outer loop, Orlin's algorithm will terminate.
\end{thm}
\begin{proof}[Proof Sketch]
\maintainforestpar and \sendflowpar terminate due to the decrease in the number of forest components or $\Phi$, respectively.
Obviously, there can be at most $n-1$ component merges, which gives the statement together with Lemmas~\ref{lemma:if_imporant_then} and~\ref{lemma:wait_for_important}.
Note that due to Invariant~\ref{invar7}, there can be at most one important vertex per component.
\end{proof}

\paragraph{Formalisation.}
The usual way to prove termination in Isabelle/HOL is to find a termination measure which decreases in every loop iteration.
This is not possible for the function \torlinsfun since we only know that after at most $k + {\ell} + 2$ iterations, the number of forest components must have decreased.
We introduce \orlinsonestep and \orlinsonestepcheck to express a single iteration of the loop body in \orlinspar.
We define \orlinsonestepcheck \state to execute \orlinsonestep \state if \tflag \state \ttteq \notyetterm.
Otherwise, i.e.\ if \tflag \state is \success or \infeasible, the state remains unchanged.
We iterate \orlinsonestepcheck (written as \compow \ttti \orlinsonestepcheck or \orlinsonestepcheck \char`\^\char`\^ \ttti) to express the state after \ttti iterations.
\footnote{If the algorithm terminates after $i$ iterations, we say for all $j \geq i$ that the state after $j$ iterations is the final state.}

\torlinsfun will terminate on the respective input state if the flag is \success or \infeasible after a strictly positive number of iterations (\succfailtermsamedom).
The companion \succfailtermsame states that the state reached after those iterations is exactly the one \torlinsfun returns.
We prove \finallytermination, which shows that the flag is finally \success or \infeasible, by strong induction over the number of forest components.
In the induction step, we apply the lemma \importantormergeortermination, which corresponds to Lemma~\ref{lemma:wait_for_important}, and obtain \tttk, followed by a case analysis:
(1) If the flag indicates termination in the next \tttk steps, we are done.
(2) If the number of components decreases, we can immediately apply the induction hypothesis.
(3) Otherwise, there is an important vertex \tttz, to which we apply \ifimportantthencompincreaseortermination, which corresponds to Lemma~\ref{lemma:if_imporant_then}, to obtain $\ell$.
(3.1) If the flag indicates termination in the next $\ell$ steps, we are done.
(3.2) Otherwise, \tttz's component is increased and we can apply the induction hypothesis.
\succfailtermsamedom and \finallytermination yield termination of \torlinsfun on the input state \initial.

To prove invariants for the function \torlinsfun, we deviate from our standard approach by proving the invariants for \orlinsonestepcheck \char`\^\char`\^ \ttti \state instead.
For the formal termination proof, we need that the invariants are satisfied after an arbitrary number of steps.
In combination with \succfailtermsamedom, we obtain the invariants for the final state for free which is why we can deviate from the standard approach of an induction over the actual function. \initialstateorlins combines the formalised results presented up to now into total correctness.

\subsection{Running Time}
\label{sec:running_time}
\newcommand{\torlinstime}{\isainl{Mincost_Flow_Algorithms/Orlins_Time.thy}{15}{24}{orlins-time}\xspace}
\newcommand{\tmaintainforesttime}{\isainl{Mincost_Flow_Algorithms/Maintain_Forest_Time.thy}{11}{56}{maintain-forest-time}\xspace}
\newcommand{\tsendflowtime}{\isainl{Mincost_Flow_Algorithms/Send_Flow_Time.thy}{12}{34}{send-flow-time}\xspace}
\newcommand{\torlinsonesteptime}{\isainl{Mincost_Flow_Algorithms/Orlins_Time.thy}{27}{32}{orlins-one-step-time}\xspace}
\newcommand{\torlinsonesteptimecheck}{\isainl{Mincost_Flow_Algorithms/Orlins_Time.thy}{34}{39}{orlins-one-step-time-check}\xspace}
\newcommand{\runningtimeinitial}{\isainl{Mincost_Flow_Algorithms/Orlins_Time.thy}{1725}{1734}{running-time-initial}\xspace}
\newcommand{\tlaminar}{\isainl{Mincost_Flow_Algorithms/Laminar_Family.thy}{7}{10}{laminar}\xspace}
\newcommand{\laminarfamilynumberofsets}{\isainl{Mincost_Flow_Algorithms/Laminar_Family.thy}{27}{29}{laminar-family-number-of-sets}\xspace}
\newcommand{\tSB}{\ensuremath{\mathtt{t}_{SB}}\xspace}
\newcommand{\tFB}{\ensuremath{\mathtt{t}_{FB}}\xspace}
\newcommand{\tSuf}{\ensuremath{\mathtt{t}_{Suf}}\xspace}
\newcommand{\tFuf}{\ensuremath{\mathtt{t}_{Fuf}}\xspace}

The last thing we consider for the theory behind the correctness of Orlin's algorithm is its worst-case running time analysis.
The following argument is an informal reconstruction from our formalisation that follows and substantially refines Korte's and Vygen's argument~\cite{KorteVygenOptimisation}.

We perform our analysis assuming elementary bounds for elementary computations performed by the algorithm.
For instance, $t_{FB}$ is an upper bound for the time that the algorithm consumes when executing the loop body in \maintainforestpar and $t_{FC}$ is the time for checking the condition
(analogously $t_{SC}$ and $t_{SB}$ for \sendflowpar, and $t_{OC}$ and $t_{OB}$ for \orlinspar).
$\text{\#}comps$ means the number of $\F$-components.
An $i$ in the subscript means the $i$th iteration of the outer loop.

First, the following should be clear about the number of loop iterations in the subprocedures during iteration $i$:
\begin{itemize}
\item $\text{\#}itsF_i = \text{\#}comps_{i-1} - \text{\#}comps_{i}$.
$\text{\#}comps_j$ is the number of components \textit{after} iteration $j$.
\item From Lemma~\ref{lem:phiB} (proof of~\ref{P52} for \sendflowpar) we know $ \text{\#}itsS_i \leq \Phi(state\refline{line_cond})_i$ where $state\refline{line_cond}$ is the state at the end of Line~\ref{line_cond},
i.e.\ before calling \sendflowpar.
For $i \geq 2$, $\Phi(state\refline{line_cond})_i = \Phi(state\refline{call_maintain_forest})_{i-1}\leq \Phi(state\refline{line27})_{i-1} + \text{\#}comps_{i-2} - \text{\#}comps_{i-1}$ due to Lemma~\ref{phiA}.
For $i=1$, let $\Phi(state\refline{call_maintain_forest})_{i-1} $ depend on the initialisation.
\end{itemize}

We obtain an upper bound for the running time of $\nu$ iterations:
{\allowdisplaybreaks
\begin{align}
&\sum\limits_{i = 1}^{\nu} \Bigl [ t_{OC} + t_{OB}  + (t_{SC} + t_{SB}) \cdot \text{\#}itsS_i + t_{SC} + (t_{FC} + t_{FB}) \cdot \text{\#}itsF_i + t_{FC} \Bigr ] \label{boundineq}\\
& \leq \;\nu \cdot (t_{OC} + t_{OB} + t_{SC} + t_{FC})  + (t_{SC} + t_{SB}) \cdot \sum\limits_{i = 1}^{\nu} \Phi(state\refline{line27})_{i-1} \;+\nonumber\\
& \;  (t_{FC} + t_{FB}) \cdot \sum\limits_{i = 1}^{\nu} (\text{\#}comps_{i-1} - \text{\#}comps_{i})\, \nonumber\\
& \; \; \; +\,(t_{SC} + t_{SB}) \cdot  \sum\limits_{i = 2}^{\nu} (\text{\#}comps_{i-2} - \text{\#}comps_{i-1}) \nonumber
\end{align}
}
The terms involving component numbers are telescope sums amounting to at most $n-1$.
The sum over the potentials remains to be bounded.
Now, we call the pairs of iteration number $i$ and a vertex that is important after Line~\ref{line27} an \textit{occurrence of an important vertex}.
For $i = 0$, we define it as the number of important vertices after the initialisation.
The sum over the potentials is the number of all occurrences, since a vertex's contribution to $\Phi(state\refline{line27})$ is either $1$ (important) or $0$ (not important) after modifying $\gamma$.
We define $comps_i$ as the set of all forest components during the $i$th iteration which have an important vertex after Line~\ref{line27} or after the initialisation, respectively.
Note that $|comps_i| \not = \#comps_i$ in general.
Since there is at most one important vertex per $\F$-component (Invariant~\ref{invar7}), $\Phi(state\refline{line27})_i \leq |comps_i|$.
We will look at the components instead of occurrences.

For a sequence $S_1,S_2,\dots S_n \subseteq \mathcal{S}$ of sets and $x \in \mathcal{S}$ we call the difference between the first and last indices $1 \leq i \leq j \leq n$ where $x\in S_i$ and $x\in S_j$ the \textit{lifetime} of $x$.
We can bound the sum of the cardinalities $|S_i|$ by $|\mathcal{S}| \cdot \delta$ where $\delta$ is an upper bound for all lifetimes.
Now we set $\mathcal{S} = \bigcup comps_i$.
According to Lemma~\ref{lemma:if_imporant_then}, the lifetime of one of those components is at most $\ell + 1$.
It therefore remains to look at the overall number of $\F$-components that can occur.
In combinatorics, a family of sets $\mathcal{X}$ is called \textit{laminar} over a non-empty universe $U$
iff $\emptyset \subset X \subseteq U$ for any $X \in \mathcal{X}$, and for any $X \in \mathcal{X}$ and $Y \in \mathcal{X}$, we have that either $X \cap Y = \emptyset$, $X \subseteq Y$, or $Y \subseteq X$ hold.
One can show that $\mathcal{S}$ is laminar over $\V$:
\isathm{Mincost_Flow_Algorithms/Orlins_Time.thy}{1554}{1554}
\begin{lem}
\label{laminarS}
$\mathcal{S} = \bigcup \limits_{i=1}^{\nu} comps_i$ is a laminar family.
\end{lem}
\begin{proof}
 Assume $X, Y \in \mathcal{S}$ with $X \neq Y$.
 Moreover, suppose $X \cap Y \neq \emptyset$.
 We know these are forest components during some iterations $i$ and $j$ respectively.
 Wlog.\ assume $i \leq j$.
 Until iteration $j$, $X$ might have been merged with other components.
 By $X'$ we denote the (possibly new) component $X$ is part of in iteration $j$.
 Of course, we have $X \subseteq X'$.
 This makes $X' \cap Y \neq \emptyset$ which means $X' = Y$ because they are from the same iteration.
 Hence, $X \subseteq Y$.
\end{proof}
The following lemma is a well-known result central to the theory of laminar families:
\isathm{Mincost_Flow_Algorithms/Laminar_Family.thy}{27}{29}
\begin{lem} \label{laminarcard}
For a family $\mathcal{X}$ laminar over $U$, we have $|\mathcal{X}| \leq 2 \cdot |U| - 1$.
\end{lem}
\begin{proof}\let\qed\relax
We conduct an induction on $|U|$.
The case $|U|=1$ is trivial as $\mathcal{X}$ only contains a singleton set.
Now, assume $|U| \geq 2$ and fix some $x \in U$.
We define $\mathcal{Y} = \lbrace X \setminus \lbrace x \rbrace. \; X \in \mathcal{X} \rbrace \setminus \lbrace \emptyset \rbrace$.
For any $(X \setminus \lbrace x\rbrace) \in \mathcal{Y}$, we call $X$ the (not necessarily unique) \textit{parent set}.
$\mathcal{Y}$ is still a laminar family:
assume neither $X \cap Y = \emptyset$, nor $X \subseteq Y$, nor $Y \subseteq X$ and the same hold for their respective parent sets which contradicts laminarity of $\mathcal{X}$.
When the induction hypothesis is applied to $\mathcal{Y}$, we obtain $|\mathcal{Y}| \leq 2 \cdot (|U| - 1) - 1$.
Now, we show $|\mathcal{X}| \leq |\mathcal{Y}| + 2$ which would already give the claim.
Let us look at how the sets in $\mathcal{X}$ and $\mathcal{Y}$ correspond to one another:
\begin{itemize}
\item $\lbrace x \rbrace$ may be in $\mathcal{X}$ but (by definition) not in $\mathcal{Y}$
\item For any set $X \in \mathcal{Y}$ there is a parent set in $\mathcal{X}$.
However, there could be two parent sets, namely, $X$ and $X \cup \lbrace x \rbrace$.
We claim that this cannot happen more than once.
Aiming for a contradiction, assume $X, \, Y \in \mathcal{Y}$ where $X \neq Y$ and $X, \,Y, \, X \uplus \lbrace x \rbrace, \, Y \uplus \lbrace x \rbrace \in \mathcal{X}$.
Then we have $(X \cup \lbrace x \rbrace) \cap (Y \cup \lbrace x \rbrace) \supseteq \lbrace x \rbrace$; thus $X \cup \lbrace x \rbrace \subseteq Y \cup \lbrace x \rbrace$ or $Y \cup \lbrace x \rbrace \subseteq X \cup \lbrace x \rbrace$ as well, due to laminarity of $\mathcal{X}$.
Since $x$ is neither in $X$ nor $Y$ by definition, we conclude $X \subseteq Y$ or $Y \subseteq X$.
Wlog. $X \subseteq Y$. By $X \neq Y$ we have $X \subset Y$.
We infer $Y \cap (X \cup \lbrace x \rbrace) = X \neq \emptyset$, and $Y \not \subseteq X \cup \lbrace x \rbrace$
(as $x \notin Y$ and $X \subset Y$), and $X \cup \lbrace x \rbrace \not \subseteq Y$ (as $x \notin Y$) which finally contradicts the laminarity of $\mathcal{X}$.
\hfill \qedsymbol
\end{itemize}
\end{proof}

By Inequality~\ref{boundineq}, Lemmas~\ref{laminarS} and~\ref{laminarcard}, and the remarks about important vertices, we can bound the overall running time by the term
\isaeq{Mincost_Flow_Algorithms/Orlins_Time.thy}{1725}{1734}
\begin{equation}
\begin{split}
&(n \cdot (\ell + k + 2) - 1) \cdot (t_{OC} + t_{OB} + t_{SC} + t_{FC})  \,+ \\
&(n-1) \cdot (t_{FC} + t_{FB} + t_{SC} + t_{SB}) \; \; + \\
&(2n -1) \cdot (\ell + 1) \cdot (t_{SC} + t_{SB}) + t_{SC} + t_{OC}
\end{split}\label{final_orlins_time}
\end{equation}
provided that Orlin's algorithm terminates in the $(\nu + 1)$th iteration.
To the last iteration we have to charge a time of $(t_{SC}+t_{SB})\cdot (\text{\#}comps_{\nu} - \text{\#}comps_{\nu - 1}) + t_{SC} + t_{OC}$.

Elementary times depend on reachability and path selection for which well-known strongly polynomial algorithms exist.
In this work, however, we only formalise the running time argument up to Term~\ref{final_orlins_time} without details on the selection procedures.

\begin{rem}
$\epsilon = \frac{1}{n}$ yields a strongly polynomial bound.
Dijkstra's algorithm results in the best running time of $\mathcal{O}(n\log n (m + n \log n))$, but it requires preprocessing to keep edge weights positive.
As long as the selection functions are strongly polynomial, the whole algorithm is strongly polynomial, as well.
\end{rem}

\paragraph*{Formalisation}
We model algorithm running times as functions returning natural numbers, using an extension of Nipkow et al.'s approach~\cite{nipkowRunningTime}.
In this approach, for every function $f:\alpha\rightarrow\beta$,
we devise a functional program $f_\text{Time}:\alpha\rightarrow\mathbb{N}$ with the same recursion and control-flow structure as the algorithm whose running time we measure.
In its most basic form, for a given input $x:\alpha$, $f_\text{Time}(x)$ returns the number of the recursive calls that $f$ would perform while processing $x$.
If $f$ involves calls to other functions, we add the running times of the called functions to the number of recursive calls of $f$.
To modularise the design of these running time models, we use locales to assume running times of the called functions.

\torlinstime takes a program state and returns a pair of a program state and a natural number, which models the running time of Orlin's algorithm.
\torlinstime uses \tmaintainforesttime and \tsendflowtime modelling the running times of the two lower-level loops.
Without explicitly specifying them, we assume times like \tSB and \tFB in the locale containing the definition of \torlinstime.
We incorporate the running times of the path selection and augmentation functions into \tSB and \tFB, respectively.
In the formal proofs, we also have \tFuf and \tSuf to model the time for unfolding the program state.
In the informal analysis, we include these into $t_{SB}$, $t_{FB}$, $t_{SC}$ and $t_{FC}$, respectively.

The proof of Lemma~\ref{laminarcard} sticks to a minimal number of concepts and  was devised for the formal statement \laminarfamilynumberofsets.
To talk about the running time consumed up to a specific number of iterations, we introduce \torlinsonesteptime and \torlinsonesteptimecheck and use the function iteration \compow, analogously to the previous subsection.
\runningtimeinitial is the final formalised statement on the running time of Orlin's algorithm, which corresponds to Term~\ref{final_orlins_time} (plus the time for the last iteration).

\section{Bellman-Ford Algorithm for Path Selection}
\label{sec:bellman_ford}
\newcommand{\bellmanfordspec}{\isainl{Set_Graphs/Graph_Algorithms/Bellman_Ford.thy}{164}{169}{bellman-ford-spec}\xspace}
\newcommand{\edgecosts}{\isainl{Set_Graphs/Graph_Algorithms/Bellman_Ford.thy}{169}{169}{edge-costs}\xspace}
\newcommand{\fold}{\texttt{fold}\xspace}
\newcommand{\connections}{\texttt{connections}\xspace}
\newcommand{\ereal}{\texttt{ereal}\xspace}
\newcommand{\trelax}{\isainl{Set_Graphs/Graph_Algorithms/Bellman_Ford.thy}{175}{181}{relax}\xspace}
\newcommand{\bellmanford}{\isainl{Set_Graphs/Graph_Algorithms/Bellman_Ford.thy}{183}{186}{bellman-ford}\xspace}
\newcommand{\searchrevpath}{\isainl{Set_Graphs/Graph_Algorithms/Bellman_Ford.thy}{197}{201}{search-rev-path}\xspace}
\newcommand{\topt}{\isainl{Set_Graphs/Graph_Algorithms/Bellman_Ford.thy}{316}{317}{OPT}\xspace}
\newcommand{\bellmanequation}{\isainl{Set_Graphs/Graph_Algorithms/Bellman_Ford.thy}{373}{375}{Bellman-Equation}\xspace}
\newcommand{\bellmanfordshortest}{\isainl{Set_Graphs/Graph_Algorithms/Bellman_Ford.thy}{560}{562}{bellman-ford-shortest}\xspace}
\newcommand{\predgraph}{\isainl{Set_Graphs/Graph_Algorithms/Bellman_Ford.thy}{1209}{1212}{pred-graph}\xspace}
\newcommand{\searchrevpathisinpredgraph}{\isainl{Set_Graphs/Graph_Algorithms/Bellman_Ford.thy}{1994}{2001}{search-rev-path-is-in-pred-graph}\xspace}
\newcommand{\searchrevpathdombellmanford}{\isainl{Set_Graphs/Graph_Algorithms/Bellman_Ford.thy}{2093}{2096}{search-rev-path-dom-bellman-ford}\xspace}
\newcommand{\computationofoptimumpath}{\isainl{Set_Graphs/Graph_Algorithms/Bellman_Ford.thy}{2261}{2265}{computation-of-optimum-path}\xspace}

To fully implement Orlin's algorithm, we have to provide suitable functions for path selection, which we currently only assume.
To obtain paths as needed for \maintainforestpar, we reuse an existing DFS formalisation~\cite{graphAlgosBook}.
For \sendflowpar, which needs minimum-cost paths, we use the Bellman-Ford algorithm~\cite{bellmanFordI,FordNetworkFLowTheory,MooreShortestPath}, as costs of residual edges can be negative.\footnote{
One could introduce an auxiliary residual graph with non-negative weights, allowing for the faster Dijkstra's algorithm instead~\cite{KorteVygenOptimisation}.
Then, Bellman-Ford is required only once in the beginning. We do not pursue this further for this formalisation.
}
There is an Isabelle/HOL formalisation of Bellman-Ford~\cite{wimmerMonadDP} by Wimmer et al., which, however, only computes the distances instead of actual paths.
Furthermore, they formalised the vertices' distances to the source as a function, instead of a data structure.
Nevertheless, it served as a useful inspiration.
For the Bellman-Ford algorithm, we assume a simple graph whose edges $\E$ are ordered pairs.
In that graph, we fix a single vertex $s$ as the \textit{source}.

\paragraph{Graph Format.}
To formalise the algorithm, we fix a list of the graph's edges and vertices in the locale \bellmanfordspec, together with a cost/weight function \edgecosts.
\edgecosts \tttu \tttv is the cost of directly going from \tttu to \tttv.
Effectively, \edgecosts takes the role of a weighted adjacency matrix.
We iterate over edges and vertices by a simple functional \fold.
We model edges $e$ not belonging to the graph by $e$ having infinite weight.
The algorithm works on a data structure \connections, specified as ADT, mapping vertices $x$ to a pair consisting of an \ereal and a vertex.
Recall that \ereal is the Isabelle/HOL type for reals extended with $\lbrace + \infty, - \infty\rbrace$.

\begin{rem}
Here we do not reuse the standard graph representation like the \multigraph locale that was introduced in Section~\ref{sec:isabelle_intro}.
A graph as a set of edges is suitable for describing theory and algorithms for flows, but it is not suitable for the Bellman-Ford algorithm whose locales are tailored to implementing this algorithm.
\end{rem}

\paragraph{Dynamic Programming.}
For every vertex $v$, the goal is to find an $s$-$v$-path whose accumulated edge weight is minimum for all $s$-$v$-paths.
Similar to the previous formalisation~\cite{wimmerMonadDP},
we formally define $OPT(l,s,v)$ as the least accumulated weight of an $s$-$v$-path of length $\leq l$ (number of vertices), formally \topt.
We also prove a \bellmanequation.
One can compute the weighted distances between $s$ and all other vertices, i.e.\ costs of the shortest paths, recursively via the Bellman Equation.
We can simulate the recursive process with an iterative computation (dynamic programming).
This is much more efficient since intermediate results can be shared among the computation paths that would appear during the recursion.

For all vertices $v$, the Bellman-Ford algorithm computes continuously improving approximations $\delta(v)$ for the weight of the shortest path from $s$ to $v$. $\delta(v)$ is the weight of some $s$-$v$-path $p$. It also maintains the predecessor $\pi(v)$ of $v$, which is the second-to-last vertex in $p$.
On well-defined inputs, i.e.\ weight-conservative graphs, the \textit{predecessor graph} (or formally \predgraph) $\Pi = \lbrace (\pi (v), v) | \, v.\, v \in \V \wedge \pi(v) \text{ is defined}\rbrace$ is always a \textit{directed acyclic graph (DAG)}.

In every iteration of the Bellman-Ford algorithm, we perform so-called \textit{relaxations} for all edges $e$~\cite{CLRS}:
if for the edge $e = (u, v)$, the distance $\delta(s,v)$ is greater than $\delta(s,u) + c(e)$ for the cost function $c$, then store $\delta(s,u) + c(u, v)$ as $\delta(s,v)$ and record $u$ as $v$'s new predecessor $\pi(v)$.
This is what \trelax does, where the pairs of the $\pi$ and $\delta$ values are stored in \connections.
By iterating over all edges while relaxing each of them, we consider paths of increasing length.

An induction over $l$ and application of the Bellman Equation leads to the statement \bellmanfordshortest,
which says that after relaxing all edges $l$ times, $\delta(v) \leq OPT(l,s,v)$ for every vertex $v$.

\paragraph{Invariants.}
We now have to prove that we will have $\delta(v) = OPT(l,s,v)$ for $l=n-1$.
While developing the formalisation, we identified the following invariants:
\begin{enumerate}[label=(IBF-\theenumi)]
\renewcommand{\theenumi}{IBF-\arabic{enumi}}
\item For every vertex $v$, $\delta (s,v) \leq OPT(l,s,v)$ where $l < n$ is the iteration number (essentially \bellmanfordshortest).\label{bfinv1}
\item For every vertex $v$, $\delta(s,v) < \infty$ iff $\pi (v)$ is defined or $v=s$.\label{bfinv2}
\item If $\delta (s,s) \neq 0$ anymore, $s$ has a defined predecessor.\label{bfinv3}
\item If a vertex $v$ has a predecessor $\pi (v)$, there is a path $p$ from $s$ to $\pi (v)$ and $weight \, (pv) = \delta(s,v)$.\label{bfinv4}
\item If a vertex has a predecessor, the predecessor is $s$ or it has a predecessor itself.\label{bfinv5}
\item If a vertex has a predecessor $\pi (v)$, then $\delta(s,v) \geq \delta(s,\pi (v)) + c(\pi(v), v)$.\label{bfinv6}
\item The predecessor graph $\Pi$ is acyclic.\label{bfinv7}
\end{enumerate}
The only invariant Wimmer et al.\ use in their Bellman-Ford formalisation~\cite{wimmerMonadDP} is Invariant~\ref{bfinv1}, however, as an equality instead.
But, since, unlike them, we need to extract paths, we have to prove invariants about the predecessors map, which are Invariants~\ref{bfinv2} - \ref{bfinv6}.
Invariants~\ref{bfinv1} - \ref{bfinv6} do not need assumptions on the structure of the graph and these follow by simple computational inductions over \bellmanford.
Invariant~\ref{bfinv7} holds, provided that the graph is weight-conservative (absence of negative-weight cycles).
Its proof is as follows:
\begin{proof}[Proof Sketch for Invariant~\ref{bfinv7}]
Assume the graph over edges $\E$ has no negative cycle.
We can lift the monotonicity property of Invariant~\ref{bfinv6} to paths in $\Pi$,
i.e.\ $\delta(s,v) \geq \delta(s,u) +  c(p)$ for vertices $u$, $v$ and an $u$-$v$-path $p$ in $\Pi$.
If the relaxation of $e = (x, y)$ adds a cycle to $\Pi$, $\delta(s,x)+ c(x,y) < \delta(s,y)$ and there is a path $p$ in $\Pi$ leading from $y$ to $x$.
Then $\delta(s,y) + c(p) + c(x,y)< \delta(s,y) $ which would make the weight of $px$ negative, resulting in a contradiction.
\end{proof}

From~\ref{bfinv1} and~\ref{bfinv4}, one could infer that the distances computed by the algorithm will eventually equal the lengths of shortest paths.

\paragraph{How to Obtain a Path.}
Backward search from a vertex $x$ with \searchrevpath along edges in $\Pi$ yields paths of exactly weight $\delta(x)$.
It follows from the invariants that $s$ is reachable by iterating the predecessor, which is equivalent to the termination result \searchrevpathdombellmanford.
\searchrevpathisinpredgraph shows that searching from a vertex $v$ yields a path in the predecessor graph leading from $s$ to $v$.
We also show that paths in the predecessor graph are better than optimum paths, making them optimum themselves.
\computationofoptimumpath combines these results into the correctness of the whole path computation.

\section{Orlin's Algorithm Executable}
\label{sec:executable}
\newcommand{\getedgeandcostforward}{\isainl{Mincost_Flow_Algorithms/Use_Orlins_Algorithm/Instantiation.thy}{286}{297}{get-edge-and-costs-forward}\xspace}
\newcommand{\tttnb}{\texttt{nb}\xspace}
\newcommand{\ingoingedges}{\texttt{ingoing-edges}\xspace}
\newcommand{\outgoingedges}{\texttt{outgoing-edges}\xspace}
\newcommand{\ttts}{\texttt{s}\xspace}
\newcommand{\tttt}{\texttt{t}\xspace}
\newcommand{\bellmanfordforward}{\isainl{Mincost_Flow_Algorithms/Use_Orlins_Algorithm/Instantiation.thy}{313}{315}{bellman-ford-forward}\xspace}
\newcommand{\gettargetforsourceauxaux}{\isainl{Mincost_Flow_Algorithms/Use_Orlins_Algorithm/Instantiation.thy}{321}{325}{get-target-for-source-aux-aux}\xspace}
\newcommand{\tttbf}{\texttt{bf}\xspace}
\newcommand{\searchrevpathexec}{\isainl{Set_Graphs/Graph_Algorithms/Bellman_Ford.thy}{188}{192}{search-rev-path-exec}\xspace}
\newcommand{\pairtorealisingredgeforward}{\isainl{Mincost_Flow_Algorithms/Use_Orlins_Algorithm/Instantiation.thy}{344}{347}{pair-to-realising-redge-forward}\xspace}
\newcommand{\getsourcetargetpathb}{\isainl{Mincost_Flow_Algorithms/Use_Orlins_Algorithm/Instantiation.thy}{366}{376}{get-source-target-path-b}\xspace}
\newcommand{\globalinterpretation}{\texttt{global-interpretation}\xspace}
\newcommand{\dfsinterpretation}{\isainl{Set_Graphs/Graph_Algorithms/DFS_Example.thy}{5}{15}{dfs}\xspace}
\newcommand{\getpathexec}{\isainl{Mincost_Flow_Algorithms/Use_Orlins_Algorithm/Instantiation.thy}{243}{243}{get-path}\xspace}
\newcommand{\maintainforestaxiomsextended}{\isainl{Mincost_Flow_Algorithms/Use_Orlins_Algorithm/Instantiation.thy}{980}{985}{maintain-forest-axioms-extended}\xspace}
\newcommand{\getsourcetargetpathaexec}{\isainl{Mincost_Flow_Algorithms/Use_Orlins_Algorithm/Instantiation.thy}{349}{359}{get-source-target-path-a}\xspace}
\newcommand{\getsourcetargetpathaax}{\isainl{Mincost_Flow_Algorithms/Use_Orlins_Algorithm/Instantiation.thy}{2620}{2624}{get-source-target-path-a-ax}\xspace}
\newcommand{\getsourcetargetpathbax}{\isainl{Mincost_Flow_Algorithms/Use_Orlins_Algorithm/Instantiation.thy}{3544}{3548}{get-source-target-path-b-ax}\xspace}

After verifying Orlin's algorithm, which involved assuming path selection functions and abstract data types for the program variables, and after implementing and verifying the Bellman-Ford algorithm to find minimum-cost paths, we are finally ready to combine these to get an executable implementation of Orlin's algorithm.
We instantiate the locale constants and prove the corresponding axioms.

\paragraph{Basic Data Structures.}
We use appropriate data structures to instantiate the variables like the flow or balances in the program state.
In the locale headers, we assume structures that map keys/indices, e.g.\ vertices or edges, to values, e.g.\ reals.
We use a red-black tree implementation of maps by Nipkow~\cite{RBTIsabelle}.
However, we add a number of special operations needed in our context, e.g.\ to update all values simultaneously or to select the \textit{argmax}.
We implement the set of active edges as a list, which would achieve the best possible asymptotic running time as we have to iterate over all edges in every iteration of \maintainforestpar.

\paragraph{Forest Paths.}
Forest paths are necessary for moving balance from one component's representative to the other when we join two components (Section~\ref{section:loopa}).
As we saw earlier, forest paths are always minimum-cost augmenting paths (Lemma~\ref{lemma:optforestpath}),
which is why we can select paths in the forest irrespectively of the costs, e.g.\ by BFS or DFS.

We use the DFS formalisation by Abdulaziz~\cite{graphAlgosBook} that computes a path in an adjacency map leading from a source to a target if there exists one.
From the auxiliary and implementation invariants we prove that the forest, represented as an adjacency map, satisfies all necessary invariants, and that the represented graph is finite.
These are the preconditions for DFS working correctly.
Recall that we encapsulated such preconditions in definitions like \maintainforestgetpathcond, which were provable inside the locales for Orlin's algorithm.
The DFS is instantiated with red-black tree operations from which we get the actual selection function \getpathexec. Its correctness lemma \maintainforestaxiomsextended coincides with the assumptions on \getpath in the locale \maintainforesttt.

\paragraph{Application of Bellman-Ford.}
The function \sendflowpar sends flow from sources to targets along minimum-cost paths that use active and forest edges only.
The residual graph in which we need paths may have multiple edges between two vertices (multigraph).
In contrast the Bellman-Ford algorithm assumes a simple graph given by a weighted adjacency matrix.
As weights for the matrix we take the costs of the cheapest edge that is available between the respective pair of vertices.

In greater detail, let us consider the case of path selection where we first select the source (Line~\ref{line10} of \sendflowpar).
\getedgeandcostforward performs the conversion to a problem on which we can use Bellman-Ford:
for a flow \tttf, 'not blocked' map \tttnb, and two vertices \tttu and \tttv, we first look up all original graph edges from \tttu to \tttv (\ingoingedges) and from \tttv to \tttu (\outgoingedges), respectively.
We then find the cheapest forward or backward residual edge pointing from \tttu to \tttv, along with the corresponding costs.
We say that this residual edge realises the simple edge $(\tttu,\tttv)$.
We only look at those edges that have positive residual capacity w.r.t.\ \tttf, and that are not blocked, i.e.\ either forest or active edges.
The central property of this function is that for every $e \in actives \cup \F$ with $\mathfrak{u}_f(e) > 0$, it finds $e' \in actives \cup \F$ with $\mathfrak{u}_f(e') > 0$ connecting the same pair of vertices where $\mathfrak{c}(e') \leq \mathfrak{c}(e)$.
We can use the second component of the pair returned by \getedgeandcostforward, namely, the residual costs, as an input to the Bellman-Ford algorithm, along with a source \ttts, as done by definition \bellmanfordforward.

For a program state, we obtain a path by the function \getsourcetargetpathaexec: \gettargetforsourceauxaux searches for a vertex \tttt with $\tttb \tlpar \tttt \trpar < - \epsilon \cdot \gamma$ such that \tttbf found a path, i.e.\ \tttt's distance to \ttts is less than $\infty$.
We then use \searchrevpathexec to go from predecessor to predecessor until we reach \ttts.
Edge by edge, we have to transform this path back into a path over residual edges, for which we use \pairtorealisingredgeforward, which gives for every simple edge a residual edge of minimum costs connecting the same vertices.

We give a proof sketch explaining why \getsourcetargetpatha satisfies the locale assumptions \getsourcetargetpathaxioms whose most important part is proven by \getsourcetargetpathaax.

\begin{proof}[Proof Sketch of Correctness of Minimum-Cost Path Selection]
Assume that the state satisfies \getsourcetargetpathacond (henceforth \textit{C}).
Let $P$ be the path returned by \getsourcetargetpatha.
Let $P_1$ be an arbitrary minimum-cost augmenting path from \ttts to \tttt.
$P_1$ might use deactivated edges, i.e.\ edges that are blocked.
Because of Invariant~\ref{invarsamecomp}, which is entailed by \textit{C}, the endpoints $x$ and $y$ of those edges $e$ belong to the same forest component.
The forest paths connecting $x$ and $y$ are optimum (Lemma~\ref{lemma:optforestpath}), hence cheaper or have the same cost as $e$.
Therefore, there is a $P_2$ that is cheaper than $P_1$ and uses active and forest edges only.
Also because of $C$, all forest edges have positive flow, which means that $P_2$ is still augmenting.
By the construction of \getedgeandcostforward, there is a corresponding path $P_3$ in the graph we use for BF and that might be even cheaper.
Since the optimality invariant is part of $C$, the graph we use for BF does not contain a negative cycle.
From properties of BF, we know that $P_3$ is distinct.
BF computes a distinct path $P_4$ from \ttts to \tttt that might even be cheaper.
\getedgeandcostforward also gives realising edges
that we use to transform $P_4$ into the path $P$ in the residual graph, which has positive capacity and whose costs are equal to the costs of $P_4$.
The costs of $P$ are cheaper than those of $P_1$, making $P$ a minimum-cost augmenting path.
Also, $P$ uses only forest and active edges.
\end{proof}

Line~\ref{line17} in \sendflowpar is the other use case for minimum-cost path selection: given a target $t$, we need to find a source $s$ reaching $t$ and a minimum-cost path from $s$ to $t$.
We reverse all directions and use $t$ as source for BF instead.
Except for this change and the different balance-related properties required for $t$, the definitions for \getsourcetargetpathb and the corresponding proofs for the assumptions of the locale \sendflowproof are analogous.

Although the proofs are conceptually simple, the formalisation is comparably verbose due to the translations between different graph formats.
We use a translation between multigraphs and simple graphs including weights, and the reduction of the single-target-shortest-path to the single-source-shortest-path problem.

\section{Finite and Infinite Capacity Networks}
\label{sec:finite_infinite}
\newcommand{\outedge}{\isainl{Flow_Theory/Hitchcock_Reduction.thy}{18}{19}{outedge}\xspace}
\newcommand{\inedge}{\isainl{Flow_Theory/Hitchcock_Reduction.thy}{18}{19}{inedge}\xspace}
\newcommand{\vtovedge}{\isainl{Flow_Theory/Hitchcock_Reduction.thy}{18}{19}{vtovedge}\xspace}
\newcommand{\dummy}{\isainl{Flow_Theory/Hitchcock_Reduction.thy}{18}{19}{dummy}\xspace}
\newcommand{\hitchcockedge}{\isainl{Flow_Theory/Hitchcock_Reduction.thy}{18}{19}{hitchcock-edge}\xspace}
\newcommand{\hitchcockwrapper}{\isainl{Flow_Theory/Hitchcock_Reduction.thy}{10}{10}{hitchcock-wrapper}\xspace}
\newcommand{\getedge}{\isainl{Flow_Theory/Hitchcock_Reduction.thy}{21}{24}{get-edge}\xspace}
\newcommand{\newfstvgen}{\isainl{Flow_Theory/Hitchcock_Reduction.thy}{31}{34}{new-fstv-gen}\xspace}
\newcommand{\newsndvgen}{\isainl{Flow_Theory/Hitchcock_Reduction.thy}{35}{38}{new-sndv-gen}\xspace}
\newcommand{\reductiongeneral}{\isainl{Flow_Theory/Hitchcock_Reduction.thy}{66}{97}{reduction-of-min-cost-flow-to-hitchcock-general}\xspace}
\newcommand{\existenceofoptimumflow}{\isainl{Mincost_Flow_Algorithms/Use_Orlins_Algorithm/Existence_Optflows.thy}{104}{105}{existence-of-optimum-flow}\xspace}
\newcommand{\G}{\ensuremath{\mathcal{G}}\xspace}

Orlin's algorithm is currently the fastest known algorithm to solve the minimum-cost flow problem, when there are
\begin{enumerate*}\item no negative cycles (since the zero flow has to be a minimum-cost flow for $b-b$ in the beginning, see Invariant~\ref{invaropt}) and
\item no edges with finite capacity, which apparently pose significant limitations to the algorithm's usability.
\end{enumerate*}
Fortunately, every finite and mixed capacity instance $I$ can be reduced to a corresponding instance $I'$ without capacities with a blow-up in graph size that is $\mathcal{O}(n + m)$~\cite{OrdenTranshipment,FordFulkersonFlowsNetworks}.
The result $I'$ is acyclic, hence without negative cycles, if there are no negative infinite-capacity cycles in $I$.
If there are such cycles in $I$, however, flows of arbitrarily low costs exist.
This reduction makes Orlin's algorithm applicable to networks with finite capacities, as well, provided that the optimisation problem is not unbounded.

In this section we describe our formalisation and verification of the functional correctness of this reduction.
This is a generalisation of a method first given by Ford and Fulkerson~\cite{FordFulkersonFlowsNetworks}, which is also presented by Korte and Vygen~\cite{KorteVygenOptimisation}.
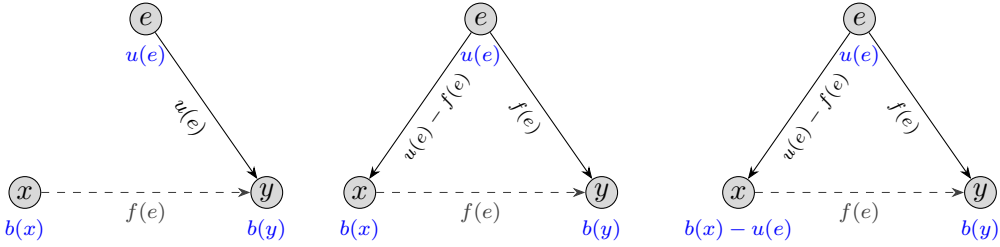
\begin{figure}
\centering
\trimbox{0mm 0mm -2mm 0mm}{
\begin{tikzpicture}[
      mycircle/.style={
         circle,
         draw=black,
         fill=gray,
         fill opacity = 0.3,
         text opacity=1,
         inner sep=0pt,
         minimum size=12pt,
         font=\normalsize},
      myarrow/.style={-Stealth},
      node distance=2cm and 1.3cm
      ]
      \node[label=below:{\scriptsize \color{blue} $b(x)$}, mycircle] (x) {$x$} ;
      \node[label=below:{\scriptsize \color{blue} $u(e)$\color{black}}, mycircle,above right=of x] (e) {$e$};
      \node[label=below:{\scriptsize \color{blue} $b(y)$\color{black}}, mycircle,below right=of e] (y) {$y$};

\draw[draw=darkgray, dashed, style={-Stealth}] (x) -- node[sloped, color=darkgray,font = \scriptsize, below] {$f(e)$} (y);
\draw[myarrow] (e) -- node[sloped, font = \scriptsize, below] {$u(e)$} (y);
    \end{tikzpicture}
    }\trimbox{0mm 0mm -2mm 0mm}{
\begin{tikzpicture}[
      mycircle/.style={
         circle,
         draw=black,
         fill=gray,
         fill opacity = 0.3,
         text opacity=1,
         inner sep=0pt,
         minimum size=12pt,
         font=\normalsize},
      myarrow/.style={-Stealth},
      node distance=2cm and 1.3cm
      ]
      \node[label=below:{\scriptsize \color{blue} $b(x)$}, mycircle] (x) {$x$} ;
      \node[label=below:{\scriptsize \color{blue} $u(e)$\color{black}}, mycircle,above right=of x] (e) {$e$};
      \node[label=below:{\scriptsize \color{blue} $b(y)$\color{black}}, mycircle,below right=of e] (y) {$y$};

\draw[draw=darkgray, dashed, style={-Stealth}] (x) -- node[sloped, color=darkgray,font = \scriptsize, below] {$f(e)$} (y);
\draw[myarrow] (e) -- node[sloped, font = \tiny, below] {$f(e)$} (y);
\draw[myarrow] (e) -- node[sloped, font = \tiny, below] {$u(e) - f(e)$} (x);
    \end{tikzpicture}
    }
    \trimbox{0mm 0mm -2mm 0mm}{
\begin{tikzpicture}[
      mycircle/.style={
         circle,
         draw=black,
         fill=gray,
         fill opacity = 0.3,
         text opacity=1,
         inner sep=0pt,
         minimum size=12pt,
         font=\normalsize},
      myarrow/.style={-Stealth},
      node distance=2cm and 1.3cm
      ]
      \node[label=below:{\scriptsize \color{blue} $b(x) - u(e)$}, mycircle] (x) {$x$} ;
      \node[label=below:{\scriptsize \color{blue} $u(e)$\color{black}}, mycircle,above right=of x] (e) {$e$};
      \node[label=below:{\scriptsize \color{blue} $b(y)$\color{black}}, mycircle,below right=of e] (y) {$y$};

\draw[draw=darkgray, dashed, style={-Stealth}] (x) -- node[sloped, color=darkgray,font = \scriptsize, below] {$f(e)$} (y);
\draw[myarrow] (e) -- node[sloped, font = \tiny, below] {$f(e)$} (y);
\draw[myarrow] (e) -- node[sloped, font = \tiny, below] {$u(e) - f(e)$} (x);
    \end{tikzpicture}
    }
\caption{We replace an edge $e$ from $x$ to $y$ with finite capacity $u(e) < \infty$ with a new vertex $e$ with $b(e) = u(e)$ and two edges $e_1=(e,x)$ and $e_2=(e,y)$, each with infinite capacity.
After first adding $e_2$, the flow along this new edge is exactly $u(e)$ although it should be equal to $f(e) \leq u(e)$.
Adding the other edge $e_1$ makes this possible because we can then send $f(e)$ units of flow from $e$ to $y$ and $u(e) - f(e)$ units to $x$.
The vertex $x$ has now an abundant flow of $u(e)$ ($f(e)$ units do not leave $x$ anymore and there is an additional ingoing flow of $u(e) - f(e)$).
We need to adjust the balance $b(x)$ to again obtain a balance-complying flow.}
\label{fig:replace_finite_edges}
\end{figure}

\paragraph{The Idea.}
For an original network with capacities, we construct a new network without capacities and modified balances such that a valid flow in one network allows us to easily find a valid flow in the other one.
The idea is to encode edge capacities into vertex balances, which means we introduce vertices that replace edges in the network, as shown in Figure~\ref{fig:replace_finite_edges}:
for an edge $e$ with endpoints $x$ and $y$, and capacity $u(e) < \infty$, we need a corresponding edge $e_2$ that carries a flow $f'(e_2) = f(e) \leq u(e)$ despite $e_2$'s infinite capacity. We can achieve this by deleting $e$ from the network and replacing it with a vertex $e$ with $b(e) = u (e)$ and an edge $e_2 = (e, y)$.
Since this would enforce $f'(e_2) = u(e)$, the flow needs another way to leave the vertex $e$.
We introduce $e_1 = (e, x)$ to remove the remaining flow from $e$ (Figure~\ref{fig:replace_finite_edges}).
This results in $f'(e_2) = f(e) \leq u(e)$, $f'(e_1) = u(e) - f(e)$ and $f'(d) = f(d)$ for all other edges $d$ as a possible flow in the modified network.

Flow validity is not only about capacities but about balances as well.
Balance constraints that $f$ satisfied before might not be satisfied by $f'$.
While the excess of $y$ is still unaltered, the modified flow violates the balance constraints at $x$:
An amount of $f(e)$ is not leaving the vertex $x$ anymore and there is an additional ingoing flow of $u(e) - f(e)$,
which means that we must decrease $b(x)$ by $u(e)$ to maintain feasibility.

\paragraph{The Reduction.}
For the network $(\V, \E)$ with costs $c:\E\rightarrow\mathbb{R}$, balances $b:\V \rightarrow \mathbb{R}$ and capacities $u:\E \rightarrow \mathbb{R}^+_0 \cup \lbrace +\infty \rbrace$, we define $\V'$, $\E'$, $c'$, $b'$:
let $\delta^-_{\infty}(x)$ and $\delta^+_{\infty}(x)$ be the ingoing and outgoing edges of $x$ with infinite capacity, respectively.
Let $\E_{\infty}$ be the set of edges in $\E$ with infinite capacity.
We set $\mathcal{V'} = \V \cup (\E - \E_{\infty})$.
Let $\E_1 = \lbrace (e, x). \; e \in \delta^+(x) - \delta^+_{\infty}(x) \wedge x \in \V\rbrace$ and $\E_2 = \lbrace \langle e, x \rangle. \; e \in \delta^-(x) - \delta^-_{\infty}(x) \wedge x \in \V\rbrace$.
By definition, let $(a, b)$ and $\langle a, b\rangle$ denote two distinct pairs with components $a$ and $b$.
Now set $\E_3 = \E_{\infty}$ and $\E' = \E_1 \cup \E_2 \cup \E_3$.
Let $c'(e, x) = 0$ for $(e,x) \in \E_1$, $c'\langle e,x \rangle = c(e)$ for $\langle e,x \rangle \in \E_2$ and $c'(e) = c(e)$ for $e \in \E_3$.
For $e \in \E- \E_{\infty}$ we set $b'(e) = u (e)$.
For vertices $x \in \V$ we define
\[
b'(x) = b(x) - \sum \limits _{e \in \delta^+(x) - \delta^+_{\infty}(x)} u(e).
\]

\isathm{Flow_Theory/Hitchcock_Reduction.thy}{66}{97}
\begin{thm}[Reduction of mixed-capacity to infinite-capacity Problems\label{genred}]
We assume a flow network $\G$ based on a directed multigraph.
We construct $\G'$ with infinite capacities as sketched above.
The following hold:
\begin{enumerate}
\item For a flow $f$ feasible in $\G$, $f'$ is feasible in $\G'$ where $f'(e, x) =  u(e) -  f(e)$ for $(e,x) \in \E_1$, $f'\langle e, x \rangle =  f(e)$ for $\langle e, x \rangle \in \E_2$ and $f'(e) = f(e)$ for $e \in \E_3$.
The costs correspond as well, i.e.\ $c(f) = c'(f')$.
\item For a flow $f'$ feasible in $\G'$, $f$ is feasible in $\G$ where $f(e) = f'\langle e, x \rangle$ for $x$ being the second vertex of $e \in \E \setminus \E_{\infty}$.
$f(e) = f'(e)$ for $e \in \E_{\infty}$.
Again $c'(f') = c(f)$.
\item An optimum flow in one network defines an optimum flow in the other network.
\item $\G'$ contains a negative cycle iff $\G$ contains a negative cycle with infinite capacity.
\end{enumerate}
\end{thm}

\paragraph{Formalisation.}
Endpoints of edges in the new graph are uniquely determined by the original edge the new edge $e$ refers to and whether $e$ is in $\E_1$, $\E_2$ or $\E_3$.
Storing the vertices in the new graph's edges would be redundant.
We encode whether an edge is in $\E_1$, $\E_2$ or $\E_3$ into a data type \hitchcockedge with constructors \inedge, \outedge and \vtovedge, which are meant to realise edges in $\E_1$, $\E_2$ and $\E_3$, respectively.
The new graph has edges of type \hitchcockedge.
\newfstvgen and \newsndvgen (endpoints of edges in the new graph) take a \hitchcockedge and return an edge or vertex of the old graph.
In the formalisation, the latter two are of different types which is why we wrap them in a \hitchcockwrapper to obtain vertices of uniform type (i.e. type not depending on whether they wrap an edge or a vertex) as required by the \multigraph locale.
The \dummy is necessary to ensure the possibility of an edge between every vertex, as required by the axioms of the \multigraph locale.
\getedge returns the edge in the original graph a \hitchcockedge refers to if there is any.
Functions like \getedge are necessary to express the translation/correspondence between the two graphs.
The formal proof of Theorem~\ref{genred} involves plenty of sum manipulation which is hinted by the informal definitions.

\paragraph{Solving several Types of Flows.}
We made the above reduction executable, as well, allowing us to solve different types of flow problems in strongly polynomial time:
Orlin's algorithm itself finds minimum-cost flows in networks with infinite capacity.
Our formalisation contains an instantiation for both simple graphs and multigraphs.
We can compute mixed-capacity and finite-capacity minimum-cost flows by combining Orlin's algorithm with the reduction from Theorem~\ref{genred}.
For all of those problems, there is now verified code.

\paragraph{Existence of Optimum Solutions.}
We can characterise the existence of an optimum solution for feasible problems by the absence of cycles with infinite capacities and negative total costs:
after applying the reduction to such an instance, Orlin's algorithm will terminate and return an optimum flow, resulting in a constructive proof for one implication.
Conversely, a negative cycle with infinite capacities allows for arbitrarily small costs by arbitrarily high/many augmentations.
It is necessary to use Orlin's algorithm here as the simpler methods \ssppar or \ascalingpar might not terminate.
\existenceofoptimumflow, which is a characterisation of the existence of optimum solutions, is the final theorem on minimum-cost flows in our formalisation.

\begin{rem}
Combining this reduction with Orlin's algorithm is still the most efficient method (in terms of worst-case running time) for general minimum-cost flows, i.e.\ those with arbitrary finite or infinite capacities, and arbitrary balances and costs~\cite{KorteVygenOptimisation,schrijverBook}.
Another method, that directly works on the capacitated network, achieves the same running time~\cite{VygenDualMinCost}.
For restricted problems, weakly polynomial methods might be faster, e.g.\ almost linear if the balances, capacities and costs are integral and polynomially bounded~\cite{almostLinearTimeFlows,almostlinearmincostDeterministic}.
\end{rem}

\section{Discussion}
\newcommand{\maxflowmincut}{\isainl{Flow_Theory/STFlow.thy}{862}{865}{max-flow-min-cut}\xspace}
\newcommand{\stflowdecomposition}{\isainl{Flow_Theory/STFlow.thy}{195}{205}{s-t-flow-decomposition}\xspace}
\paragraph{Related Work}
The first formalisation of maximum flows, namely, the max-flow-min-cut theorem and the Ford-Fulkerson algorithm, was in 2005 in Mizar by Lee~\cite{FordFulkersonMizar}.
Lean's \textit{mathlib}~\cite{mathlib} does not contain network flows; the only substantial Lean development we are aware of is an unpublished AI-assisted formalisation of the max-flow-min-cut theorem~\cite{leanmaxflowmincut}.
We are also not aware of any developments in Rocq.

The closest work to ours is Lammich and Sefidgar's work on maximum flows~\cite{LammichFlows} in Isabelle.
The algorithms that we considered here share a number of features with other maximum flow algorithms that were formally analysed before,
most notably the fact that they iteratively compute augmenting paths to incrementally improve a solution until an optimal solution is reached.
Those algorithms also use residual graphs, which are intuitively graphs containing the remaining capacity w.r.t.\ the current flow maintained by the algorithm, and which Lammich and Sefidgar have already formalised~\cite{LammichFlows}.
Our development on min-cost flows is powerful enough to rederive most of the results about max-flows they have already formalised.
A notable result contained in both formalisations is the Max-Flow-Min-Cut Theorem~\cite{FordFulkerson,DantzigFulkersonMaxflowMincut} (formally \maxflowmincut) and the max-flow analogue to Theorem~\ref{optcrit}.
A result that we formalised for the first time is the \textit{decomposition} of so-called \textit{$s$-$t$-flows}~\cite{GallaiDecomposition,FordFulkerson} (formally \stflowdecomposition).

Results on augmenting paths (and their characterisation of optimality, like Theorem~\ref{optcrit}) have been formalised for the theory of optimisation of matching (Berge's Lemma)~\cite{BlossomIsabelleJAR}, and for matroid intersection~\cite{MatroidIsabelle}.

We are, however, the first to formalise algorithms that involve scaling, which is used to design fast optimisation algorithms, including algorithms with the fastest worst-case running times for different variants of matching and shortest path problems~\cite{DuanScalingMatching,GabowMatchingScaling,GabowTarjanMatchingScaling,GabowTarjanScalingMatching,OrlinScalingAssignment},
in addition to minimum-cost flows, which we consider here.
Our work here provides a blueprint to formalise the correctness of those other scaling algorithms.

The work by Lammich and Sefidgar~\cite{LammichFlows} also contains formalisations of the running time analyses.
They too count loop iterations as the main indicator of running time.
For Orlin's algorithm, however, the proofs involve properties of laminar families, which were never formalised before.
These properties are necessary (and reusable) for arguing about running times of many other state-of-the-art combinatorial optimisation algorithms, including Edmonds' weighted matching algorithm~\cite{EdmondsWeightedMatching}, the Micali-Vazirani algorithm~\cite{MVproofthree}, and Gabow's cost scaling weighted matching algorithm~\cite{GabowMatchingScaling}.

With the locales \multigraphspec and \multigraph, we model multigraphs as sets of edges together with functions giving the two endpoints.
Noschinski's Isabelle/HOL representation~\cite{noschinskiGraphLib} is very close to this, except that it treats the vertices as a distinct parameter.
To make the reasoning smoother, our representation is purely edge-based, and thus, avoids invariants about relationships between edges and vertices.
In Isabelle/HOL, Edmonds formalised undirected hypergraphs~\cite{edmondsHypergraphs} as a set of sets (the edges) and a set of allowed elements (possible vertices).
In Lean's \textit{mathlib}, digraphs are an adjacency predicate over vertices, and (with dependent types) multigraphs are modelled by Quivers.

\paragraph{Mathematical Insights} A main outcome of our work, from a mathematical perspective, is our proof of Theorem~\ref{thm911},
which is its first complete combinatorial proof.
We have devised a first attempt at a combinatorial proof for the gap (Lemma~\ref{FBPelim}), described in the conference paper~\cite{ScalingIsabelle}.
That proof, which involved a complex graphical construction, was later simplified by Vygen~\cite{vygenemailproof}, who proposed a proof that depends on Eulerian circuits, which we then further simplified to one main insight: by adding a fresh edge, we can make the structure regular, i.e.\ by making the set of edges form a set of disjoint cycles instead of cycles and one path.
We use a similar approach of regularisation by adding an edge to prove the decomposition theorem for max-flows.
Another less important and rather technical insight that came out of our formalisation is that we showed that~\ref{P1} holds, even for problems with finite capacities.
This is unlike Korte and Vygen's as well as Plotkin and Tardos's proof of this claim~\cite{PlotkinTardosDualNetworkSimplex}, which are the only two detailed proofs of this claim.
This highlights the potential role of formalising deep proofs in filling gaps as well as in the generation of new proofs and/or insights.
There were also other non-trivial claims in the main textbook we used as reference~\cite{KorteVygenOptimisation}, for which no proof is given.
E.g.\ it is claimed that the accumulated augmentations through a forest edge are below $2 (n-1) \gamma$.
We could not formalise the textbook's proof as it also had gaps, and we devised Invariants~\ref{invara1} and~\ref{invara2} to be able to prove it.
Other gaps like the proof of Lemma~\ref{phiA}, which is merely asserted in the literature in a non-formalisable way, had to be closed for the Properties~\ref{P51} and \ref{P52} which both involve the change in $\Phi$.
We also detailed a few technicalities for Theorem~\ref{SSPtotalcorrect}, and a few parts of the termination proof.%

\paragraph{Methodology} The formalisation of the algorithms presented here is around 36k lines of proof scripts (flow theory: 7k, min-cost flow algorithms and running time: 20k, instantiation of Orlin's algorithm and Bellman-Ford: 9k).
Our methodology is based on using Isabelle/HOL's \emph{locale}s to implement Wirth's notion of stepwise refinement~\cite{refinementWirth}, thereby compartmentalising different types of reasoning.
Many authors, e.g.\ Nipkow~\cite{AmortizedComplexity}, Maric~\cite{maricIsomorphism}, and Abdulaziz~\cite{graphAlgosBook},
used this locale-based implementation of refinement earlier.
In this approach, non-deterministic computation is handled by assuming the existence and properties of functions that compute non-deterministically, without assuming anything about the functions' implementation.
Our formalisation is one further example showing that this approach scales to proving the correctness of some of the most sophisticated algorithms.

A limitation is that the locales become verbose quite soon, especially when employing several instances of the same ADT locale.
This is a common issue when formalising an algorithm with many variables like Orlin's algorithm since each program variable comes with its own ADT.
The verbosity of the approach could be improved using records, as shown by Lammich and Lochbihler~\cite{collectionsFramework}.

We note that our focus here is more on formalising the mathematical argument behind the algorithm and less on obtaining the most efficient executable program.
A notable alternative implementation of refinement is Peter Lammich's~\cite{LammichRefinement} framework.
It has the advantage of generating fast imperative implementations, but comes at the cost of a steep learning curve and somewhat fragile automation.

We would, however, always start with the least understood (sub-)algorithms and then go on to the better understood ones.
This led us to using the locale-based approach in two ways: top-down and bottom-up.
For instance, our approach to specifying \tsendflow was top-down, where we assumed the existence of a procedure for finding shortest paths between sources and targets.
On the other hand, for specifying \torlinsfun we went bottom-up, where we first defined and proved the correctness of \tsendflow and \tmaintainforest,
and then started defining and proving the correctness of \torlinsfun, which calls both \tsendflow and \tmaintainforest.

We define procedures as recursive functions using Isabelle/HOL's function package \cite{functionPackageIsabelle}, program states as records, and invariants as predicates on program states.
We derive theorems characterising different properties of recursive functions, and supply them to Isabelle/HOL's classical reasoning and simplification.
In this approach, the automation handles proofs of invariants and monotone properties, e.g.\ the growth of the forest or the changes in $\Phi$.
Again, one can use other approaches to model algorithms and automate reasoning about them, like using monads~\cite{lammichMonads} or while combinators~\cite{nipkowWhileCombinator}.
Both of those approaches allow for greater automation, which is particularly useful for reasoning about low-level implementations.
However, those two approaches are not as useful where the majority of the correctness argument lies in background mathematical proofs,
as in the case of Orlin's algorithm, because they add at least one layer of definitions between the theorem prover's basic logic, i.e.\ a monad, and the algorithm's model.

A last noteworthy methodological point is the avoidance of code replication.
This happens with most approaches to stepwise refinement, e.g.\ Lammich's framework~\cite{LammichRefinement} or Nipkow's While-combinator~\cite{nipkowWhileCombinator}.
In those approaches, one defines an algorithm a number of times, where each time a higher level of detail is added, e.g.\ at one level one might use mathematical sets, at the next a list-based implementation, and at the third an optimised array-based implementation.
A somewhat extreme example is Nipkow's verification of a functional implementation of the Gale-Shapley algorithm, where the algorithm was defined 8 different times~\cite{galeshapleyIsabelle}, at an increasing level of detail.
Here, we specify the algorithm only once using ADTs~\cite{hoareRefinement} and then obtain executable code by instantiating those ADTs with purely functional implementations.
We adopted this approach, that was also used for matroid and greedoid algorithms~\cite{MatroidIsabelle}, from earlier work on verifying graph traversal algorithms~\cite{graphAlgosBook}.
The main advantage of this approach is that it avoids replication, but a risk is that starting directly with the ADTs could lead to too many low-level implementation-related details cluttering the verification process.
In our experience here, however, Isabelle's automation, with minimal setup, can discharge most implementation related proof goals, allowing the user to focus on more interesting goals at the algorithmic level.

\bibliographystyle{alphaurl}
\bibliography{long_paper,third_bib}

\newcommand{\etalchar}[1]{$^{#1}$}
\begin{thebibliography}{AAMR25}

\bibitem[AA24]{ScalingIsabelle}
Mohammad Abdulaziz and Thomas Ammer.
\newblock A {{Formal Analysis}} of {{Capacity Scaling Algorithms}} for
  {{Minimum Cost Flows}}.
\newblock In {\em The 15th {{International Conference}} on {{Interactive
  Theorem Proving}} ({{ITP}})}, 2024.
\newblock \href {https://doi.org/10.4230/LIPICS.ITP.2024.3}
  {\path{doi:10.4230/LIPICS.ITP.2024.3}}.

\bibitem[AA26]{abdulaziz2026formalanalysiscapacityscaling}
Mohammad Abdulaziz and Thomas Ammer.
\newblock {A Formal Analysis of Capacity Scaling Algorithms for Minimum-Cost
  Flows}, 2026.
\newblock \href {https://arxiv.org/abs/2602.03701v1}
  {\path{arXiv:2602.03701v1}}.

\bibitem[AAMR25]{MatroidIsabelle}
Mohammad Abdulaziz, Thomas Ammer, Shriya Meenakshisundaram, and Adem Rimpapa.
\newblock A {{Formal Analysis}} of {{Algorithms}} for {{Matroids}} and
  {{Greedoids}}.
\newblock In {\em The 16th {{International Conference}} on {{Interactive
  Theorem Proving}} ({{ITP}})}, 2025.
\newblock \href {https://doi.org/10.4230/LIPIcs.ITP.2025.2}
  {\path{doi:10.4230/LIPIcs.ITP.2025.2}}.

\bibitem[Abd25]{graphAlgosBook}
Mohammad Abdulaziz.
\newblock Graph {{Algorithms}}.
\newblock In {\em Functional {{Data Structures}} and {{Algorithms}}: {{A Proof
  Assistant Approach}}}, volume~67, pages 273--298. Association for Computing
  Machinery, New York, NY, USA, 1 edition, September 2025.
\newblock \href {https://doi.org/10.1145/3731369.3731394}
  {\path{doi:10.1145/3731369.3731394}}.

\bibitem[AM23]{RankingIsabelle}
Mohammad Abdulaziz and Christoph Madlener.
\newblock A {{Formal Analysis}} of {{RANKING}}.
\newblock In {\em The 14th {{Conference}} on {{Interactive Theorem Proving}}
  ({{ITP}})}, 2023.
\newblock \href {https://doi.org/10.4230/LIPIcs.ITP.2023.3}
  {\path{doi:10.4230/LIPIcs.ITP.2023.3}}.

\bibitem[AM26]{BlossomIsabelleJAR}
Mohammad Abdulaziz and Kurt Mehlhorn.
\newblock A {{Formal Correctness Proof}} of {{Edmonds}}' {{Blossom Shrinking
  Algorithm}}.
\newblock {\em The Journal of Automated Reasoning}, 70(1), 2026.
\newblock \href {https://doi.org/10.1007/s10817-026-09747-y}
  {\path{doi:10.1007/s10817-026-09747-y}}.

\bibitem[AMO93]{OrlinFlowBook}
Ravindra~K. Ahuja, Thomas~L. Magnanti, and James~B. Orlin.
\newblock {Network Flows: Theory, Algorithms, and Applications}.
\newblock 1993.
\newblock URL: \url{https://dl.acm.org/doi/10.5555/137406}.

\bibitem[AP19]{isabelleGreensThm}
Mohammad Abdulaziz and Lawrence~C. Paulson.
\newblock An {{Isabelle}}/{{HOL Formalisation}} of {{Green}}'s {{Theorem}}.
\newblock {\em Journal of Automated Reasoning}, 63(3):763--786, 2019.
\newblock \href {https://doi.org/10.1007/s10817-018-9495-z}
  {\path{doi:10.1007/s10817-018-9495-z}}.

\bibitem[Bal14]{LocalesBallarin}
Clemens Ballarin.
\newblock Locales: {{A Module System}} for {{Mathematical Theories}}.
\newblock {\em J. Autom. Reason.}, 52(2):123--153, 2014.
\newblock \href {https://doi.org/10.1007/s10817-013-9284-7}
  {\path{doi:10.1007/s10817-013-9284-7}}.

\bibitem[BCK{\etalchar{+}}23]{almostlinearmincostDeterministic}
Jan Van~Den Brand, Li~Chen, Rasmus Kyng, Yang~P. Liu, Richard Peng,
  Maximilian~Probst Gutenberg, Sushant Sachdeva, and Aaron Sidford.
\newblock A deterministic almost-linear time algorithm for minimum-cost flow.
\newblock In {\em 64th {{IEEE Annual Symposium}} on the {{Foundations}} of
  {{Computer Science}} ({{FOCS}})}, 2023.
\newblock \href {https://doi.org/10.1109/FOCS57990.2023.00037}
  {\path{doi:10.1109/FOCS57990.2023.00037}}.

\bibitem[Bel58]{bellmanFordI}
Richard Bellman.
\newblock On a routing problem.
\newblock {\em Quarterly of Applied Mathematics}, pages 87--90, 1958.
\newblock \href {https://doi.org/10.1090/qam/102435}
  {\path{doi:10.1090/qam/102435}}.

\bibitem[BN02]{nipkowWhileCombinator}
Stefan Berghofer and Tobias Nipkow.
\newblock Executing {{Higher Order Logic}}.
\newblock In {\em Types for {{Proofs}} and {{Programs}}}, Berlin, Heidelberg,
  2002.
\newblock \href {https://doi.org/10.1007/3-540-45842-5_2}
  {\path{doi:10.1007/3-540-45842-5_2}}.

\bibitem[CKL{\etalchar{+}}22]{almostLinearTimeFlows}
Li~Chen, Rasmus Kyng, Yang~P. Liu, Richard Peng, Maximilian~Probst Gutenberg,
  and Sushant Sachdeva.
\newblock Maximum {{Flow}} and {{Minimum-Cost Flow}} in {{Almost-Linear Time}}.
\newblock In {\em 63rd {{IEEE Annual Symposium}} on {{Foundations}} of
  {{Computer Science}} ({{FOCS}})}, 2022.
\newblock \href {https://doi.org/10.1109/FOCS54457.2022.00064}
  {\path{doi:10.1109/FOCS54457.2022.00064}}.

\bibitem[CLRS09]{CLRS}
Thomas~H. Cormen, Charles~E. Leiserson, Ronald~L. Rivest, and Clifford Stein.
\newblock {\em Introduction to {{Algorithms}}}.
\newblock MIT Press, Cambridge, MA, USA, 3rd edition, July 2009.

\bibitem[DF57]{DantzigFulkersonMaxflowMincut}
G.~B. Dantzig and D.~R. Fulkerson.
\newblock {\em {On the Max-Flow Min-Cut Theorem of Networks}}, pages 215--222.
\newblock Princeton University Press, Princeton, 1957.
\newblock \href {https://doi.org/10.1515/9781400881987-013}
  {\path{doi:10.1515/9781400881987-013}}.

\bibitem[DPS17]{DuanScalingMatching}
Ran Duan, Seth Pettie, and Hsin-Hao Su.
\newblock Scaling {{Algorithms}} for {{Weighted Matching}} in {{General
  Graphs}}.
\newblock In {\em Proceedings of the {{Twenty-Eighth Annual ACM-SIAM
  Symposium}} on {{Discrete Algorithms}}, {{SODA}}}, 2017.
\newblock \href {https://doi.org/10.1137/1.9781611974782.50}
  {\path{doi:10.1137/1.9781611974782.50}}.

\bibitem[Edm65]{EdmondsWeightedMatching}
Jack Edmonds.
\newblock Maximum matching and a polyhedron with 0,1-vertices.
\newblock {\em Journal of Research of the National Bureau of Standards Section
  B Mathematics and Mathematical Physics}, 69(1 and 2):55--56, January 1965.
\newblock \href {https://doi.org/10.6028/jres.069B.013}
  {\path{doi:10.6028/jres.069B.013}}.

\bibitem[EK72]{EdmondsKarpScalingFlows}
Jack Edmonds and Richard~M. Karp.
\newblock Theoretical {{Improvements}} in {{Algorithmic Efficiency}} for
  {{Network Flow Problems}}.
\newblock {\em J. ACM}, 19(2):248--264, 1972.
\newblock \href {https://doi.org/10.1145/321694.321699}
  {\path{doi:10.1145/321694.321699}}.

\bibitem[EP24]{edmondsHypergraphs}
Chelsea Edmonds and Lawrence~C. Paulson.
\newblock Formal {{Probabilistic Methods}} for {{Combinatorial Structures}}
  using the {{Lov\'asz Local Lemma}}.
\newblock In {\em The 13th {{ACM SIGPLAN International Conference}} on
  {{Certified Programs}} and {{Proofs}} ({{CPP}})}, London UK, January 2024.
\newblock \href {https://doi.org/10.1145/3636501.3636946}
  {\path{doi:10.1145/3636501.3636946}}.

\bibitem[FF56]{FordFulkerson}
Lester~Randolph Ford and Delbert~R. Fulkerson.
\newblock Maximal flow through a network.
\newblock {\em Canadian Journal of Mathematics}, 8:399--404, 1956.
\newblock \href {https://doi.org/10.4153/CJM-1956-045-5}
  {\path{doi:10.4153/CJM-1956-045-5}}.

\bibitem[FF62]{FordFulkersonFlowsNetworks}
L.~R. Ford and D.~R. Fulkerson.
\newblock {\em {Flows in Networks}}.
\newblock Princeton University Press, 1962.
\newblock \href {https://doi.org/10.2307/jj.19783677}
  {\path{doi:10.2307/jj.19783677}}.

\bibitem[For56]{FordNetworkFLowTheory}
L.~R. Ford.
\newblock {\em {Network Flow Theory}}.
\newblock RAND Corporation, Santa Monica, CA, 1956.
\newblock URL: \url{https://www.rand.org/pubs/papers/P923.html}.

\bibitem[Gab17]{GabowMatchingScaling}
Harold~N. Gabow.
\newblock The {{Weighted Matching Approach}} to {{Maximum Cardinality
  Matching}}.
\newblock {\em Fundam. Informaticae}, 154(1-4):109--130, 2017.
\newblock \href {https://doi.org/10.3233/FI-2017-1555}
  {\path{doi:10.3233/FI-2017-1555}}.

\bibitem[Gal58]{GallaiDecomposition}
T.~Gallai.
\newblock Maximum-{{Minimum S\"atze}} \"uber {{Graphen}}.
\newblock {\em Acta Mathematica Academiae Scientiarum Hungarica},
  9(3):395--434, September 1958.
\newblock \href {https://doi.org/10.1007/BF02020271}
  {\path{doi:10.1007/BF02020271}}.

\bibitem[GT89]{GabowTarjanMatchingScaling}
Harold~N. Gabow and Robert~Endre Tarjan.
\newblock Faster {{Scaling Algorithms}} for {{Network Problems}}.
\newblock {\em SIAM J. Comput.}, 18(5):1013--1036, 1989.
\newblock \href {https://doi.org/10.1137/0218069} {\path{doi:10.1137/0218069}}.

\bibitem[GT91]{GabowTarjanScalingMatching}
Harold~N. Gabow and Robert~Endre Tarjan.
\newblock Faster {{Scaling Algorithms}} for {{General Graph-Matching
  Problems}}.
\newblock {\em J. ACM}, 38(4):815--853, 1991.
\newblock \href {https://doi.org/10.1145/115234.115366}
  {\path{doi:10.1145/115234.115366}}.

\bibitem[Hoa72]{hoareRefinement}
C.~A.~R. Hoare.
\newblock Proof of correctness of data representations.
\newblock {\em Acta Informatica}, 1(4):271--281, December 1972.
\newblock \href {https://doi.org/10.1007/BF00289507}
  {\path{doi:10.1007/BF00289507}}.

\bibitem[IT20]{poincarebendixson}
Fabian Immler and Yong~Kiam Tan.
\newblock The {{Poincar\'e-Bendixson}} theorem in {{Isabelle}}/{{HOL}}.
\newblock In {\em The 9th {{ACM SIGPLAN International Conference}} on
  {{Certified Programs}} and {{Proofs}} ({{CPP}})}, 2020.
\newblock \href {https://doi.org/10.1145/3372885.3373833}
  {\path{doi:10.1145/3372885.3373833}}.

\bibitem[Jew58]{jewellMinCostSSP}
William~S Jewell.
\newblock Optimal flow through networks.
\newblock {\em Operations Research}, 6(4):633--644, 1958.

\bibitem[Jo26]{leanmaxflowmincut}
Sunghyeon Jo.
\newblock Max-flow min-cut: a comprehensive proof process, 04 2026.
\newblock URL:
  \url{https://github.com/ainta/lean-algorithms/tree/main/combinatoric_optimization/max_flow_min_cut/MaxFlowMinCut}.

\bibitem[KHSA25]{cookLevinFramework}
Kevin Kappelmann, Fabian Huch, Lukas Stevens, and Mohammad Abdulaziz.
\newblock Proof-{{Producing Translation}} of {{Functional Programs}} into a
  {{Time}} \& {{Space Reasonable Model}}, March 2025.
\newblock \href {https://arxiv.org/abs/2503.02975} {\path{arXiv:2503.02975}},
  \href {https://doi.org/10.48550/arXiv.2503.02975}
  {\path{doi:10.48550/arXiv.2503.02975}}.

\bibitem[Kle67]{optCritMinCostFlows}
Morton Klein.
\newblock A {{Primal Method}} for {{Minimal Cost Flows}} with {{Applications}}
  to the {{Assignment}} and {{Transportation Problems}}.
\newblock {\em Management Science}, November 1967.
\newblock \href {https://doi.org/10.1287/mnsc.14.3.205}
  {\path{doi:10.1287/mnsc.14.3.205}}.

\bibitem[Kra09]{functionPackageIsabelle}
Alexander Krauss.
\newblock {\em Automating Recursive Definitions and Termination Proofs in
  Higher-Order Logic}.
\newblock PhD thesis, Technical University Munich, 2009.
\newblock URL: \url{https://mediatum.ub.tum.de/681651}.

\bibitem[KV12]{KorteVygenOptimisation}
Bernhard Korte and Jens Vygen.
\newblock {\em Combinatorial {{Optimization}}}, volume~21 of {\em Algorithms
  and {{Combinatorics}}}.
\newblock Springer Berlin Heidelberg, Berlin, Heidelberg, 2012.
\newblock \href {https://doi.org/10.1007/978-3-642-24488-9}
  {\path{doi:10.1007/978-3-642-24488-9}}.

\bibitem[Lam19]{LammichRefinement}
Peter Lammich.
\newblock Refinement to {{Imperative HOL}}.
\newblock {\em Journal of Automated Reasoning}, 62(4):481--503, 2019.
\newblock \href {https://doi.org/10.1007/s10817-017-9437-1}
  {\path{doi:10.1007/s10817-017-9437-1}}.

\bibitem[Lee05]{FordFulkersonMizar}
Gilbert Lee.
\newblock {Correctness of Ford-Fulkerson's Maximum Flow Algorithm}.
\newblock {\em Formalized Mathematics}, 13(2):305--314, 2005.

\bibitem[LL10]{collectionsFramework}
Peter Lammich and Andreas Lochbihler.
\newblock The {{Isabelle Collections Framework}}.
\newblock In {\em Interactive {{Theorem Proving}}}, Berlin, Heidelberg, 2010.
\newblock \href {https://doi.org/10.1007/978-3-642-14052-5_24}
  {\path{doi:10.1007/978-3-642-14052-5_24}}.

\bibitem[LS19]{LammichFlows}
Peter Lammich and S.~Reza Sefidgar.
\newblock Formalizing {{Network Flow Algorithms}}: {{A Refinement Approach}} in
  {{Isabelle}}/{{HOL}}.
\newblock {\em J. Autom. Reason.}, 62(2):261--280, 2019.
\newblock \href {https://doi.org/10.1007/s10817-017-9442-4}
  {\path{doi:10.1007/s10817-017-9442-4}}.

\bibitem[LT12]{lammichMonads}
Peter Lammich and Thomas Tuerk.
\newblock Applying {{Data Refinement}} for {{Monadic Programs}} to
  {{Hopcroft}}'s {{Algorithm}}.
\newblock In {\em Interactive {{Theorem Proving}}}, Berlin, Heidelberg, 2012.
\newblock \href {https://doi.org/10.1007/978-3-642-32347-8_12}
  {\path{doi:10.1007/978-3-642-32347-8_12}}.

\bibitem[Mar20]{maricIsomorphism}
Filip Mari{\'c}.
\newblock Verifying {{Farad\v zev-Read Type Isomorph-Free Exhaustive
  Generation}}.
\newblock In {\em Automated {{Reasoning}}}, Cham, 2020.
\newblock \href {https://doi.org/10.1007/978-3-030-51054-1_16}
  {\path{doi:10.1007/978-3-030-51054-1_16}}.

\bibitem[mat24]{mathlib}
Leanprover-community/mathlib4.
\newblock leanprover-community, November 2024.
\newblock URL: \url{https://github.com/leanprover-community/mathlib4}.

\bibitem[Moo59]{MooreShortestPath}
Edward~F. Moore.
\newblock {The Shortest Path through a Maze}.
\newblock In {\em Proc. Int. Symp. Switching Theory 1957, Part II}, pages
  285--292, 1959.

\bibitem[Nip15]{AmortizedComplexity}
Tobias Nipkow.
\newblock Amortized {{Complexity Verified}}.
\newblock In {\em The 6th {{International Conference}} on {{Interactive Theorem
  Proving}} ({{ITP}})}, 2015.
\newblock \href {https://doi.org/10.1007/978-3-319-22102-1_21}
  {\path{doi:10.1007/978-3-319-22102-1_21}}.

\bibitem[Nip24]{galeshapleyIsabelle}
Tobias Nipkow.
\newblock Gale-{{Shapley Verified}}.
\newblock {\em Journal of Automated Reasoning}, 68(2):12, 2024.
\newblock \href {https://doi.org/10.1007/s10817-024-09700-x}
  {\path{doi:10.1007/s10817-024-09700-x}}.

\bibitem[Nip25a]{ADTIsabelle}
Tobias Nipkow.
\newblock {Abstract Data Types}.
\newblock In {\em {Functional Data Structures and Algorithms: A Proof Assistant
  Approach}}, page 79–86. Association for Computing Machinery, New York, NY,
  USA, 1 edition, 2025.
\newblock \href {https://doi.org/10.1145/3731369.3731376}
  {\path{doi:10.1145/3731369.3731376}}.

\bibitem[Nip25b]{nipkowRunningTime}
Tobias Nipkow.
\newblock Basics.
\newblock In {\em Functional {{Data Structures}} and {{Algorithms}}: {{A Proof
  Assistant Approach}}}, volume~67, pages 1--10. Association for Computing
  Machinery, New York, NY, USA, 1 edition, September 2025.
\newblock \href {https://doi.org/10.1145/3731369.3731371}
  {\path{doi:10.1145/3731369.3731371}}.

\bibitem[Nip25c]{RBTIsabelle}
Tobias Nipkow.
\newblock Red--{{Black Trees}}.
\newblock In {\em Functional {{Data Structures}} and {{Algorithms}}: {{A Proof
  Assistant Approach}}}, volume~67, pages 97--106. Association for Computing
  Machinery, New York, NY, USA, 1 edition, September 2025.
\newblock \href {https://doi.org/10.1145/3731369.3731378}
  {\path{doi:10.1145/3731369.3731378}}.

\bibitem[Nos15]{noschinskiGraphLib}
Lars Noschinski.
\newblock A {{Graph Library}} for {{Isabelle}}.
\newblock {\em Math. Comput. Sci.}, 9(1):23--39, 2015.
\newblock \href {https://doi.org/10.1007/S11786-014-0183-Z}
  {\path{doi:10.1007/S11786-014-0183-Z}}.

\bibitem[OA92]{OrlinScalingAssignment}
James~B. Orlin and Ravindra~K. Ahuja.
\newblock New scaling algorithms for the assignment and minimum mean cycle
  problems.
\newblock {\em Math. Program.}, 54:41--56, 1992.
\newblock \href {https://doi.org/10.1007/BF01586040}
  {\path{doi:10.1007/BF01586040}}.

\bibitem[Ord56]{OrdenTranshipment}
Alex Orden.
\newblock {The Transhipment Problem}.
\newblock {\em {Manag. Sci.}}, 2(3):276--285, 1956.
\newblock \href {https://doi.org/10.1287/mnsc.2.3.276}
  {\path{doi:10.1287/mnsc.2.3.276}}.

\bibitem[Orl93]{OrlinScalingFlow}
James~B. Orlin.
\newblock A {{Faster Strongly Polynomial Minimum Cost Flow Algorithm}}.
\newblock {\em Operations Research}, 41(2):338--350, April 1993.
\newblock \href {https://doi.org/10.1287/opre.41.2.338}
  {\path{doi:10.1287/opre.41.2.338}}.

\bibitem[PT90]{PlotkinTardosDualNetworkSimplex}
Serge~A. Plotkin and \'{E}va Tardos.
\newblock {Improved Dual Network Simplex}.
\newblock In {\em Proc. {{First Annu}}. {{ACM-SIAM Symp}}. {{Discrete
  Algorithms SODA}}}, page 367–376. {SIAM}, 1990.

\bibitem[Sch03]{schrijverBook}
A.~Schrijver.
\newblock {\em Combinatorial {{Optimization}}: {{Polyhedra}} and
  {{Efficiency}}}.
\newblock Number~24 in Algorithms and Combinatorics. Springer, Berlin ; New
  York, 2003.

\bibitem[Vaz24]{MVproofthree}
Vijay~V. Vazirani.
\newblock A {{Theory}} of {{Alternating Paths}} and {{Blossoms}} from the
  {{Perspective}} of {{Minimum Length}}.
\newblock {\em Math. Oper. Res.}, 49(3):2009--2047, 2024.
\newblock \href {https://doi.org/10.1287/MOOR.2020.0388}
  {\path{doi:10.1287/MOOR.2020.0388}}.

\bibitem[vMN23]{hpcpl}
Floris {van Doorn}, Patrick Massot, and Oliver Nash.
\newblock Formalising the h-{{Principle}} and {{Sphere Eversion}}.
\newblock In {\em The 12th {{ACM SIGPLAN International Conference}} on
  {{Certified Programs}} and {{Proofs}} ({{CPP}})}, New York, NY, USA, January
  2023.
\newblock \href {https://doi.org/10.1145/3573105.3575688}
  {\path{doi:10.1145/3573105.3575688}}.

\bibitem[Vyg02]{VygenDualMinCost}
Jens Vygen.
\newblock {On Dual Minimum Cost Flow Algorithms}.
\newblock {\em {Mathematical Methods of Operations Research}}, 56(1):101--126,
  2002.
\newblock \href {https://doi.org/10.1007/s001860200202}
  {\path{doi:10.1007/s001860200202}}.

\bibitem[Vyg24]{vygenemailproof}
Jens Vygen.
\newblock Communication citation, March 2024.

\bibitem[WHN18]{wimmerMonadDP}
Simon Wimmer, Shuwei Hu, and Tobias Nipkow.
\newblock Verified {{Memoization}} and {{Dynamic Programming}}.
\newblock In {\em Interactive {{Theorem Proving}}}, Cham, 2018.
\newblock \href {https://doi.org/10.1007/978-3-319-94821-8_34}
  {\path{doi:10.1007/978-3-319-94821-8_34}}.

\bibitem[Wir71]{refinementWirth}
Niklaus Wirth.
\newblock Program {{Development}} by {{Stepwise Refinement}}.
\newblock {\em Commun. ACM}, 14(4):221--227, 1971.
\newblock \href {https://doi.org/10.1145/362575.362577}
  {\path{doi:10.1145/362575.362577}}.

\end{thebibliography}

\end{document}